\documentclass[preprint]{imsart}

\usepackage[margin=1.5in]{geometry}
\RequirePackage{amsthm,amsmath,amsfonts,amssymb}
\RequirePackage[numbers,sort&compress]{natbib}
\RequirePackage[colorlinks,citecolor=blue,urlcolor=blue]{hyperref}
\RequirePackage{graphicx}
\usepackage{subcaption} 
\usepackage{xcolor}
\graphicspath{{figs/}}
\usepackage{algorithm}
\usepackage{booktabs}
\usepackage[noend]{algpseudocode}
\usepackage{footnote}
\usepackage{cleveref}
\usepackage{verbatim}
\usepackage{epstopdf}
\usepackage{amsbsy, amscd}

\usepackage{xr}
\makeatletter
\newcommand*{\addFileDependency}[1]{
	\typeout{(#1)}
	\@addtofilelist{#1}
	\IfFileExists{#1}{}{\typeout{No file #1.}}
}
\makeatother

\newcommand*{\myexternaldocument}[1]{%
	\externaldocument{#1}%
	\addFileDependency{#1.tex}%
	\addFileDependency{#1.aux}%
}

\myexternaldocument{supplement}

\startlocaldefs
\theoremstyle{plain}

\newtheorem{theorem}{Theorem}[section]
\newtheorem{lemma}[theorem]{Lemma}
\newtheorem{othertheorem}{othertheorem}[section]

\theoremstyle{remark}
\newtheorem{definition}[theorem]{Definition}

\newtheorem*{fact}{Fact}
\newtheorem{rem}[othertheorem]{Remark}

\DeclareMathOperator*{\plim}{plim}
\DeclareMathOperator*{\PP}{\mathbb{P}}
\newcommand{\E}{\mathbb{E}}

\newcommand{\R}{\mathbb{R}}

\newcommand{\abs}[1]{\lvert#1\rvert}

\newcommand{\bm}{\mathbf}

\newcommand{\bet}{\boldsymbol{\beta}}

\newcommand{\p}{p}

\newcommand{\n}{n}
\newcommand{\A}{\mathrm{A}}
\newcommand{\Aeff}{\mathrm{A}_\textup{eff}}

\newcommand{\Z}{\mathbf{Z}}
\newcommand{\X}{\mathbf{X}}
\newcommand{\w}{\mathbf{w}}
\newcommand{\y}{\mathbf{y}}

\newcommand{\blam}{\boldsymbol \lambda}
\newcommand{\thet}{\boldsymbol \theta}

\newcommand{\bfalph}{\boldsymbol \alpha}

\newcommand{\TPP}{\textup{TPP}^\infty}
\newcommand{\FDP}{\textup{FDP}^\infty}

\renewcommand{\S}{\textup{CDF}}

\newcommand{\jk}[1]{{\color{orange}\textbf{Jason:} #1}}

\endlocaldefs

\begin{document}
	
	\begin{frontmatter}
		\title{Sharp Trade-Offs in High-Dimensional Inference via 2-Level SLOPE}
		
		\begin{aug}
			\author[A]{\fnms{Zhiqi}~\snm{Bu}\ead[label=e1]{woodyx218@gmail.com}},
			\author[B]{\fnms{Jason M.}~\snm{Klusowski}\ead[label=e2]{jason.klusowski@princeton.edu}},
			\author[C]{\fnms{Cynthia}~\snm{Rush}\ead[label=e3]{cynthia.rush@columbia.edu}},
			\and
			\author[D]{\fnms{Ruijia}~\snm{Wu}\ead[label=e4]{rjwu@sjtu.edu.cn}}
			\address[A]{\printead[presep={,\ }]{e1}}
			
			\address[B]{Department of Operations Research \& Financial Engineering, Princeton University\printead[presep={,\ }]{e2}}
			
			\address[C]{Department of Statistics, Columbia University\printead[presep={,\ }]{e3}}
			
			\address[D]{Antai College of Economics and Management, Shanghai Jiao Tong University\printead[presep={,\ }]{e4}}
		\end{aug}
		
		\begin{abstract}
			Among techniques for high-dimensional linear regression, Sorted L-One Penalized Estimation (SLOPE) \citep{bogdan2015slope} generalizes the LASSO via an adaptive $l_1$ regularization that applies heavier penalties to larger coefficients in the model. To achieve such adaptivity, SLOPE requires the specification of a complex hierarchy of penalties, i.e., a monotone penalty sequence in $\R^p$, in contrast to a single penalty scalar for LASSO. Tuning this sequence when $ p $ is large poses a challenge, as brute force search over a grid of values is computationally prohibitive. In this work, we study the \textbf{2-level SLOPE}, an important subclass of SLOPE, with only \emph{three} hyperparameters. We demonstrate both empirically and analytically that 2-level SLOPE not only preserves the advantages of general SLOPE—such as improved mean squared error and overcoming the Donoho-Tanner power limit—but also exhibits computational benefits by reducing the penalty hyperparameter space. In particular, we prove that 2-level SLOPE admits a sharp, theoretically tight characterization of the trade-off between true positive proportion (TPP) and false discovery proportion (FDP), contrasting with general SLOPE where only upper and lower bounds are known. Empirical evaluations further underscore the effectiveness of 2-level SLOPE in settings where predictors exhibit high correlation, when the noise is large, or when the underlying signal is not sparse. Our results suggest that 2-level SLOPE offers a robust, scalable alternative to both LASSO and general SLOPE, making it particularly suited for practical high-dimensional data analysis.
		\end{abstract}

	\end{frontmatter}
	


\section{Introduction}
The basic linear model is a widely applied regression model that uses  linear predictor functions to model the relationships  between the response $\y\in\R^\n$ and data $\X\in\R^{\n\times \p}$, with $n$ being the sample size and $p$ being the number of features, formulated as 
$$\y=\X\bet+\w.$$
Here $\bet\in\R^\p$ is the unknown target of inference and  $\w\in\R^\n$ is the noise. We assume the data is truly generated from such a model throughout this paper.

In this setting, a number of methods, such as the LASSO \citep{tibshirani1996regression}, have been developed to extract and analyze the signal $\bet$ from the dataset $(\y, \X)$ in high-dimensional settings, that is, when the sample size $n$ is smaller than or proportional to the feature size $p$. Such a high-dimensional model is often a reasonable framework for statistical machine learning \citep{zou2005regularization,su2016slope,bayati2011lasso}, signal processing \citep{candes2006stable,candes2005decoding}, genomic sequence analysis \citep{ogutu2012genomic}, and medical imaging \citep{lustig2008compressed}.

Methods for estimation in high-dimensional settings often employ an $l_1$ regularization or its variants to leverage the potential sparsity of signals and constrain the search space of possible $\bet$ vectors, in hopes that the corresponding estimator, $\widehat{\bet}$ will have better statistical properties, like a smaller estimation error or a higher true positive rate. In the LASSO, the $l_1$ regularization or the penalty, however, is not adaptive in the sense that all predictors $\widehat\beta_i$ are penalized at the same level. 

Recent works have discussed several drawbacks of adopting such a non-adaptive $l_1$ penalty. For example, when the signal is not sufficiently sparse or the sample size not sufficiently large, the LASSO can have a fundamental limit on its power, regardless of the choice of penalty level. To give a specific scenario, under the setting of \cite{su2017false,bu2021characterizing} where $\#\{j:\beta_j\neq 0\}/p=0.5$, the LASSO true positive proportion (TPP) is at most 36.69\%, asymptotically as $n$ and $p$ both grow with $n/p=0.3$  fixed, where the TPP is the proportion of non-zero coefficients in the estimate compared to the truth. This TPP upper bound holds irrespective of the choice of regularization parameter $\lambda$ in the LASSO and is known as the Donoho-Tanner power limit (DT limit). Moreover, the asymptotic estimation error, namely $\plim_p \frac{1}{p} \|\bet - \widehat{\bet}\|^2_2$, where $\plim$ is the probability limit, of LASSO can be significantly worse than that of more adaptive estimators, like the 
sorted $l_1$ penalized estimation (SLOPE) \citep{bu2021characterizing,chen2021asymptotic}. For instance, in the previous example 
with a Bernoulli prior on the elements of the unknown signal $\bet$, the optimal SLOPE can achieve an estimation error that is significantly (e.g.,\ 24\% relatively) lower than that is achievable by the optimal LASSO (c.f.,\ Figure 7 in \cite{bu2021characterizing}).

To be specific, SLOPE \citep{bogdan2015slope,su2016slope} is an adaptive regularization procedure that applies a sorted $l_1$ regularization to penalize the larger entries of the estimate more heavily, as follows:
\begin{align}
\begin{split}  \widehat{\boldsymbol{\beta}}(\blam)&\in\underset{\boldsymbol{b}\in\R^p}{\arg \min } \,\, C(\bm b;\blam) \quad \text{ where } \quad C(\bm b;\blam)
:= \frac{1}{2}\|\boldsymbol{y}-\boldsymbol{X} \boldsymbol{b}\|_{2}^{2}+\sum_{i=1}^{p} \lambda_{i}|b|_{(i)}.
\end{split}
\label{eq:SLOPE}
\end{align}
Here $\blam\in\R^\p$ is a penalty vector with $\lambda_{1} \geq \cdots \geq \lambda_{p}\geq 0$, 
and $\abs b_{(i)}$ denotes the order statistics of the absolute values of $\boldsymbol{b}$ such that $|b|_{(1)} \geq \cdots \geq |b|_{(p)}$. Notice that LASSO is a special case of SLOPE when all $\{\lambda_i\}$ take the same value. 

\subsection{Choices of SLOPE penalty}\vspace{-0.3cm}

For each penalty vector $\blam$, SLOPE in \eqref{eq:SLOPE} is a convex problem \citep{bogdan2015slope}  solvable by existing optimization algorithms like subgradient methods and proximal gradient descent. Furthermore, thanks to the flexibility of its penalty vector, SLOPE performs well in a number of ways according to the pre-defined choice of $\blam$. For instance, if we use the  critical values from the Benjamini-Hochberg procedure \citep{bogdan2015slope}, $\lambda_i=\Phi^{-1}(1-iq/2p)$\footnote{Here $\Phi$ is the cumulative distribution function of standard normal (i.e.,\ $\Phi^{-1}$ is the quantile). 
} for some $q\in(0,1)$, then SLOPE demonstrates two desirable features not available in other methods: (1) SLOPE provably controls the false discovery rate below $q$ under orthogonal design matrices \citep{bogdan2015slope}; (2) SLOPE achieves minimax estimation properties without knowledge about the true sparsity of the signal $\bet$ \citep{su2016slope}. Another example is to use $\lambda_i=\lambda_p+c(p-i)$ for some constant $c>0$, which can group correlated predictors when the data $\X$ is correlated \citep{bondell2008simultaneous,zeng2014decreasing,figueiredo2016ordered}, e.g., the uniform SLOPE where $\lambda_i=\lambda_1 (1-0.99(i-1)/p)$ in \cite{wang2022does}.

However, it is much more difficult to determine good choices of the parameter vector $\blam$ in the general setting. In other words, unless one makes strong assumptions on the data $\mathbf{X}$, it is rarely possible to pre-define the penalty sequence in a smart way. In practice, one may need to conduct hyperparameter tuning in a case-by-case manner for $\blam$. Unfortunately, for large-scale datasets like those often encountered in practice where $p$ can easily be in the order of thousands to millions, tuning a $p$-dimensional penalty vector $\blam$ can be extremely costly if even possible. These computational difficulties necessitate a fundamentally different design of $\blam$ with theoretical guarantees.

In this paper, we propose to consider the so-called `2-level SLOPE' which requires that $\blam$ takes only two non-negative values (clearly, $2\ll p$). Notably, the LASSO is a special instance of 2-level SLOPE, in the sense that LASSO is actually 1-level SLOPE; however, we remark that
2-level SLOPE is the simplest form of a SLOPE penalty that still retains adaptivity.
In fact, as we will introduce in detail in later sections, the performance of 2-level SLOPE has been studied empirically (but neither theoretically nor optimally) in \cite{zhang2021efficient, bu2021characterizing}, where 2-level SLOPE has shown some improved performance compared to the LASSO.

\begin{figure}[!htb]
\centering
\vspace{-0.5cm}
\includegraphics[width=0.46\linewidth]{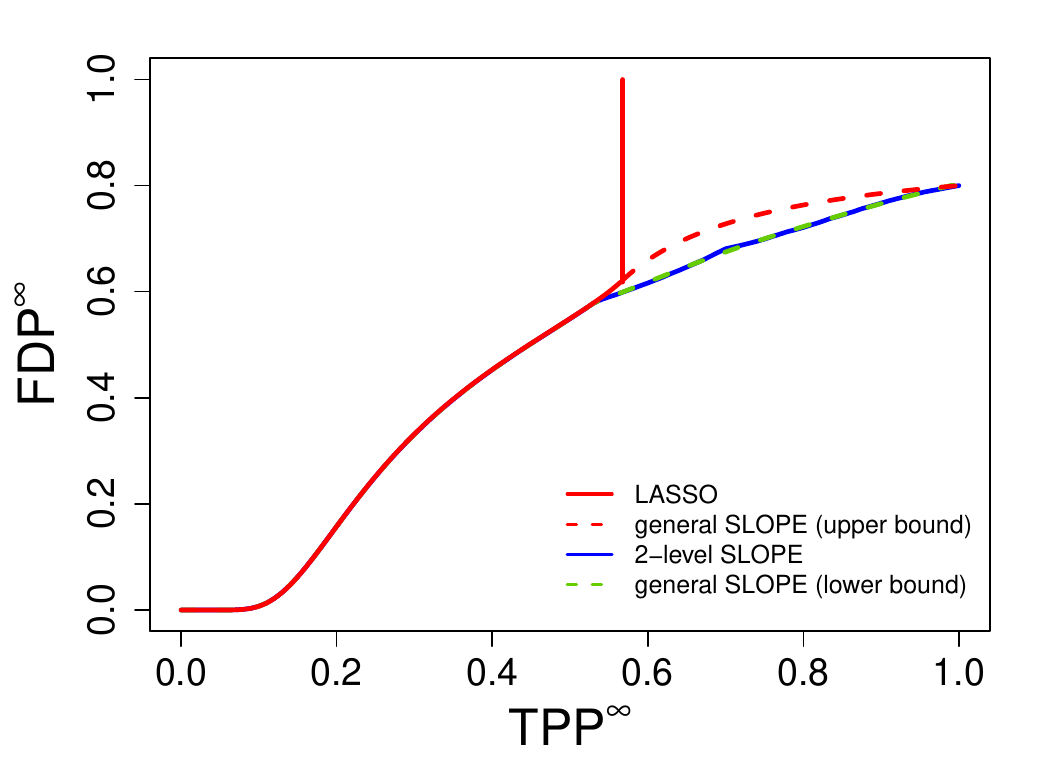}
\includegraphics[width=0.46\linewidth]{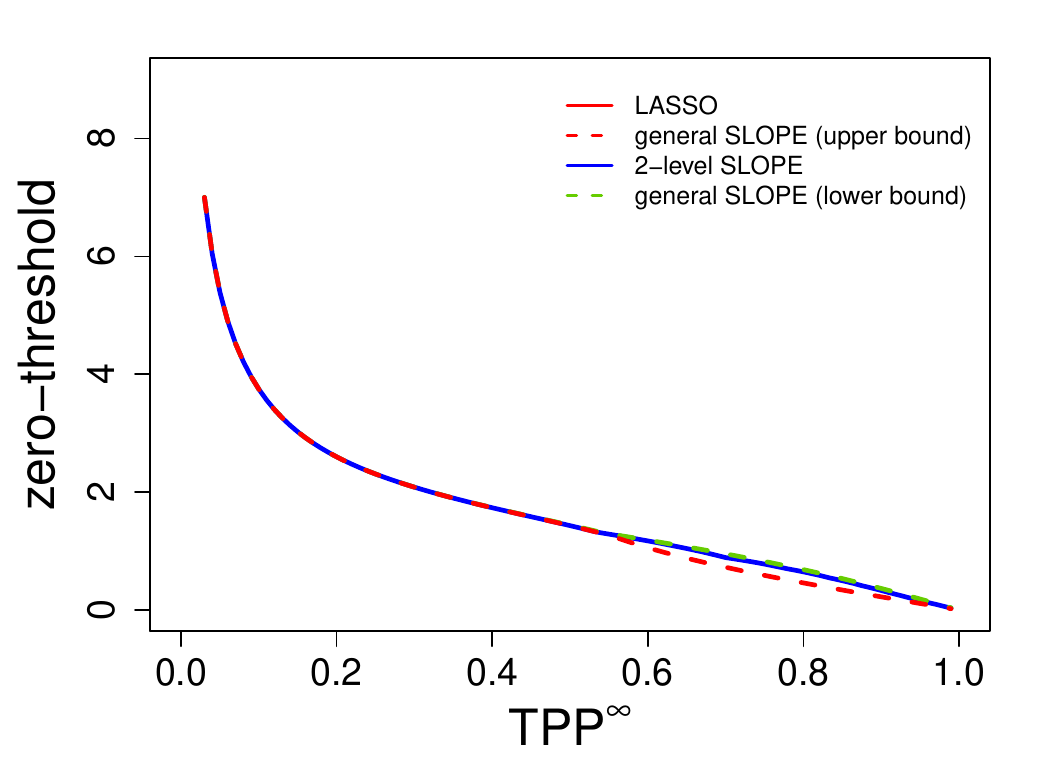}
\caption{The left plot shows the $\TPP$-$\FDP$ trade-off for LASSO (red/solid), general SLOPE (upper bound in red/dashed and lower bound in green/dashed), and 2-level SLOPE (blue/solid). The right plot shows the zero-threshold (see \Cref{def: zero threshold}) for LASSO, general SLOPE, and 2-level SLOPE. Here we have used $\epsilon=0.2,\delta=0.3$ under the assumptions in \Cref{sec:AMP}.}
    \label{fig:summary}
\end{figure}

On the one hand, we study the model selection performance of the general SLOPE, the 2-level SLOPE and the LASSO in terms of the false discovery proportion (FDP) and true positive proportion (TPP). While it is desirable for a model to have a small FDP and a large TPP, we demonstrate in \Cref{fig:summary} that all models show a trade-off between TPP and FDP. In other words, there is an asymptotic lower bound of FDP for each model at any level of TPP, regardless of the choice of prior and penalty.

Interestingly, both the general SLOPE and 2-level SLOPE can overcome the DT limit of LASSO at $\TPP\approx 0.5676$. However, the general SLOPE's trade-off curve beyond a certain $\TPP$ is only bounded (below and above) but not precisely characterized. In contrast, the trade-off for 2-level SLOPE is not only tight (to be shown in \Cref{thm:trade-off}), but also closely matches the minimum $\FDP$ that is given by the lower bound of the trade-off for the general SLOPE (see also \Cref{fig:trade-off curves} for various settings).

\begin{figure}[!htp]
\centering
\includegraphics[width=0.45\linewidth]{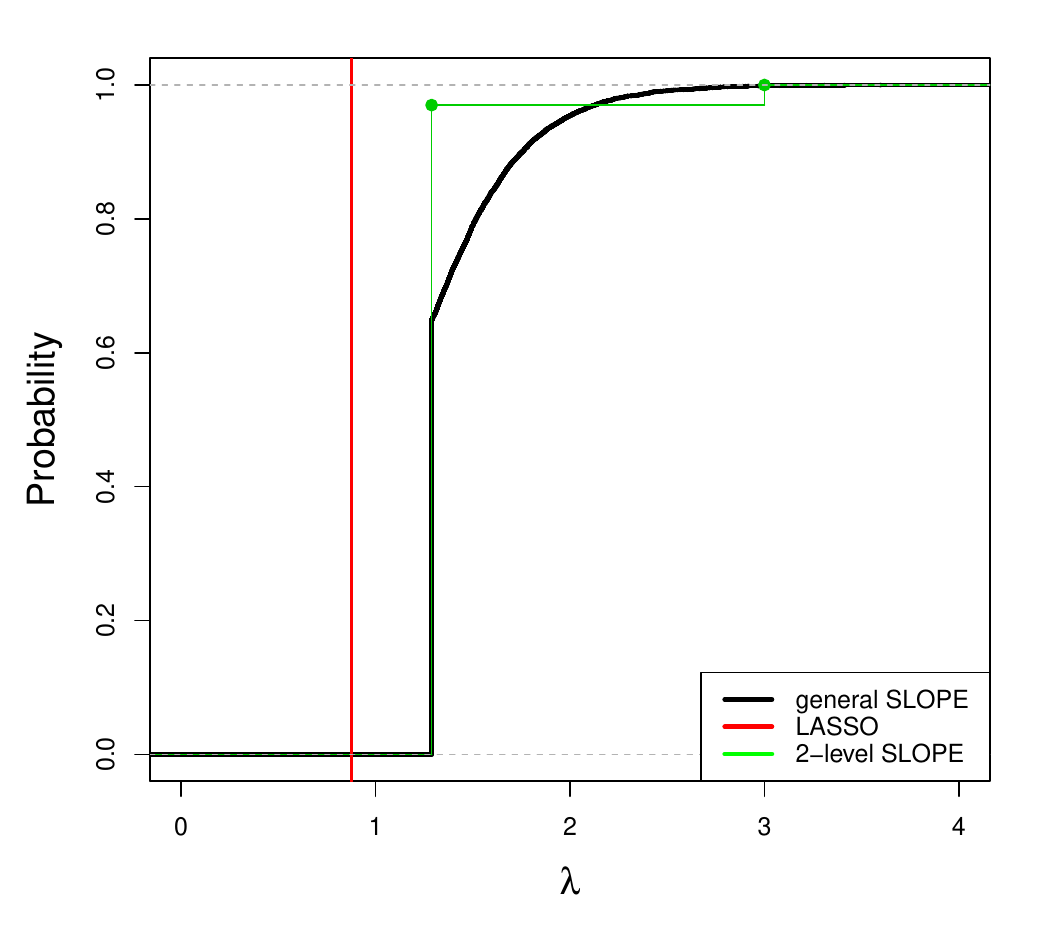}
\vspace{-0.5cm}
    \caption{Here we plot the cumulative distribution function (CDF) of the general SLOPE penalty $\blam$ that minimizes the $\FDP$, as well as the optimal 2-level SLOPE penalty CDF in green and the LASSO penalty CDF in red. In this experiment, $\X\in\R^{n\times p}$ is i.i.d.\ $\mathcal{N}(0,1/n)$ with dimensions $n=300$ and $p=1000$ using $\w\sim\mathcal{N}(\mathbf{0},\mathbf{I}_n)$ and the elements of the true prior $\bet$ are assumed to be i.i.d.\ $\text{Bernoulli}(0.7)*4$. We restrict to $\TPP=0.55$.
 For the LASSO, $\lambda\approx 0.878$; for the optimal 2-level SLOPE, $\blam\approx\langle 3.0,1.3;0.03\rangle$.
}
\label{fig:optimal SLOPE}
\end{figure}

Additionally, we study the FDP and TPP under fixed priors in \Cref{sec:6}. We plot the cumulative distribution functions of $\Lambda$ for the optimal SLOPE that minimizes the $\FDP$ in \Cref{fig:optimal SLOPE}, from which the coefficients of a SLOPE penalty vector $\blam$ are drawn. Specifically, the black curve represents the optimal general SLOPE penalty, which takes an explicit form only because of the simple Bernoulli prior. We emphasize that in general the optimal general SLOPE penalty is implicit and can only be optimized numerically (c.f., Section 4.3 in \cite{bu2021characterizing}). The green curve represents the optimal 2-level SLOPE penalty (given by \Cref{alg:fixed prior simpli}) and the red curve represents the optimal LASSO penalty (given by a grid search).

On the other hand, we also study the optimal SLOPE penalty in terms of minimizing the estimation error $\text{MSE}(\widehat{\bet})=\frac{1}{p} \|\bet - \widehat{\bet}\|^2_2$. In \Cref{sec:6}, we show that 2-level SLOPE has significantly lower MSE than the LASSO, under the constraint that we also require 2-level SLOPE to have lower FDP at the same TPP. Furthermore, in \Cref{sec:error fixed priors}, we remove such constraint and minimize the MSE as a stand-alone problem. Through extensive experiments on independent or correlated $\X$ with Gaussian or heavy-tailed distribution, small or large noises, uniform (non-tied) or Bernoulli (tied) priors, and multiple choices of SLOPE penalty, we observe that 2-level SLOPE achieves consistently smaller or equal MSE compared to the LASSO.

\subsection{Contributions and organization}
Our main contributions are listed as follows and summarized in \Cref{tab:summary}. 
\begin{enumerate}
    \item We formally propose the 2-level SLOPE as a convex minimization problem that is practically applicable to high-dimensional datasets. Specially, we show that 2-level SLOPE has an explicit  proximal operator in the limit, unlike general SLOPE, which makes it easy to provide the exact statistical characterization of 2-level SLOPE including, for example, the asymptotic TPP, FDP, and estimation error (i.e.,\ the mean squared error or MSE), whereas the TPP-FDP trade-off for the general SLOPE can only be approximated in \cite{bu2021characterizing}.
    \item We study the asymptotic TPP-FDP trade-off of the 2-level SLOPE under the approximate message passing (AMP) regime, and derive its tight characterization. We show that 2-level SLOPE captures the essence of the general SLOPE, e.g., breaking the Donoho–Tanner limit that upper bounds the LASSO TPP and achieving 100\% TPP,
    and hence not much additional benefit is gained from using the fully general SLOPE.
    \item We show that, when the  prior distribution on the signal $\bet$ is fixed, the 2-level SLOPE (with well chosen penalties) significantly reduces the estimation error and improves the TPP-FDP tradedoff compared to the LASSO. Moreover, we develop the algorithms to find these penalties efficiently in practice, including scenarios beyond the AMP regime.
\end{enumerate}

\begin{table}[!htb]
    \centering
    \begin{tabular}{c|c|c|c|c}
     \hline
     & limiting scalar function & adaptive penalty& overcoming DT limit& \# hyper-parameters \\\hline
    \textbf{LASSO}     &explicit (see $\eta_\text{soft}$)& no&no&1\\
    \textbf{2-level SLOPE}&explicit (see \eqref{eq:2-level limiting})& yes&yes&3\\
    \textbf{general SLOPE}     &implicit& yes&yes&$\p$\\
    \hline
    \end{tabular}
    \caption{Summary of proximal operators in the SLOPE family and LASSO.}
    \label{tab:summary}
\end{table}

The paper is organized as follows. In \Cref{sec:k-level}, we introduce the adaptive 2-level SLOPE. In \Cref{sec:AMP}, we lay out the AMP regime where our theoretical analysis is carried out. In \Cref{sec:limiting scalar function}, we derive the explicit 2-level SLOPE proximal operator, which is leveraged to derive the theoretical results in later sections. Specifically, we seek the optimal 2-level penalty (in terms of TPP-FDP trade-off) across all priors in \Cref{sec:tpp fdp all priors}; we seek the optimal 2-level penalty (in terms of TPP-FDP trade-off and estimation error) for any fixed prior in \Cref{sec:6} and \Cref{sec:error fixed priors}, respectively.

\section{2-level sorted $\ell_1$ penalized estimation}
\label{sec:k-level}

To avoid optimizing in the high dimensions of $\blam$ space, as is needed for the general SLOPE, we propose to restrict ourselves to a subclass of the possible penalty vectors $\blam$ that were introduced in \eqref{eq:SLOPE}. In particular, we consider a subclass such that the $p$ elements of $\blam$ take only two non-negative values, denoted by
\begin{align}
\blam&=\underbrace{(\lambda_1,\cdots,\lambda_1,}_{\text{about } s p \text{ of } \lambda_1}\underbrace{\lambda_2,\cdots,\lambda_2)}_{\text{about } (1-s)p \text{ of } \lambda_2}:=\langle\lambda_1,\lambda_2;s\rangle \, .
\label{eq:K-level SLOPE seq}
\end{align}
Here $\lambda_1\geq\lambda_2\geq 0$ (with at least one value being non-zero) denote the \textit{penalty levels} and $s\in (0,1)$ represents the \textit{splitting ratio}, at which the penalty changes from the larger level $\lambda_1$ to the lower level $\lambda_2$. Note that when $\lambda_1=\lambda_2$, this reduces to the LASSO penalty. We name such a restricted SLOPE formulation as the \textit{2-level SLOPE} and notice that in such settings, the penalty $\blam$ has three degrees of freedom. Similarly, for any integer $1\leq K\leq p$, we define $K$-level SLOPE as that having at most $K$ distinct penalty levels $\lambda_1\geq\cdots\geq\lambda_K\geq 0$, with ratios $0< s_1,\cdots,s_{K-1}<1$ such that $0 < \sum_k s_k<1$.

Considering \eqref{eq:SLOPE}, we can view the general SLOPE as the $p$-level SLOPE and the LASSO as a special subclass of 2-level SLOPE, namely the 1-level SLOPE having only a single $\lambda_1 > 0$, which is the only non-adaptive penalty sequence in all possible 2-level SLOPE penalties. In other words, for $p > 2$,
\begin{center}
\fbox{LASSO (1-level)$\ \subset$ 2-level SLOPE$\ \subset$ SLOPE ($p$-level).}
\end{center}

\subsection{2-level SLOPE optimization program}
Mathematically speaking, 2-level SLOPE gives rise to a simplified version of  \eqref{eq:SLOPE}, that is
\begin{align}
\begin{split}
  \widehat{\boldsymbol{\beta}}(\lambda_1,\lambda_2,s)
  & = \underset{\bm b\in\R^\p}{\arg\min }\,\, C(\bm b;\lambda_1,\lambda_2,s)\,,
    \\
\text{where } C(\bm b;\lambda_1,\lambda_2,s)  &\equiv \frac{1}{2}\|\boldsymbol{y}-\boldsymbol{X} \boldsymbol{b}\|_{2}^{2}+\lambda_{1}\sum_{i/p< s} |b|_{(i)}+\lambda_{2}\sum_{i/p\geq s} |b|_{(i)}\,.
 \end{split}
  \label{eq:K-level}
\end{align}
Although prior to this work, its formulation has not been explicitly described and its theoretical analysis is lacking, 2-level SLOPE is well motivated via its application in previous works. For example, in \cite{bu2021characterizing}, a 2-level SLOPE is constructed to overcome the DT limit asymptotically. That is, under some conditions (indeed, the same as are considered here: see \Cref{sec:AMP}), the LASSO has an upper limit on the TPP regardless of the choice of the penalty coefficient; however, there exists a 2-level SLOPE that can break this limit and achieve 100\% power as shown in \cite[Proposition 2.3]{bu2021characterizing}. As another example, under the same conditions, it is shown both empirically \citep{zhang2021efficient} and theoretically \citep{bu2021characterizing} that 2-level SLOPE sometimes strictly outperforms the LASSO in terms of the estimation error, up to an improvement of 50\% in some cases. Specifically, 2-level SLOPE can outperform LASSO, uniform SLOPE (SLOPE with a specific penalty sequence; see Section \ref{sec:error fixed priors}) and ridge regression in terms of estimation error (see Figure 4 of \cite{wang2022does}, also extended in \Cref{fig:arian}) in some settings including when the prior $\bet$ has tied elements, where having tied elements means there is a non-zero probability that multiple elements of $\bet$ take the same value, or equivalently that the elements of $\bet$ have a repeated magnitude. 

Because the 2-level SLOPE is a subclass of the general SLOPE, it is indeed a convex minimization problem that is readily solvable by the isotonic regression and proximal gradient descent using the  \texttt{SLOPE} R-package, according to \cite[Proposition 1.2]{bogdan2015slope}. 


\subsection{Proximal operator}\vspace{-0.3cm}

To solve the general SLOPE problem in \eqref{eq:SLOPE}, including its 2-level sub-problem \eqref{eq:K-level}, one needs to iteratively apply the proximal gradient descent by \cite{bogdan2015slope}: for arbitrarily initialization, iterate
\begin{align}
\bet^{k+1}=\eta_\text{SLOPE}(\bet^k
+ t_k \X^\top(\y-\X\bet^k);t_k \blam) \, ,
\label{eq:fista}
\end{align}
where $k$ is the index of iteration, $\bet^k$ converges to the true minimizer $\widehat\bet$ in \eqref{eq:K-level}, and $t_k$ is the step size. The proximal operator of SLOPE is defined through
\begin{align}
  \eta_\text{SLOPE}(\boldsymbol{v};\boldsymbol{\theta}):=\underset{\boldsymbol{b}}{\arg\min }\,\, \frac{1}{2}\|\boldsymbol{v}-\boldsymbol{b}\|^{2}+\sum_{i=1}^p \theta_i|b|_{(i)} \, .
  \label{eq:SLOPE prox}
\end{align}
For the LASSO, the penalty vector $\thet$ can be replaced by a penalty scalar $\theta > 0$, and the proximal operator is known as the soft-thresholding function $\eta_\text{soft}(\boldsymbol{v};\theta)=\text{sign}(\boldsymbol{v})(|\boldsymbol{v}|-\theta)_+$.
However, for general $\boldsymbol{\theta}\in\R^\p$, the SLOPE proximal operator has no explicit form and at each iteration \eqref{eq:SLOPE prox} has to be solved by the pool adjacent violators algorithm (PAVA) in \cite{brunk1972statistical,kruskal1964nonmetric,bogdan2015slope}. Nevertheless, it is well known that SLOPE often produces estimators with shared values, meaning that the $p$ values in the estimator may share a number of magnitudes. For example, $\eta_\text{SLOPE}([5,-4,0.5];[4,1,0.7])=[2,-2,0]$ has a shared magnitude at 2. 

Throughout the paper, we define the set of shared magnitudes of a vector $\bm w$ as $\{\bm w\}_\text{shared}:=\{h:|\{i:|w_i|=h\}|>1\}$, and use the $l_0$ norm to count the non-zero shared magnitudes. For instance, $\{(3,3,-1,0,2,1,0)\}_\text{shared}:=\{3,1,0\}$ and $\|\{3,1,0\}\|_0=2$. Lemma~\ref{lem:1 flat} below, proved in \Cref{app:lem1proof}, shows that for 2-level SLOPE, the proximal operator shares at most one non-zero magnitude (the meaning is clarified in the lemma statement, by taking $K=2$). 

\begin{lemma}\label{lem:1 flat}
For the $K$-level SLOPE problem, the components of the minimizer $\widehat\bet$ share at most $K-1$ non-zero magnitudes, meaning $\|\{\widehat\bet\}_\textup{shared}\|_0\leq K-1$.
\end{lemma}

\section{Approximate message passing regime}
\label{sec:AMP}
To prepare for our asymptotic study of the 2-level SLOPE performance, throughout the paper we work under the following assumptions that are conventional in the literature of approximate message passing (AMP) \citep{donoho2009message,bayati2011dynamics,bu2019algorithmic} and convex Gaussian min-max theorem (CGMT) \citep{hu2019asymptotics}.

\begin{enumerate}
    \item[(A1)] The data matrix $\X$ has independent and identically-distributed (i.i.d.) elements that follow a Gaussian distribution with mean 0 and variance $1/n$.
    
    \item[(A2)] The signal $\boldsymbol{\bet}$ has elements that are i.i.d.\ with distribution $\Pi$, where $\PP(\Pi\neq 0)=\epsilon$ for some $0<\epsilon<1$ and $\mathbb{E}(\Pi^{2} \max \{0, \log \Pi\})<\infty$.
    
    \item[(A3)] The noise $\bm w$ has elements that are i.i.d.\ with distribution $W$ and $\sigma^{2}:=\mathbb{E}(W^{2})<\infty$.
    
    \item[(A4)] The penalty vector $\boldsymbol{\lambda}(p)=\left(\lambda_{1}, \ldots, \lambda_{p}\right)$ contains the order statistics of $p$ i.i.d.\ realizations of a non-negative random variable  $\Lambda$, with $\mathbb{E}(\Lambda^{2})<\infty$.
    
    \item[(A5)] The ratio $n/p$ reaches a scaling constant $\delta \in (0,\infty)$ in the so-called `large system limit' where $n$ and $p \rightarrow \infty$.
\end{enumerate} 
The assumptions above enable the tight characterization of the asymptotic properties of the SLOPE estimator, without which it is difficult to analyze the estimation error as the SLOPE (including the LASSO) has no closed form solution in general. 

We note that the assumption $\PP(\Pi\neq 0)=\PP(\beta_j\neq 0)=\epsilon$ in (A2) is not required by the AMP theory, but is required to properly define the TPP and FDP as follows.

\begin{align}
\operatorname{TPP}(\boldsymbol{\beta}, \blam) & =\frac{\left|\left\{j:|\widehat{\beta}(\blam)_j|>0, \, \beta_j \neq 0\right\}\right|}{\left|\left\{j: \beta_j \neq 0\right\}\right|} \,, \quad 
\operatorname{FDP}(\boldsymbol{\beta}, \blam) =\frac{\left|\left\{j:|\widehat{\beta}(\blam)_j|>0, \, \beta_j=0\right\}\right|}{\left|\left\{j:|\widehat{\beta}(\blam)_j|>0\right\}\right|} \,.
\end{align}
\begin{rem}
It is clear that the 2-level SLOPE satisfies assumption (A4) as $\mathbb{E}(\Lambda^2)=s\lambda_1^2+(1-s)\lambda_2^2<\infty$ with $s \in (0,1)$ and $0 \leq \lambda_2 \leq \lambda_1 < \infty$. Specifically, we can write
\begin{align}
\Lambda=\langle\lambda_1,\lambda_2;s\rangle=
\begin{cases}
\lambda_1 &\text{ w.p.\ } \,\,\, s \,,
\\
\lambda_2 &\text{ w.p.\ } \,\,\, 1-s \,.
\end{cases}
\label{eq:asymptotic lambda}
\end{align}

Note that unlike the general SLOPE, the number of penalty levels in 2-level SLOPE is fixed at two and does not go to infinity when $n,p\to\infty$, making it easier to tune and analyze. 
\end{rem}

\subsection{Approximate message passing theory}\vspace{-0.3cm}
Generally speaking, AMP \citep{donoho2009message,bayati2011dynamics} is a class of optimization algorithms that both converge to the true minimizer in \eqref{eq:SLOPE} and exactly characterize the minimizer's asymptotic distribution under certain problem settings. See \cite{feng2022unifying} for a tutorial on AMP. From the AMP theory for SLOPE (see \cite[Theorem 3]{bu2019algorithmic} and \cite[Theorem 1]{hu2019asymptotics}), the SLOPE estimator $\widehat\bet(\blam)$ can be informally viewed as 
\begin{align}
\widehat\bet(\blam)\approx \eta_{\Pi+\tau Z,\A\tau}\left(\boldsymbol{\beta}+\tau\Z\right)\to \eta_{\Pi+\tau Z,\A\tau}\left(\bm\Pi+\tau \Z\right)\,,
\label{eq:AMP result}
\end{align}
in which $\Z\in\R^\p$ is an i.i.d. standard Gaussian vector, $Z$ is a standard Gaussian $N(0,1)$, the random variable $\Pi$ is the signal prior defined in Assumption (A2) from which $\bm{\Pi}\in\R^\p$ is i.i.d. drawn, and $\eta_{\Pi+\tau Z,\A\tau}$ is the limiting scalar function \citep{hu2019asymptotics} with $\A$ and $\tau$ defined through the original penalty $\Lambda$ and prior $\Pi$ in the following paragraph via equations \eqref{eq:calibration} and \eqref{eq:state evolution}.
Specifically, the limiting scalar function $\eta$ is an element-wise mapping defined as follows. For sequences $\{\boldsymbol{\theta}(p)\}$ and $\{\boldsymbol{v}(p)\}$ growing in $p$, with empirical distributions that weakly converge to distributions $\Theta$ and $V$, respectively, there exists a limiting scalar function $\eta$ (determined by $\Theta$ and $V)$ such that,
\begin{align}
\frac{1}{p}\left\|\eta_\text{SLOPE}(\boldsymbol{v}(p) ; \boldsymbol{\theta}(p))-\eta_{V, \Theta}(\boldsymbol{v}(p))\right\|^2 \rightarrow 0\,  \text{ as $p \rightarrow \infty$}.
\label{eq:limiting scalar}
\end{align}
In light of \eqref{eq:AMP result}, we set $\thet=\bfalph\tau$ and $\bm v=\bet+\tau \Z$ where $\bfalph$ is the normalized SLOPE penalty to be introduced in \eqref{eq:alpha seq}. We note that the SLOPE proximal operator is implicit, but asymptotically separable, as the function applies itself coordinatewise to the vector $\bet+\tau \Z$. 

To determine $\A := \A(\Lambda)$ and $\tau := \tau(\Lambda)$  in \eqref{eq:AMP result},  where we will sometimes drop the dependence on $\Lambda$ to save space,  we introduce two key equations for the analysis, namely the calibration and the state evolution equations \citep{bu2019algorithmic, hu2019asymptotics}, for which $(\A(\Lambda),\tau(\Lambda))$ is the unique solution. 
We state these two equations below: given $(\delta, \Pi, \Lambda)$,
\begin{align}
\textbf{Calibration:}\quad  \Lambda&=\A \tau\left(1-\frac{1}{\delta} \, \mathbb{E}\left[\eta'_{\Pi+\tau Z,\A\tau}(\Pi+\tau Z)\right]\right)\,,
  \label{eq:calibration}
  \\
\textbf{State evolution:}\quad  \tau^{2}&=\sigma^{2}+\frac{1}{\delta} \, \mathbb{E}\left[\left(\eta_{\Pi+\tau Z,\A\tau}(\Pi+\tau Z)-\Pi\right)^{2}\right]\,.
\label{eq:state evolution}
\end{align}
In the above, $\Pi$, $\Lambda$, and $\delta$ are defined in Assumptions (A2), (A4), and (A5) and $\eta'$ is the derivative of the limiting scalar function $\eta$. On one hand, the calibration equation \eqref{eq:calibration} describes a mapping between the original penalty distribution $\Lambda$ and the normalized penalty distribution, denoted $\A$. This mapping is bijective by \cite[Proposition II.6]{bu2019algorithmic}, allowing us to work with $\Lambda$ or the normalized $\A$, exchangeably. Using the same notation as in \eqref{eq:K-level SLOPE seq}, we write the normalized 2-level SLOPE penalty for $\alpha_1 \geq \alpha_2 \geq 0$ (with at least one value being non-zero) as
\begin{align}
&\bfalph=\langle \alpha_1,\alpha_2;s\rangle
:=\underbrace{(\alpha_1,\cdots,\alpha_1,}_{\text{about }s p}\underbrace{\alpha_2,\cdots,\alpha_2)}_{\text{about }(1-s) p} \,,
\label{eq:alpha seq}
\end{align}
and its asymptotic distribution as
\begin{align}
\A=\langle \alpha_1,\alpha_2;s\rangle=
\begin{cases}
\alpha_1 &\text{ w.p.\ } \,\,\, s  \,,
\\
\alpha_2 &\text{ w.p.\ }  \, \, \, 1-s  \,.
\end{cases}
\label{eq:A two-level}
\end{align}
Notice that $\bfalph$ and $\blam$ have the same splitting ratio $s$ because they are parallel to each other according to \eqref{eq:calibration}.

On the other hand, we employ Theorem 1 of \cite{bu2019algorithmic} to show that the state evolution equation \eqref{eq:state evolution} uniquely defines $\tau$ through $(\Pi,\Lambda)$. Consequently, we define the normalized signal distribution and its non-zero component as
\begin{align}
\pi:=\Pi/\tau, \qquad \text{ and } \qquad  \pi^*:=(\Pi|\Pi\neq 0)/\tau=(\pi|\pi\neq 0).
\label{eq:define pi}
\end{align}
Now that we have established the equivalence between the original $(\Pi,\Lambda)$ and the normalized $(\pi,\A)$, we will sometimes work with $(\pi,\A)$ when it makes the analysis easier.

These notations are summarized in Table \ref{tab:notations} below.
\begin{table}[!htb]
    \centering
    \begin{tabular}{c|c|c}
    \hline
    finite dimension & asymptotic distribution & normalized distribution \\\hline
    $\bet$ & $\Pi$ &$\pi$\\
    $\blam$ & $\Lambda$ &$\A$ \\
    $\widehat\bet$ & $\eta_{\Pi+\tau Z,\A\tau}(\Pi+\tau Z)$ &$\eta_{\pi+Z,\A}(\pi+Z)$\\
    \hline
    \end{tabular}
    \caption{Notations of the prior and the SLOPE penalty in different regimes.}
    \label{tab:notations}
\end{table}

\subsection{TPP and FDP}
Under the characterization of the asymptotic SLOPE distribution given in \eqref{eq:AMP result}, we define $\TPP$ and $\FDP$ as the large system limits of the TPP and FDP. 
Following \cite[Lemma 3.1]{bu2021characterizing}, as $n,p\to\infty$ and under convergence in probability,
\begin{align}
\begin{split}
\TPP(\Pi, \Lambda) &:=\mathbb{P}(\eta_{\pi+Z, \A}(\pi^*+Z) \neq 0)  \,\,\, \text{ and }  \,\,\,\FDP(\Pi, \Lambda) :=\frac{(1-\epsilon) \mathbb{P}(\eta_{\pi+Z, \A}(Z) \neq 0)}{\mathbb{P}(\eta_{\pi+Z, \A}(\pi+Z) \neq 0)}.
\end{split}
\label{eq:tppfdp comp}
\end{align}
Recall that $\epsilon$ is defined via $\PP(\Pi\neq 0)=\epsilon$ in Assumption (A2) above.  Here we have used $\eta_{\Pi+\tau Z,\A\tau}=\tau\eta_{\pi+Z, \A}$ so as to work with the normalized prior $\pi$ and penalty $\A$.

In this paper, we are interested in the asymptotic \textbf{$\TPP$-$\FDP$ trade-off curve}, by which we mean the following. We characterize the trade-off curve for all pairs of problem parameters $(\epsilon, \delta)$ where $\epsilon$ is the signal sparsity and $\delta$ is the undersampling ratio from (A5). In more detail, the trade-off curve answers the following: for any given $0 \leq u \leq 1$ with $\TPP = u$, what is the infimum of achievable $\FDP$ values, where we say that a given $\FDP$ value is achievable if there exists some $\epsilon$-sparse signal prior $\Pi$ (having $\mathbb{P}(\Pi \neq 0) = \epsilon$) and some penalty distribution $\Lambda$ attaining this $\FDP$.

Importantly, the TPP and FDP values exhibit an asymptotic trade-off within the SLOPE family, i.e.,\ higher TPP comes at the price of higher FDP. To characterize the SLOPE trade-off, we will leverage a critical quantity from \cite[Definition 4.1]{bu2021characterizing}, the zero-threshold $\alpha(\Pi, \Lambda)$, which is defined below. 
\begin{definition}
\label{def: zero threshold}
Let $(\Pi, \Lambda)$ be a pair of prior and penalty distributions with normalized version $(\pi, \mathrm{A})$, where $\pi:=\Pi/\tau$ and $\A$ is defined via \eqref{eq:calibration} and \eqref{eq:state evolution}. The \emph{zero-threshold}, denoted $\alpha=\alpha(\Pi, \Lambda)\equiv\alpha(\pi,\A)$ is defined such that $\eta_{\pi+Z, \mathrm{A}}(x)\neq 0$ if and only if $|x|>\alpha$.
\end{definition}
We note that the zero-threshold always exists and is positive if the SLOPE solution contains zero values (see \cite[Proposition C.5]{bu2021characterizing}). Put differently, conditioning on the prior $\Pi$, there is an one-to-one mapping between the SLOPE penalty $\Lambda$ (or $\A$) and the zero-threshold $\alpha$. Specially, the zero-threshold is equivalent to the normalized penalty $\A$ for the LASSO.
Using the zero-threshold, the limiting values in \eqref{eq:tppfdp comp} can be simplified to
\begin{align}
&\TPP(\Pi, \Lambda)=\mathbb{P}\left(\left|\pi^*+Z\right|>\alpha(\Pi, \Lambda)\right)\,, 
\label{eq:tppfdp rewrite2}\\
&\FDP(\Pi, \Lambda)=\frac{2(1-\epsilon) \Phi(-\alpha(\Pi, \Lambda))}{2(1-\epsilon) \Phi(-\alpha(\Pi, \Lambda))+\epsilon \cdot \operatorname{TPP}^{\infty}(\Pi, \Lambda)} \,,
\label{eq:tppfdp rewrite}
\end{align}
by directly applying $\PP(\eta_{\pi+Z, \A}(x)\neq 0)=\PP(|x|>\alpha)$ and decomposing $\PP(\eta_{\pi+Z, \A}(\pi+Z)\neq 0)=(1-\epsilon)\PP(\eta_{\pi+Z, \A}(Z)\neq 0)+\epsilon\PP(\eta_{\pi+Z, \A}(\pi^*+Z)\neq 0)$, where the last term is $\TPP$.

Note from the equations above that for fixed $\TPP=u$, the formula of $\FDP$ in \eqref{eq:tppfdp rewrite} is decreasing in $\alpha$. Therefore, given any $\TPP=u$, finding the infimum of achievable $\FDP(\alpha)$ values is equivalent to finding the supremum of feasible zero-thresholds $\alpha$. 
We say that a zero-threshold is `feasible' if plugged into \eqref{eq:tppfdp rewrite2}, we find  $\TPP=u$ for some $(\Pi,\Lambda)$. As rigorously proven in \cite{su2017false,bu2021characterizing}, the supremum of feasible $\alpha$ exists and is not infinite, since in \eqref{eq:tppfdp rewrite2} $\TPP(\alpha)\to 0$ as $\alpha\to\infty$, which is visualized in the right sub-plot of \Cref{fig:summary}. With this in mind, we define some useful quantities.

\begin{enumerate}
    \item \textbf{All priors scenario: } For the general trade-off curve over all priors, we consider
\begin{align}
\alpha^*(u):=\sup_{\Pi, \Lambda} \, \alpha(\Pi, \Lambda) \quad \text{ s.t. } \quad \TPP(\Pi,\Lambda)=u\,.
\label{eq:all prior max zero thres}
\end{align}
\item \textbf{Fixed prior scenario: } For the trade-off curve over a fixed prior $\Pi$, we consider
\begin{align}
\alpha^*(u):=\sup_{\Lambda} \, \alpha(\Pi, \Lambda) \quad \text{ s.t. } \quad \TPP(\Pi,\Lambda)=u\,.
\label{eq:fixed prior max zero thres}    
\end{align}
\end{enumerate}
Either way, we derive the minimum $\FDP$ -- over $(\Pi,\Lambda)$ or only $\Lambda$ -- on the SLOPE trade-off
\begin{align}
\min \, \FDP(\Pi,\Lambda; \delta, \epsilon):=\frac{2(1-\epsilon) \Phi\left(-\alpha^*(u)\right)}{2(1-\epsilon) \Phi\left(-\alpha^*(u)\right)+\epsilon u} \quad \text{ s.t. } \quad\TPP(\Pi,\Lambda)=u\, .
\label{eq:fdp is function of tpp}
\end{align}

\section{Explicit 2-level SLOPE limiting scalar function}
\label{sec:limiting scalar function}
In fact, following \cite[Definition 4.2]{bu2021characterizing}, the limiting scalar function of any general SLOPE can also be written as an adaptive soft-thresholding function, $\eta_{\pi+Z,\A}(x)=\eta_\text{soft}(x;\widehat\A_\text{eff}(x))$, with an implicit penalty function $\widehat\A_\text{eff}$. We claim that the limiting scalar function of the 2-level SLOPE (including the LASSO), however, is explicit and in this section we derive its form.
\begin{align}
\begin{split}
\textbf{LASSO: } \,\,\eta_{\pi+Z,\A}(x)&=\eta_\text{soft}(x;\alpha) \,,
\\
\textbf{2-level SLOPE: } \,\, \eta_{\pi+Z,\A}(x)&=\eta_\text{soft}(x;\Aeff(x;\alpha_1,\alpha_2,s)) \,,
\end{split}
\label{eq:slope eff}
\end{align}
where $\Aeff$ is termed as the effective penalty and will be defined in \eqref{eq: 2level Aeff} below. 

\begin{theorem}
\label{thm:1 point mass}
Under Assumptions (A1)-(A5), the asymptotic distribution of the 2-level SLOPE solution in \eqref{eq:K-level}  is $\widehat\Pi:=\eta_{\Pi+\tau Z,\A\tau}\left(\Pi+\tau Z\right)=\tau\eta_{\pi+ Z,\A}\left(\pi+Z\right)$ where $(\tau,\A)$ depends on $(\Pi,\Lambda)$ via \eqref{eq:calibration} and \eqref{eq:state evolution}. Denoting $\A=\langle\alpha_1,\alpha_2;s\rangle$, we obtain
\begin{align}
\begin{split}
\textbf{2-level SLOPE}
\\
\textbf{limiting scalar function}
\end{split}
:\, \, \eta_{\pi+Z,\A}(x)=
\begin{cases}
\eta_\textup{soft}(x;\alpha_1) &\text{ if } \alpha_1+h < |x|,
\\
\textup{sign}(x) \max\{h,0\} &\text{ if } \alpha_2+h<|x|<\alpha_1+h, \\
\eta_\textup{soft}(x;\alpha_2) &\text{ if } |x|<\alpha_2+h,
\end{cases}
\label{eq:2-level limiting}
\end{align}
in which the value of $h$ is unique as derived in  \Cref{lem:q1q2h}.
\end{theorem}

In \Cref{thm:1 point mass}, as $n,p\to\infty$, we can derive that the asymptotic distribution of $\widehat\bet$, i.e.,\ $\widehat\Pi:=\eta_{\Pi+\tau Z,\A\tau}\left(\Pi+\tau Z\right)$, has at most one non-zero point mass at $h$. Here, we define a point mass $x$ for a distribution $X$ as any $x$ such that $\PP(X=x)>0$, e.g.,\ the standard normal distribution has no point masses, whereas a uniform distribution on $\{0, 3\}$ has a non-zero point mass at 3. Depending on whether $h$ is strictly positive, we know if there is one non-zero point mass or not (i.e.,\ whether the 2-level SLOPE reduces to LASSO).

At high level, it is important to leverage that the 2-level SLOPE solution $\widehat\bet$ has at most one \textit{non-zero shared magnitude} by \Cref{lem:1 flat}. To visualize the derivation of 2-level SLOPE proximal operator $\eta_\text{SLOPE}(\bet+\tau\Z;\bfalph\tau)$ in \eqref{eq:limiting scalar}, we refer the readers to \Cref{fig:normalized PAVA} where we demonstrate the PAVA 
step-by-step (restated in Algorithm 1 in the Supplement, which iteratively solves the proximal problem in \eqref{eq:SLOPE prox} for finite dimension $p$).
We plot the quantile function (y-axis) of each resulting distribution against the probability (x-axis), so as to approximate the asymptotic distribution $\widehat\Pi$. We present two cases in two columns, and for both cases $\pi=\text{Bernoulli}(0.5)$. In the left column $\A=\langle 2,1,0.15\rangle$ and $\widehat\Pi$ has one non-zero point mass; in the right column $\A=\langle 2,1,0.35\rangle$ and $\widehat\Pi$ has no non-zero point mass.

Focusing on the left column, in the first row, we visualize the sorted vector of $|\bet+\tau\Z|$ in the black solid curve, and the normalized 2-level penalty vector $\bfalph\tau=\langle \alpha_1,\alpha_2;0.15\rangle\cdot\tau$ in the red dashed curve. In the second row, we calculate the element-wise difference of the two vectors, which includes a vertical drop at the $0.85p$-th element, 
resulting from the fact that the black curve is continuous but the red curve is discrete, changing values from $\alpha_1\tau$ to $\alpha_2\tau$ at the $0.85p$-th element. The elements to the left of the vertical drop (with indices $\leq 0.85p$) are penalized by the constant penalty $\alpha_2\tau$, while the other elements (with indices $> 0.85p$) are penalized by the larger penalty $\alpha_1\tau$. Therefore, the 2-level penalty can be viewed as two 1-level LASSO penalties applied to two groups of elements in $\bet+\tau \Z$. In the third row, we average out all non-decreasing sub-sequences, thus creating a flattened region of one shared magnitude and not affecting the elements in the non-flattened regions. Now there are three groups of elements in $\bet+\tau \Z$:
\begin{enumerate}
    \item elements with indices $<q_2 p$, as if they are penalized by LASSO with $\alpha_2\tau$;
    \item elements with indices between $q_2 p$ and $q_1 p$, sharing the same magnitude $h$; 
    \item elements with indices $>q_1 p$, as if they are penalized by LASSO with $\alpha_1\tau$.
\end{enumerate}
In the fourth row, we truncate all non-positive elements, which may also truncate the flattened region if the shared magnitude is smaller than 0 (see the fourth row right sub-plot).

\begin{figure}
    \centering    \includegraphics[width=0.34\linewidth]{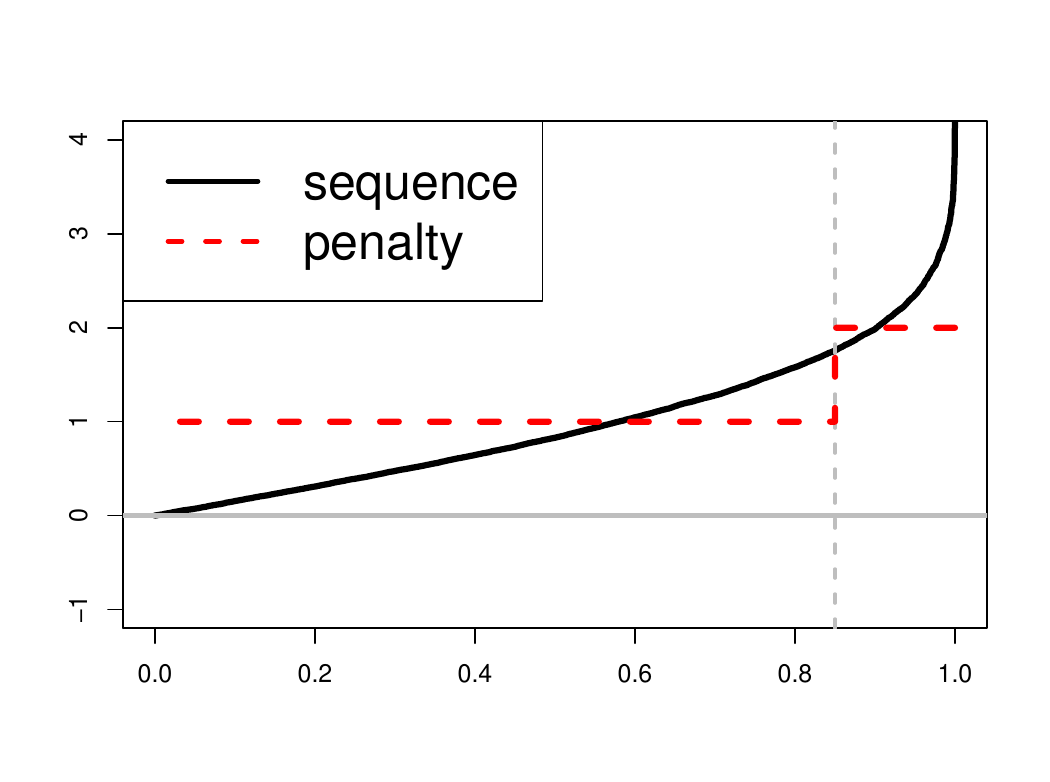}
     \includegraphics[width=0.34\linewidth]{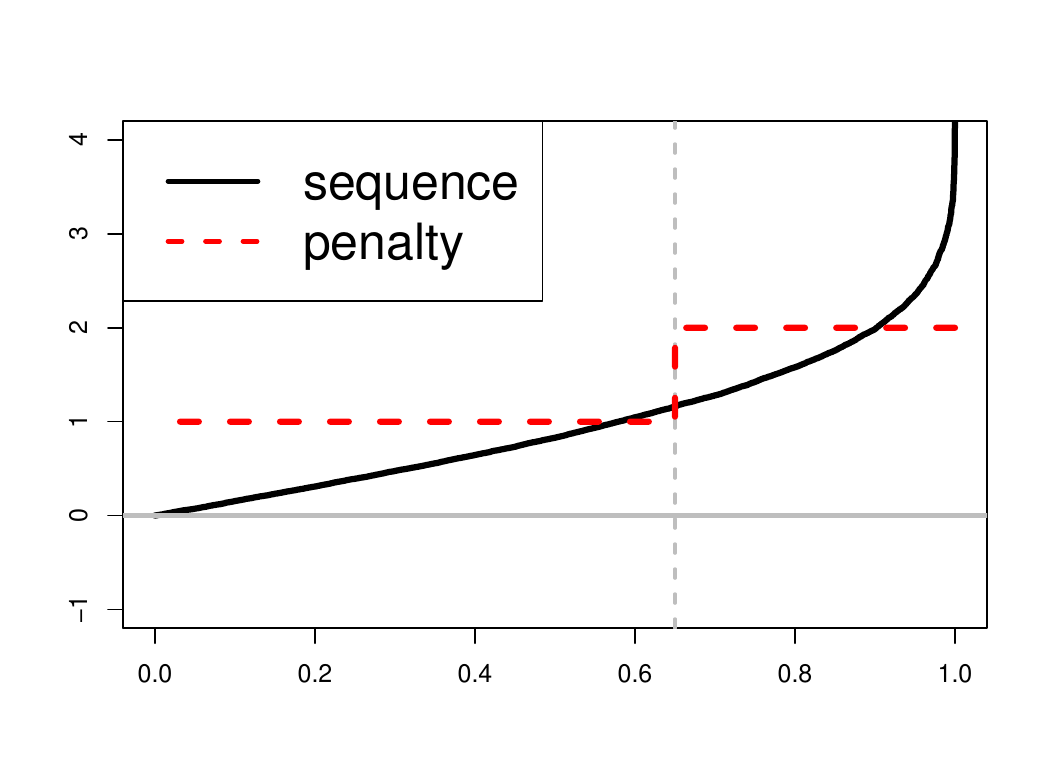}
    \\\vspace{-0.4cm}
\includegraphics[width=0.34\linewidth]{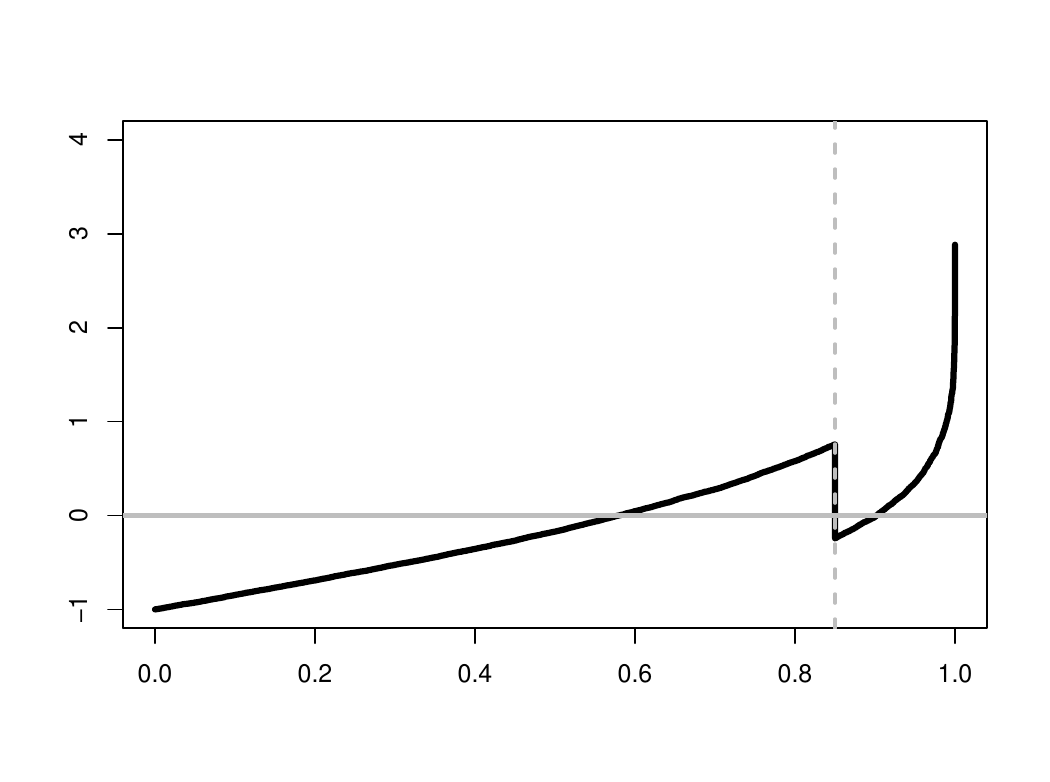}
\includegraphics[width=0.34\linewidth]{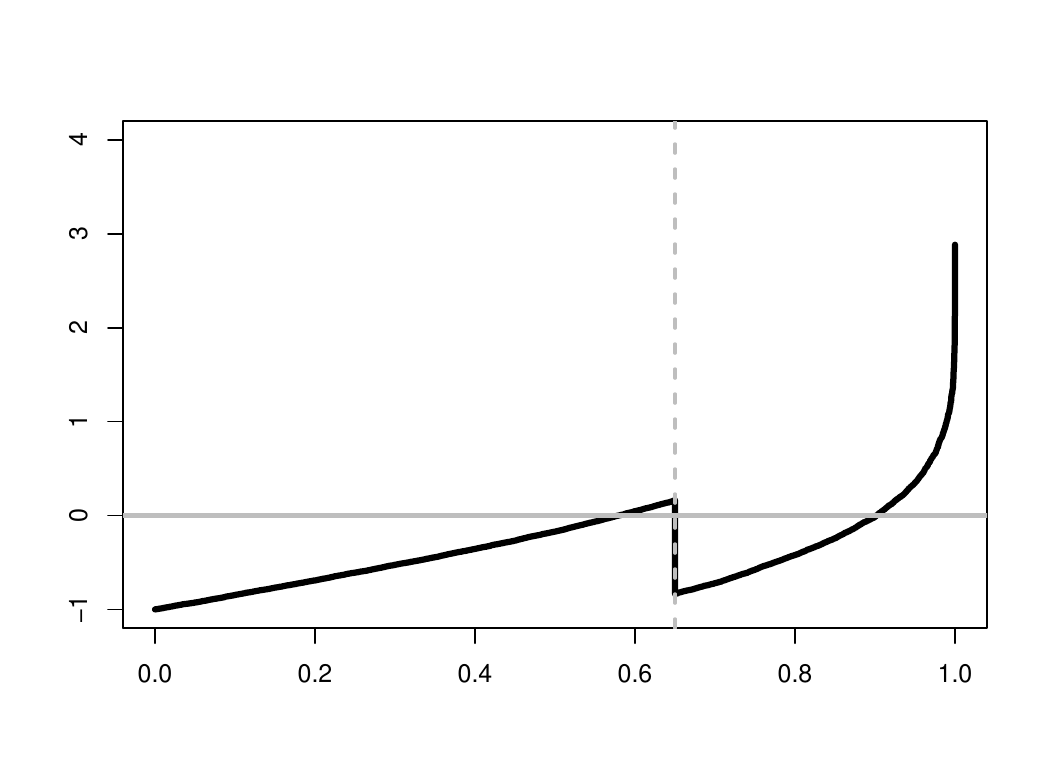}
    \\\vspace{-0.4cm}
\includegraphics[width=0.34\linewidth]{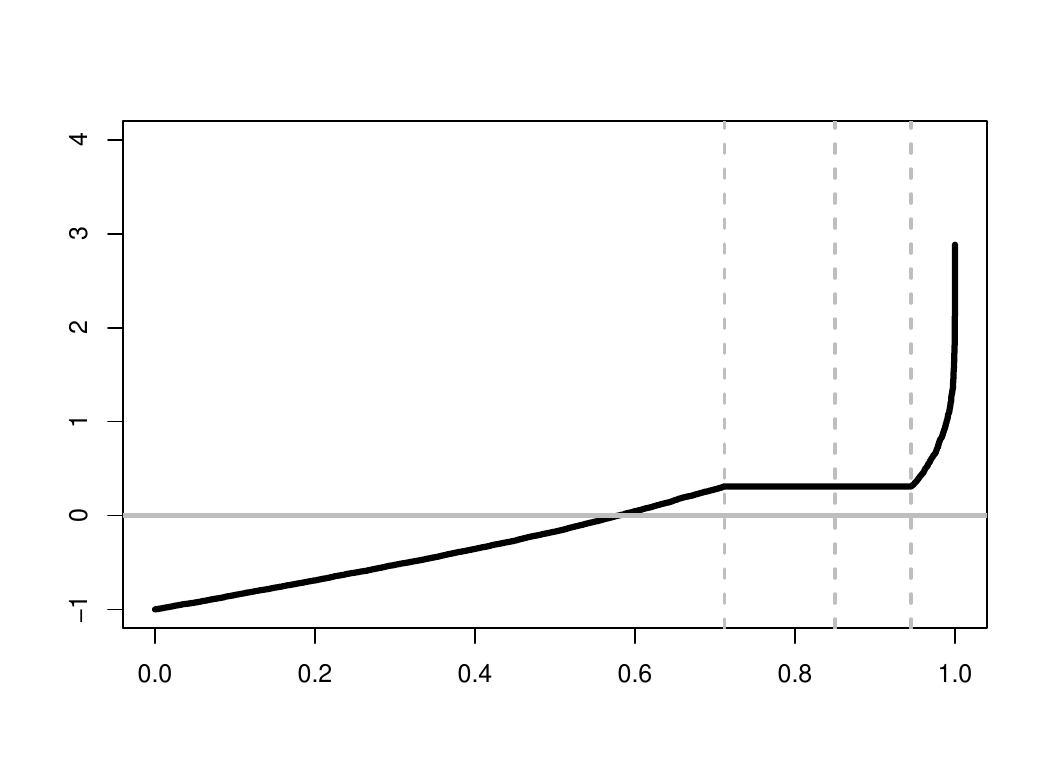}
\includegraphics[width=0.34\linewidth]{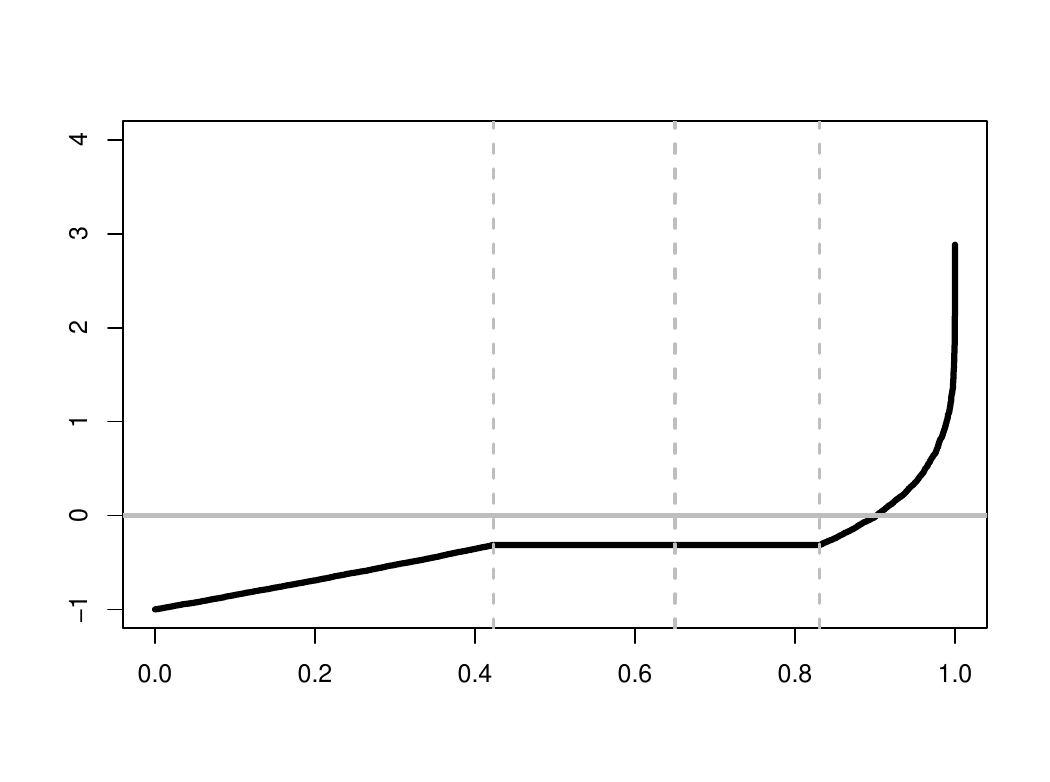}
    \\\vspace{-0.4cm}
\includegraphics[width=0.34\linewidth]{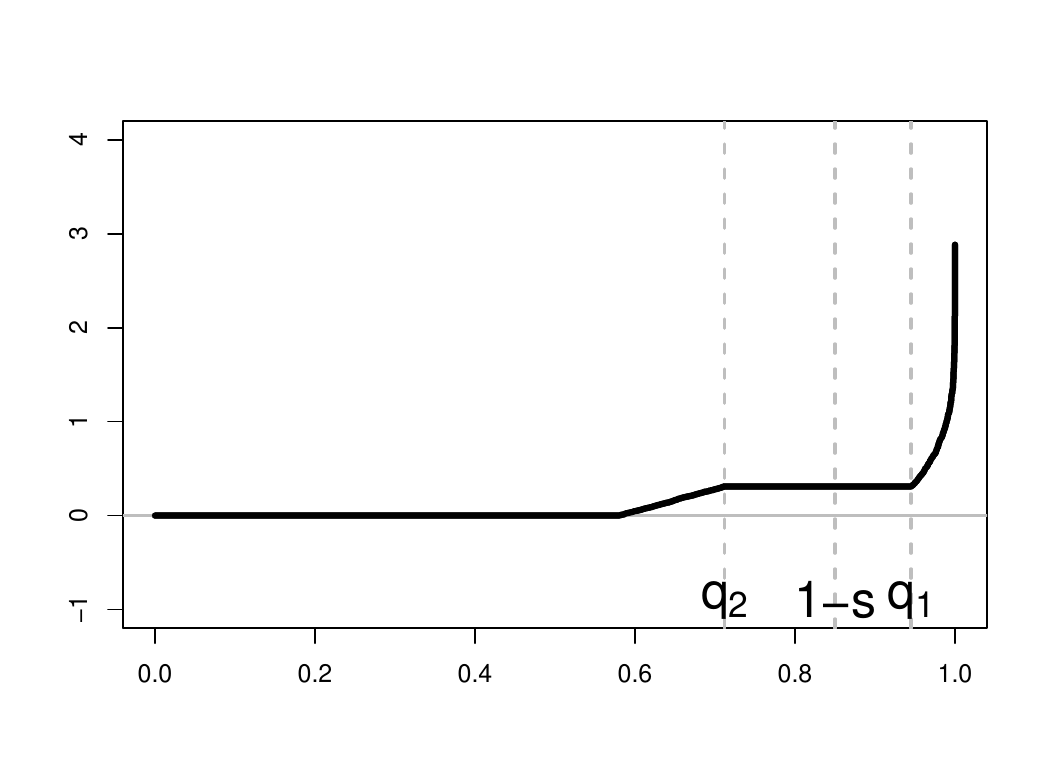}
\includegraphics[width=0.34\linewidth]{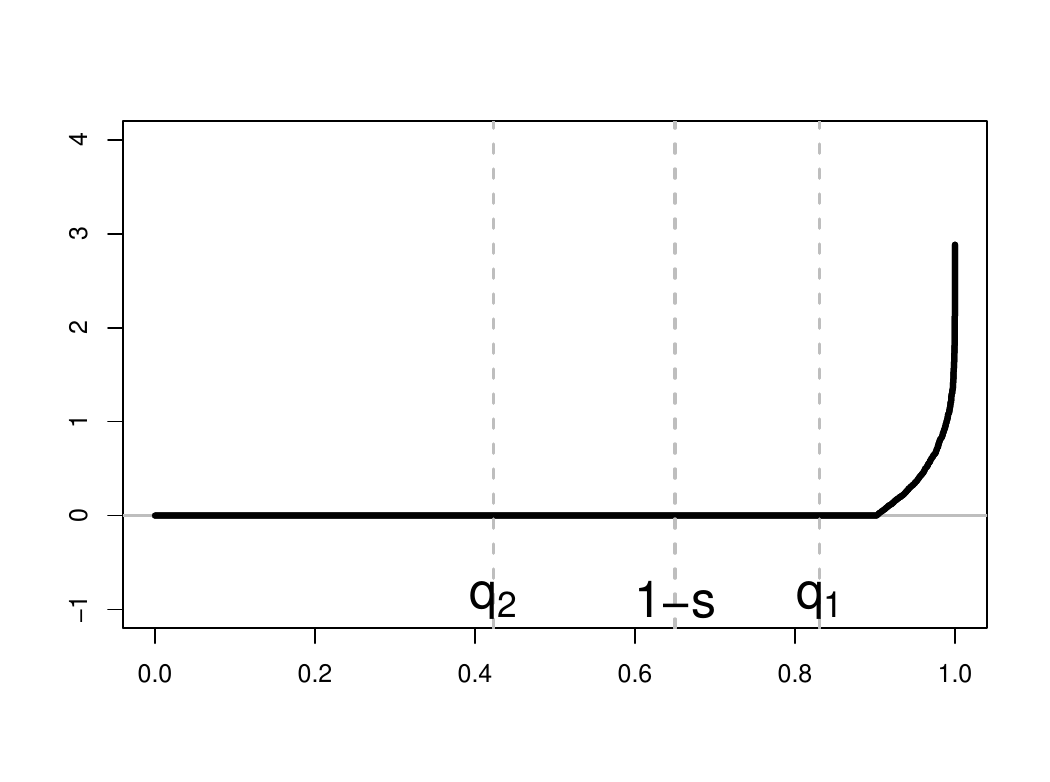}
\vspace{-0.4cm}
\includegraphics[width=0.5\linewidth]{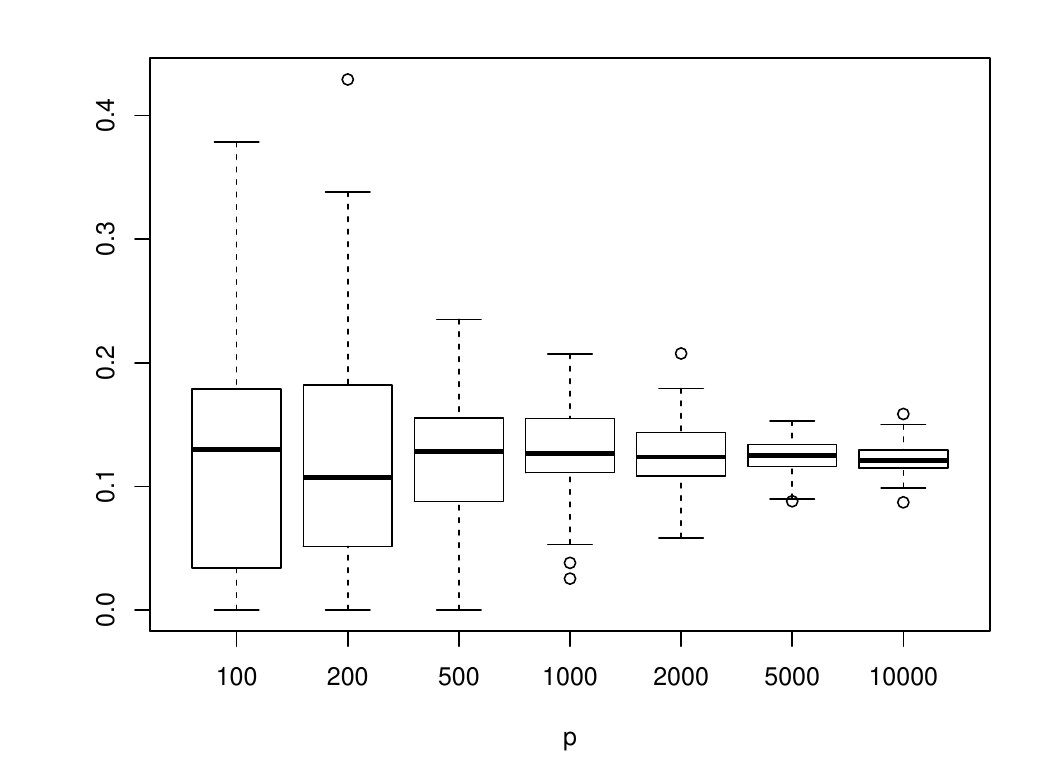}
    \caption{
    Step-by-step derivation of SLOPE solution $\eta_{\pi+Z,A}(\pi+Z)$ (c.f.,\ the PAVA algorithm in Algorithm 1 in the Supplement). Here $\pi=\text{Bernoulli}(0.5)$, and in the left column $\A=\langle 2,1,0.15\rangle$ while in the right column $\A=\langle 2,1,0.35\rangle$. 
    The first row shows the quantile functions of $|\pi+Z|$ and $\A$, corresponding to step 1 in PAVA. 
    The second row shows the difference of the two quantile functions, corresponding to step 2 in PAVA, where the probabilities $<1-s$ are penalized by the smaller level $\alpha_2=1$ and otherwise by the larger level $\alpha_1=2$. The third row averages the sequence, corresponding to step 3 in PAVA, where a shared magnitude is created within the grey vertical lines from $q_2$ to $q_1$. The fourth row truncates the sequence below 0, corresponding to step 4 in PAVA, thus sparsifying the SLOPE solution.
    The last figure uses PAVA to compute the shared magnitude under the setting of left column, for different $p$ each with 100 runs.
    }
    \label{fig:normalized PAVA}
\end{figure}

Overall, the 2-level SLOPE proximal operator is performing soft-thresholding on the non-flat sub-sequences of the input, while performing the averaging on the flat sub-sequence. In other words, the SLOPE penalty is constant on the non-flat sub-sequence of the input, while being adaptive on the flat sub-sequence so that the output is a constant $h$. Therefore, we can write $\Aeff$ for \eqref{eq:slope eff}: \begin{align}
    \Aeff(x;h, q_1, q_2, \alpha_1,\alpha_2,s)=\begin{cases}
       \alpha_1, & \quad q_1<\S_{|\pi+Z|}(|x|), \\
       |x|-h,& \quad q_2 \leq \S_{|\pi+Z|}(|x|)\leq q_1, \\
       \alpha_2, &\quad  \S_{|\pi+Z|}(|x|)<q_2.
    \end{cases}
\label{eq: 2level Aeff}
\end{align}
Here $\S_{|\pi+Z|}(x)=\PP(|\pi+Z|<x)$ is the cumulative distribution function of $|\pi+Z|$, which is monotonically increasing and differentiable. We will derive the precise values of $q_1,q_2,$ and $h$ in \Cref{lem:q1q2h}. 
\begin{lemma}
\label{lem:q1q2h}
For each $(\alpha_1,\alpha_2,s)$ and any $\pi$, we can determine $h$ in \eqref{eq: 2level Aeff} through
\begin{equation}
\int_{h+\alpha_2}^{h+\alpha_1}\PP(|\pi+Z|<x) \, dx=\int_{h+\alpha_2}^{h+\alpha_1}\S_{|\pi+Z|}(x) \, dx=(1-s)(\alpha_1-\alpha_2).
\label{eq:hcalc}
\end{equation}
The solution in $h$ to the above is unique and,
additionally, we have $q_1=\S_{|\pi+Z|}(h+\alpha_1)$ and $q_2=\S_{|\pi+Z|}(h+\alpha_2)$.
\end{lemma}

\begin{proof}[Proof of \Cref{lem:q1q2h}]
To derive $h$, we first need the formulae of $q_1$ and $q_2$. We look at the boundary where $\Aeff$ changes from $\alpha_2$ to $|x|-h$ when $\S_{|\pi+Z|}(|x|)=q_2$ and we set these values equal so as to have a continuous function.
\begin{align*}
    |x|-h=\S_{|\pi+Z|}^{-1}(q_2)-h=\Aeff(\S_{|\pi+Z|}^{-1}(q_2^+))=\Aeff(\S_{|\pi+Z|}^{-1}(q_2^-))=\alpha_2
\end{align*}
which implies that $q_2=\S_{|\pi+Z|}(h+\alpha_2)$. Similarly, we can obtain that $q_1=\S_{|\pi+Z|}(h+\alpha_1)$. Note that both $q_1$ and $q_2$ depend on $h$. Next, we look at the area under curves between $q_2<\S_{|\pi+Z|}(|x|)<q_1$, before (see the second row of \Cref{fig:normalized PAVA}) and after (see the third row of \Cref{fig:normalized PAVA}) the averaging step. It is obvious that the averaging does not change the area under curves:
\begin{align*}
\text{Area before averaging}&= \int_{q_2}^{1-s} \left(\S_{|\pi+Z|}^{-1}(x)-\alpha_2 \right)dx +\int_{1-s}^{q_1} \left(\S_{|\pi+Z|}^{-1}(x)-\alpha_1 \right)dx,
\\
&= \int_{q_2}^{q_1} \S_{|\pi+Z|}^{-1}(x)dx - \alpha_2(1- s - q_2) - \alpha_1 (q_1 - 1+s),
\\
\text{Area after averaging}& =(q_1-q_2) h.
\end{align*}
This leads to
\begin{align}
(q_1-q_2) h+(1-s-q_2) \alpha_2+(q_1-1+s) \alpha_1&=\int_{q_2}^{q_1}\S_{|\pi+Z|}^{-1}(x)dx.
\label{eq:int2}
\end{align}

Moreover, because $\S_{|\pi+Z|}$ is monotone and differentiable, $\S_{|\pi+Z|}^{-1}$ is differentiable. Using the definition of the integral of the inverse function and that $q_i=\S_{|\pi+Z|}(h+\alpha_i)$, we obtain the relationship
\begin{align}
\int_{q_2}^{q_1}\S_{|\pi+Z|}^{-1}(x)dx+\int_{h+\alpha_2}^{h+\alpha_1}\S_{|\pi+Z|}(x)dx=q_1 (h+\alpha_1)-q_2 (h+\alpha_2).
\label{eq:int3}
\end{align}
Merging \eqref{eq:int2} and \eqref{eq:int3}, we conclude that 
\begin{align}
\int_{h+\alpha_2}^{h+\alpha_1}\S_{|\pi+Z|}(x)dx
&=(1-s) (\alpha_1-\alpha_2).
\label{eq:diff}
\end{align}

We now show that, for each $(\alpha_1,\alpha_2, s)$, the solution $h$ is unique, which follows by the intermediate value theorem. This holds due to (1)  the monotonicity of the integration in \Cref{lem:q1q2h} and (2) the opposite signs of \eqref{eq:diff} at $h=-\infty$ and at $h=\infty$. To see the monotonicity, we note that $\S$ is monotonically increases, hence 
\begin{align*}
\frac{d}{dh}    \int_{h+\alpha_2}^{h+\alpha_1}  \S_{|\pi+Z|}(u) \, du &=    \lim_{dh\to 0}\frac{\int_{h+dh+\alpha_2}^{h+dh+\alpha_1}\S_{|\pi+Z|}(u) \, du-\int_{h+\alpha_2}^{h+\alpha_1}\S_{|\pi+Z|}(u) \, du}{dh} \\
&=\S_{|\pi+Z|}(h+\alpha_1)+\S_{|\pi+Z|}(h+\alpha_2)>0.
\end{align*}
To see the opposite signs, we have
$$\lim_{h\to\infty}\int_{h+\alpha_2}^{h+\alpha_1}\S_{|\pi+Z|}(x)dx=(\alpha_1-\alpha_2)
>(1-s) (\alpha_1-\alpha_2),$$
and 
$$\lim_{h\to-\infty}\int_{h+\alpha_2}^{h+\alpha_1}\S_{|\pi+Z|}(x)dx=0
<(1-s) (\alpha_1-\alpha_2).$$
\end{proof}

\section{TPP-FDP trade-off of 2-level SLOPE for all priors}
\label{sec:tpp fdp all priors}
The main purpose of this section is to derive the tight asymptotic TPP-FDP trade-off over all possible pairs of $(\Pi,\Lambda)$. In particular, as discussed previously in \eqref{eq:all prior max zero thres}, to find the infimum $\FDP$ in \eqref{eq:fdp is function of tpp}, we will find the supremum over zero-thresholds. In what follows, we use subscripts to denote the maximum feasible zero-threshold under a class of $\Lambda$, e.g.,\ the 2-level SLOPE $\alpha^*_\text{2-level}$, the general SLOPE $\alpha^*_\text{SLOPE}$, and the LASSO $\alpha^*_\text{LASSO}$. We wish to solve 
\begin{align}
\alpha_\text{2-level}^*(u; \epsilon, \delta):=\sup_{\Pi, \A=\langle\alpha_1,\alpha_2;s\rangle}\{\alpha(\pi,\A): (\pi,\A) \text{ s.t. } \eqref{eq:state evolution}, \TPP=u, \PP(\pi=0)=1-\epsilon\},
\label{eq:all prior max zero thres_new}
\end{align}
where we now have made explicit in the optimization that we are considering only 2-level SLOPE penalty sequences $\A=\langle\alpha_1,\alpha_2;s\rangle$. Moreover, we mention that the representation in \eqref{eq:all prior max zero thres} is implicitly assuming that $(\pi, \A)$ satisfy the state evolution equations \eqref{eq:state evolution} and the assumptions (A1)-(A5), which we have made explicit in the statement in \eqref{eq:all prior max zero thres_new} by adding the conditions $(\pi,\A)$ satisfies \eqref{eq:state evolution} and $\PP(\pi=0)=1-\epsilon$. 

We highlight that the derivation of $\alpha^*_\text{SLOPE}$, for the general SLOPE (i.e.,\ when $\Lambda$ is not restricted to be 2-level as in \eqref{eq:asymptotic lambda} or, equivalently, when $\A$ is not restricted as in \eqref{eq:A two-level}), is infeasible: the best known result by \cite{bu2021characterizing} only establishes an upper bound $t^*$ and a lower bound $t_*$ such that $t_*(u)\leq\alpha^*(u)\leq t^*(u)$. Consequently, the general SLOPE trade-off curve is actually unknown except between the upper and lower bounds. We refer interested readers to \cite[Figure 2]{bu2021characterizing} as well as Figures \ref{fig:summary} and \ref{fig:trade-off curves}.
In sharp contrast, for the 2-level SLOPE (i.e.,\ where $\Lambda$ is restricted to two point masses as in \eqref{eq:asymptotic lambda} or when $\A$ is restricted as in \eqref{eq:A two-level}), we can derive the exact form of $\alpha^*$ from \eqref{eq:all prior max zero thres_new}, and this is what we will do in what follows.
Armed with $\alpha^*_\text{2-level}$, we will be able to derive the 2-level SLOPE trade-off curve explicitly following \eqref{eq:tppfdp rewrite}. 

\subsection{Main ingredients for characterizing  $\alpha^*_\text{2-level}$ explicitly and efficiently}
\label{sec:main ingredients}
There are two key ingredients that allow us to characterize $\alpha^*_\text{2-level}$. The first key ingredient, given in \Cref{lem:second is zero-threshold}, is that the search over $\Lambda$, or equivalently $\A$, in \eqref{eq:all prior max zero thres_new} can be made computationally efficient and easy-to-analyze by considering only a subset $\{\A=\langle\alpha_1,\alpha_2;s\rangle \, \text{ s.t. } \, \alpha(\pi,\A)=\alpha_2\}$, instead of the full set $\{\A=\langle\alpha_1,\alpha_2;s\rangle \, \text{ s.t. } \, \alpha(\pi,\A)=\alpha_2 \text{ or } \alpha_1\}$. We note that $\alpha(\pi,\A)$ will only take the values $\alpha_1$ or $\alpha_2$ due to the fact that the zero-threshold $\alpha$ is always a quantile of $\A$; this is elaborated in the proof of Lemma~\ref{lem:second is zero-threshold} in \Cref{sec:remainingproofs}. 

\begin{lemma}\label{lem:second is zero-threshold}
The zero-threshold $\alpha(\Pi,\Lambda)\equiv \alpha(\pi,\A)$ of a 2-level penalty $\A=\langle\alpha_1,\alpha_2;s\rangle$ must be $\alpha_1$ or $\alpha_2$. To minimize $\FDP$ in \eqref{eq:fdp is function of tpp}, it suffices to only consider $\A$ (or equivalently $\Lambda$) such that the second penalty level serves as the zero-threshold, i.e.,\ $\alpha(\Pi,\Lambda)\equiv\alpha(\pi,\A)=\alpha_2$.
\end{lemma}

Using \Cref{lem:second is zero-threshold}, we can simplify \eqref{eq:all prior max zero thres_new} to
\begin{align}
\alpha_\text{2-level}^*(u; \epsilon, \delta):=\sup_{\Pi, \A=\langle\alpha_1,\alpha_2;s\rangle}\{\alpha_2: (\pi,\A) \text{ s.t. } \eqref{eq:state evolution}, \TPP=u, \PP(\pi=0)=1-\epsilon\}.
\label{eq:all prior max zero thres_new2}
\end{align}


The second key ingredient is to optimize over a 2-dimensional space $(t_1,t_2)$, instead of the infinite-dimensional space $\Pi$, for any given $\A=\langle\alpha_1,\alpha_2;s\rangle$. 
In more detail, \cite[Lemma 4.3]{bu2021characterizing} tells us that the $\pi$ that maximizes the zero-threshold in \eqref{eq:all prior max zero thres_new2} is a three-point prior. Namely,
\begin{align}
\pi_{\min}(t_1,t_2)=\begin{cases}
0 &\text{ w.p. } 1-\epsilon,\\
t_1 &\text{ w.p. } \epsilon \rho(t_1, t_2),\\
t_2 &\text{ w.p. } \epsilon(1 - \rho(t_1, t_2)),\\
\end{cases}
\text{ and } \quad 
\pi_{\min}^*(t_1,t_2)=
\begin{cases}
t_1 &\text{ w.p. } \rho(t_1, t_2),\\
t_2 &\text{ w.p. } 1 - \rho(t_1, t_2),
\end{cases}
\label{eq:three point pi}
\end{align}
for some constants $t_1, t_2$, and $0\leq \rho\leq 1$ (which depends on $u$ and $\alpha_2$), where we notice that the prior in \eqref{eq:three point pi} satisfies the constraint $\PP(\pi=0)=1-\epsilon$, and $\pi^*_{\min}:=(\pi_{\min}|\pi_{\min}\neq 0)$. 

Although there are no closed form representations for the optimal $t_1$ and $t_2$ in \eqref{eq:three point pi}, we can derive $\rho$ by leveraging the constraint $\TPP=u$ and using the fact that the zero-threshold $\alpha$ equals $\alpha_2$ for $\A = \langle\alpha_1,\alpha_2;s\rangle$. In particular, the value of $\rho$ is given in \Cref{lem:rho_def} below, whose proof is given in \Cref{sec:remainingproofs}.

\begin{lemma}
\label{lem:rho_def} Given that $\TPP = u$ and the fact that the zero-threshold $\alpha$ equals $\alpha_2$ for $A = \langle\alpha_1,\alpha_2;s\rangle$, for the optimal prior $\pi_{\min}$ from \eqref{eq:three point pi} taking non-zero values $(t_1, t_2)$, we have
\begin{align}
\rho(t_1, t_2; \alpha_2, u)=
\begin{cases}
1 &\text{ if } t_1=t_2,
\\
\frac{u-\left[\Phi\left(t_2-\alpha_2\right)+\Phi\left(-t_2-\alpha_2\right)\right]}{\left[\Phi\left(t_1-\alpha_2\right)+\Phi\left(-t_1-\alpha_2\right)\right]-\left[\Phi\left(t_2-\alpha_2\right)+\Phi\left(-t_2-\alpha_2\right)\right]}&\text{ if } t_1\neq t_2.
\end{cases}
\label{eq:compute p}
\end{align}
\end{lemma}

Employing the optimal $\pi_{\min}$ in \eqref{eq:three point pi} and using the definition of $\rho$ from \eqref{eq:compute p} of \Cref{lem:rho_def}, we further simplify the representation of \eqref{eq:all prior max zero thres_new2} to
\begin{align}
\alpha_\text{2-level}^*(u)=\sup_{(t_1,t_2), \A=\langle\alpha_1,\alpha_2;s\rangle}\{\alpha_2&: (\pi_{\min},\A) \text{ s.t. } \eqref{eq:state evolution}\}.
\label{eq:all prior max zero thres_new3}
\end{align}

Before stating our main results we make one more simplification to the representation in \eqref{eq:all prior max zero thres_new3}. Recall that \eqref{eq:state evolution} says
\[\tau^{2}=\sigma^{2}+\frac{1}{\delta} \, \mathbb{E}\left[\left(\eta_{\Pi+\tau Z,\A\tau}(\Pi+\tau Z)-\Pi\right)^{2}\right],\]
where we note from \Cref{sec:limiting scalar function} that the limiting scalar function of any SLOPE can  be written as an adaptive soft-thresholding function; namely, $\eta_{\pi+Z,\A}(x)=\eta_\text{soft}(x;\widehat\A_\text{eff}(x))$.
Now, define
\begin{align}
F[\pi;\alpha_1,\alpha_2,s]:=\mathbb{E}\left[\left(\eta_{\mathrm{soft}}\left(\pi+Z ; \mathrm{A}_{\mathrm{eff}}(\pi+Z;\alpha_1,\alpha_2,s)\right)-\pi\right)^2\right],
\label{eq:state evolution constraint}
\end{align}
then we can equivalently write the state evolution \eqref{eq:state evolution} by renormalizing and employing \eqref{eq:state evolution constraint} as $\tau^2=\sigma^2+\frac{\tau^2}{\delta}F[\pi]$ and, 
for a pair $(\pi,\A)$, we have that $\tau$ exists whenever $F[\pi]=(1-\frac{\sigma^2}{\tau^2})\delta\leq \delta$. Therefore, we can rewrite the condition that $(\pi_{\min},\A)$ satisfies \eqref{eq:state evolution} in \eqref{eq:all prior max zero thres_new3} by $F[\pi_{\min}(t_1,t_2)]\leq \delta$, and we have arrived at \eqref{eq:alpha star explicit}:
\begin{align}
\alpha^*_\text{2-level}(u;\epsilon,\delta)=\sup_{\A=\langle\alpha_1,\alpha_2;s\rangle} \left\{\alpha_2 \quad \text{ s.t. }  \quad \min_{0\leq t_1\leq t_2} F\left[\pi_{\min}(t_1,t_2;u,\epsilon);\alpha_1,\alpha_2,s \right]\leq \delta \right\}.
\label{eq:alpha star explicit}
\end{align}

In the next section, we state our main results on the asymptotic $\TPP$-$\FDP$ trade-off curve for 2-level SLOPE, which use the representation of $\alpha^*_\text{2-level}$ given in \eqref{eq:alpha star explicit}. Following the statement of the main results, we discuss how $\alpha^*_\text{2-level}$ given in \eqref{eq:alpha star explicit} can be computed explicitly and efficiently using \Cref{alg:F rho min} in \Cref{sec:derivation alpha star}. This essentially boils down to refining the definition of $\Aeff$ in  \Cref{eq: 2level Aeff} and consequently the 2-level SLOPE limiting scalar function in \eqref{eq:2-level limiting} from general $\epsilon$-sparse priors $\pi$ to the specific prior three-point prior in \eqref{eq:three point pi}.
Finally, \Cref{sec:analytic slope} compares the 2-level SLOPE trade-off curve to the general SLOPE trade-off curve and \Cref{sec:remainingproofs} contains the lemma proofs.

\subsection{Main results on the 2-level SLOPE $\TPP$-$\FDP$ trade-off curve}

We now present our main result on the asymptotic $\TPP$-$\FDP$ trade-off curve for 2-level SLOPE.


\begin{theorem}\label{thm:trade-off}
Under the working assumptions (A1)$\sim$(A5) in \Cref{sec:AMP}, the following inequality holds for all $(\Pi,\Lambda)$:
$$
\FDP(\Pi,\Lambda) \geq q_\text{2-level}\left(\TPP(\Pi,\Lambda) ; \delta, \epsilon\right),
$$
where, with $\alpha_\text{2-level}^*$ given in \eqref{eq:alpha star explicit}, the 2-level SLOPE trade-off is
\begin{equation}
q_\text{2-level}(u;\delta,\epsilon)=\frac{2(1-\epsilon) \Phi\left(-\alpha_\text{2-level}^*(u)\right)}{2(1-\epsilon) \Phi\left(-\alpha_\text{2-level}^*(u)\right)+\epsilon u}.
\label{eq:qu}
\end{equation}

Furthermore, the trade-off $q_\text{2-level}$ is tight: for any continuous curve $q(u) \geq q_\text{2-level}(u)$ with strict inequality for some $0\leq u \leq 1$, there exists a pair $(\Pi,\Lambda)$ such that
$$
\FDP(\Pi,\Lambda) < q\left(\TPP(\Pi,\Lambda) ; \delta, \epsilon\right).
$$
\end{theorem}

\begin{proof}[Proof of \Cref{thm:trade-off}]
Recall the $\FDP$, defined in \eqref{eq:fdp is function of tpp} and restated below:
$$\FDP(\Pi,\Lambda)=\frac{2(1-\epsilon) \Phi\left(-\alpha(\Pi,\Lambda)\right)}{2(1-\epsilon) \Phi\left(-\alpha(\Pi,\Lambda)\right)+\epsilon \TPP(\Pi,\Lambda)}.$$
Notice that $\FDP$ is decreasing in the zero-threshold, $\alpha$. Therefore, with $\alpha_\text{2-level}^*$ defined as the supreme in 
\eqref{eq:all prior max zero thres}, i.e.,\
\begin{equation}
\alpha_\text{2-level}^*=\sup_{\Pi, \Lambda=\langle\lambda_1,\lambda_2;s\rangle} \left\{\alpha(\Pi, \Lambda) \quad \text{ s.t. } \quad \TPP(\Pi,\Lambda)=u\right\},
\label{eq:alpha_thm}
\end{equation}
we have that, for any $(\Pi,\Lambda)$,
\begin{equation}
\begin{split}
\FDP(\Pi,\Lambda)&=\frac{2(1-\epsilon) \Phi\left(-\alpha(\Pi,\Lambda)\right)}{2(1-\epsilon) \Phi\left(-\alpha(\Pi,\Lambda)\right)+\epsilon \TPP(\Pi,\Lambda)} \\
&\geq \frac{2(1-\epsilon) \Phi\left(-\alpha_\text{2-level}^*(\Pi,\Lambda)\right)}{2(1-\epsilon) \Phi\left(-\alpha_\text{2-level}^*(\Pi,\Lambda)\right)+\epsilon u }=
q_\text{2-level}(\TPP;\delta,\epsilon).
\label{eq:thm_lower_bound}
\end{split}
\end{equation}
The tightness of \eqref{eq:thm_lower_bound} follows naturally as $\alpha_\text{2-level}^*$ is the supreme of all feasible zero-thresholds. \Cref{sec:main ingredients} shows that $\alpha_\text{2-level}^*$ given by the representation in \eqref{eq:alpha star explicit} also satisfies \eqref{eq:alpha_thm}.
\end{proof}

\begin{rem}
\Cref{thm:trade-off} is similar to previous trade-offs in \cite{su2017false,bu2021characterizing} in the sense that $\FDP\geq q_\text{LASSO}(\TPP)$ and $\FDP\geq q_\text{SLOPE}(\TPP)$ for some functions $q$ that are the lower bounds of minimum $\FDP$, although the formulation is different. In other words, the general SLOPE, 2-level SLOPE and LASSO all demonstrate trade-offs that do not allow $\FDP$ to reach 0. We refer to \Cref{app:q details} for a detailed comparison of these $q$ functions.
\end{rem}


We emphasize that the first statement of \Cref{thm:trade-off} characterizes a different $\TPP$-$\FDP$ trade-off than the LASSO trade-off, in the sense that the 2-level SLOPE trade-off overcomes the TPP upper limit present for the LASSO and that there exists a range of $\TPP$ values such that the corresponding $\FDP$ is strictly smaller for 2-level SLOPE than LASSO. In particular, the second statement of \Cref{thm:trade-off} characterizes the tightness of the 2-level SLOPE trade-off, whereas the tightness of the general SLOPE trade-off is unknown. In short, 
$$q_\text{LASSO}(\TPP)\geq q_\text{2-level}(\TPP)\geq q_\text{SLOPE}(\TPP),$$
for all $\TPP$; hence,
\begin{align}
\alpha^*_\text{LASSO}\leq \alpha^*_\text{2-level}\leq \alpha^*_\text{SLOPE}.
\label{eq:alpha order}
\end{align}

\begin{rem}
There exists a range of $\TPP$ values such that the inequalities above hold strictly: we visualize $q_\text{LASSO}(\TPP)>q_\text{2-level}(\TPP)$ in \Cref{fig:summary} and we prove  $q_\text{2-level}(\TPP)>q_\text{SLOPE}(\TPP)$ in \Cref{sec:analytic slope}. The result $q_\text{LASSO}(\TPP) > q_\text{SLOPE}(\TPP)$ is also presented in \cite[Appendix E.3]{bu2021characterizing}.
\end{rem}

\begin{figure}[!htb]
    \centering
    \includegraphics[width=0.475\linewidth]{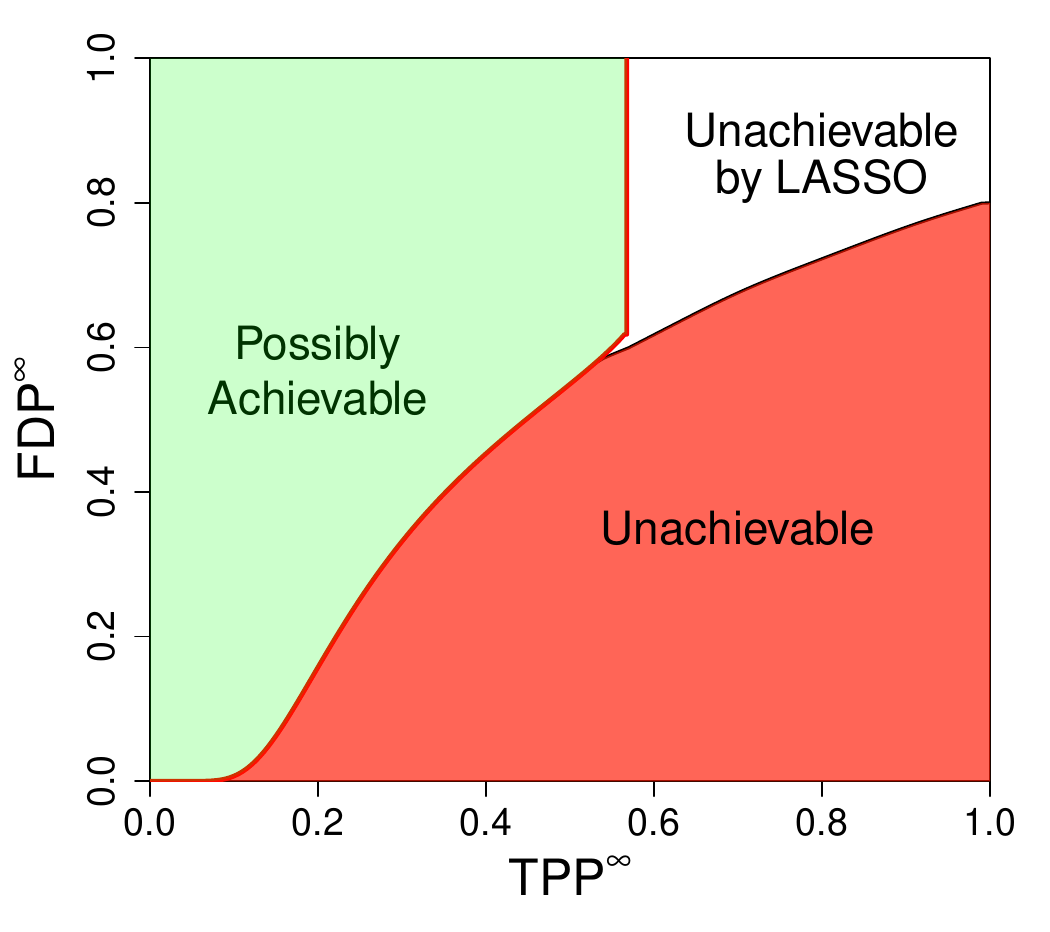}
    \includegraphics[width=0.475\linewidth]{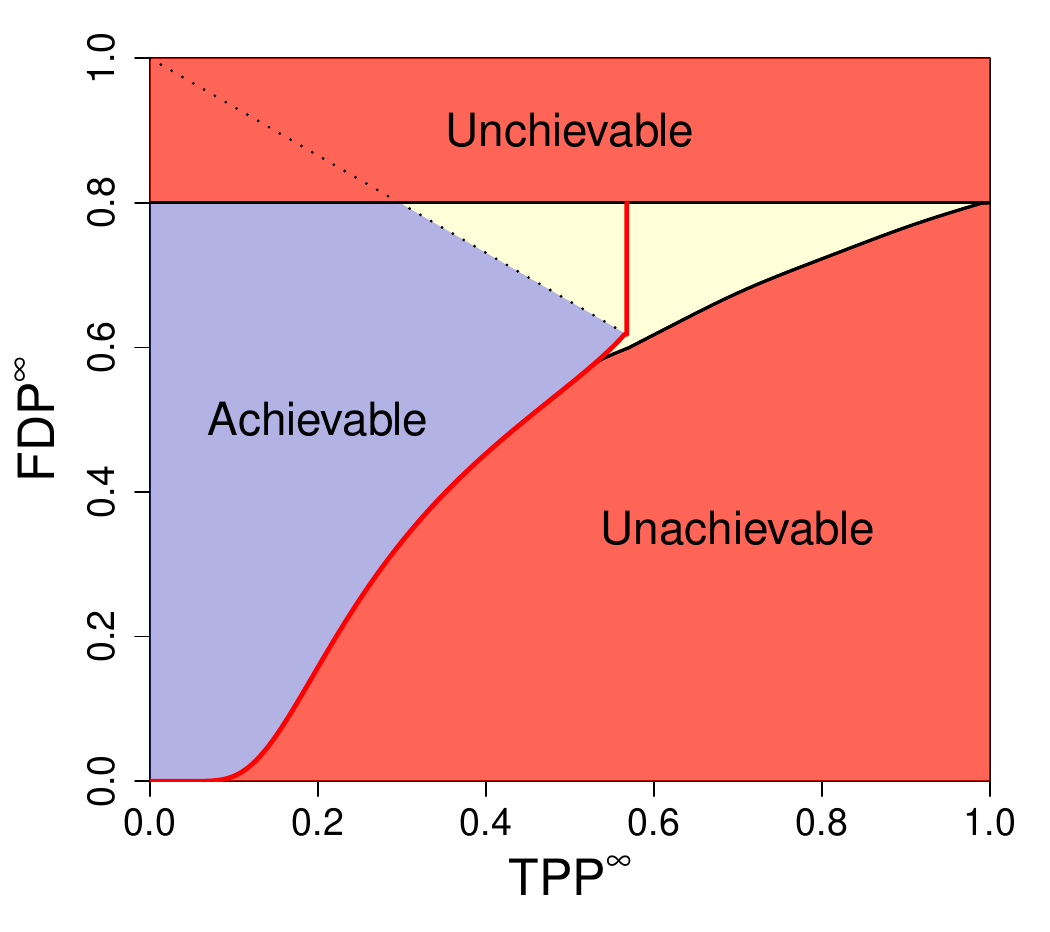}
    \caption{2-level SLOPE $\TPP$–$\FDP$ diagram  by \Cref{thm:trade-off} (left) and \Cref{thm:complete trade-off} (right). The red curve is the LASSO (1-level) trade-off and the black curve is the general 2-level SLOPE trade-off. The red region is $(\TPP,\FDP)$ pairs not achievable by 2-level SLOPE nor by the LASSO, regardless of the prior distribution on the signal and the choice of penalty. On the left sub-plot, the white region is possibly achievable by 2-level SLOPE but not achievable by the LASSO; the green region is possibly achievable by both. On the right sub-plot, the ``possibly achievable'' regions are precisely decomposed: the yellow region is achievable only by 2-level SLOPE but not the LASSO; the purple region is achievable by both. We note that the boundary between the purple region and the yellow one is a line connecting $(0,1)$ and $(u_\text{DT},q_\text{LASSO}(u_\text{DT}))$, given by \cite{wang2020complete}, where $u_\text{DT}\approx 0.5676$. Here we have used $\epsilon=0.2,\delta=0.3$.}
    \label{fig:complete or not trade-off}
\end{figure}

To be more precise, we present the complete 2-level SLOPE TPP-FDP trade-off in the right sub-plot of \Cref{fig:complete or not trade-off}. In contrast to the LASSO, a significantly large range of $(\TPP,\FDP)$ pairs are provably achievable beyond the LASSO's TPP upper limit. This result is presented in \Cref{thm:complete trade-off} following from \Cref{thm:trade-off}.

\begin{theorem}\label{thm:complete trade-off}
Let $\mathcal{D}_{\epsilon, \delta}$ be the region  enclosed by three curves: $\FDP=1-\epsilon, \FDP=q_\text{2-level}\left(\TPP\right)$, and $\TPP=0$. For $\delta<1$ and $0<\epsilon<1$, any $\left(\TPP, \FDP\right)$ in $\mathcal{D}_{\epsilon, \delta}$ is asymptotically achievable by the 2-level SLOPE. 
\end{theorem}

\begin{rem}
In \Cref{thm:complete trade-off}, the first curve $\FDP=1-\epsilon$ is the worst $\FDP$ that is achieved by a naive estimator that selects all predictors, i.e.,\ applying $\alpha(\Pi,\Lambda)=0$ (and consequently $\TPP=1$) in \eqref{eq:fdp is function of tpp} arrives at this boundary. The second curve, $\FDP=q_\text{2-level}\left(\TPP\right)$, is derived from \Cref{thm:trade-off},  which shows that the region below this curve is not achievable. The last curve, $\TPP=0$, is trivial because all $\TPP$ must be non-negative.
\end{rem}


\subsection{Computing $\alpha^*_\text{2-level}$}
\label{sec:derivation alpha star}
In this section, we discuss how to efficiently and explicitly compute  $\alpha^*_\text{2-level}$ defined in \eqref{eq:alpha star explicit}.  This is done via a search over $\alpha_1 \geq \alpha_2 \geq 0$ and $0 \leq s \leq 1$ as detailed in \Cref{alg:tpp fdp all pi}, where the main computational challenge is in computing $F[\pi_{\min}]$ from \eqref{eq:state evolution constraint}. We first discuss \Cref{alg:tpp fdp all pi} and the rest of the section details how to compute $F[\pi_{\min}]$.
\begin{algorithm}
\caption{Minimizing $\FDP$ at $\TPP=u$ over all penalties $\Lambda$ and all priors $\Pi$ }
\begin{algorithmic}[1]
\State Input: the lower and upper bounds, $\alpha_\text{(lower/upper)}$, of the maximum zero-threshold $\alpha^*_\text{2-level}$;
constants $(\epsilon,\delta)$
\For{$\alpha_2\in [\alpha_\text{(lower)}(u),\alpha_\text{(upper)}(u)]$}
\For{$\alpha_1\in [\alpha_2,\infty)$}
\For{$s\in[0,1]$}
\For{$0\leq t_1\leq t_2$}
\State Compute $F[\pi_{\min}(t_1,t_2;u,\epsilon);\alpha_1,\alpha_2,s]$ from \Cref{alg:F rho min}.
\If{$F\leq\delta$} 
\State $\alpha_2$ is feasible, break the loops and move to larger $\alpha_2$
\EndIf
\EndFor
\EndFor
\EndFor
\EndFor
\State Output: the maximum feasible zero-threshold $\alpha_2$ and the minimum $\FDP$ by \eqref{eq:fdp is function of tpp}
\end{algorithmic}
\label{alg:tpp fdp all pi}
\end{algorithm}

Overall, \Cref{alg:tpp fdp all pi} allows us to derive the maximum zero-threshold $\alpha_2$ in the 2-level SLOPE, in sharp contrast to the general SLOPE. The difference between the two settings, which we detail below, is that for 2-level SLOPE we are able refine the representation of $\Aeff$ in  \eqref{eq: 2level Aeff} -- and consequently the  2-level SLOPE limiting scalar function in \eqref{eq:2-level limiting} -- from general $\epsilon$-sparse priors $\pi$ to the specific three-point prior in \eqref{eq:three point pi} and this allows for efficient computation. In contrast, for general SLOPE, because the representation of $\Aeff$ is not known, one must leverage quadratic programming to find the optimal $\Aeff$. 

In particular, \Cref{alg:tpp fdp all pi} can be made very efficient using two tricks in line 2: (1) searching a narrow domain for $\alpha_2$ by \Cref{rem:bisection}; (2) using the bisection method to quickly converge to the maximum feasible $\alpha_2$.

\begin{rem}
\label{rem:bisection}
In \Cref{alg:tpp fdp all pi}, we can naively set $\alpha_\text{(lower)}(u)=0$ and $\alpha_\text{(upper)}(u)$ to be a sufficiently large number; though, the algorithm may be quite slow if the searching domain is large. Alternatively, we can set a much narrower domain by leveraging \eqref{eq:alpha order}: we can set $\alpha_\text{(lower)}(u)=\alpha^*_\text{LASSO}(u)$ and $\alpha_\text{(upper)}(u)=\alpha^*_\text{SLOPE}(u)$.
\end{rem}

Computation of $F[\pi_{\min}]$ is detailed in \Cref{alg:F rho min}, which references \Cref{lem:3 point} and \Cref{lem:expand E}. We will first discuss the lemmas and then formally present them afterwards.

\begin{algorithm}
\caption{Computing $F[\pi_{\min}]$}
\begin{algorithmic}[1]
\State Input: Constants $t_1, t_2$; $\TPP=u$; 2-level SLOPE penalty $\alpha_1,\alpha_2, s$; the sparsity $\epsilon$
\State Compute the probability $\rho(t_1,t_2;u, \alpha_2)$ from \eqref{eq:compute p} in \Cref{lem:rho_def}. This requires the input values $\alpha_2$ and $u$.
\State Compute the shared magnitude $h$ and quantiles $q_1$ and $q_2$ from \Cref{lem:3 point}. This requires the input values $t_1, t_2,\alpha_1,\alpha_2,s, \epsilon$ and $\rho$ from the previous step.
\State Define the function $\mathrm{A}_{\mathrm{eff}}(x)$ using \eqref{eq: 2level Aeff}. This requires the input values $\alpha_1,\alpha_2, \epsilon, t_1, t_2$, as well as $h, q_1, q_2$ and $\rho$ from the previous steps.
\State Define the function $\mathcal{E}(t)$ in \eqref{eq:mathcal E} using \Cref{lem:expand E} for $t \in \{0, t_1, t_2\}$. This requires input values $\alpha_1$ and $\alpha_2$, as well as the shared magnitude $h$ calculated in a previous step.
\State Output: $F[\pi_{\min}]$ in \eqref{eq:SE t1t2}. This requires input value $\epsilon$, as well as $\mathcal{E}(0), \mathcal{E}(t_1), \mathcal{E}(t_2)$, and $\rho$  from the previous steps.
\end{algorithmic}
\label{alg:F rho min}
\end{algorithm}

\Cref{lem:3 point}, used in step 3 of the algorithm, derives the shared magnitude $h$ and the values $q_1$ and $q_2$ following \Cref{lem:q1q2h} given $( \alpha_1,\alpha_2,s, \epsilon)$ and $\pi_{\min}$, namely, given $(t_1, t_2, \rho)$.
Recall that from \Cref{lem:q1q2h}, we know that $h$ is unique, and we can further make the formula of $h$ easily computable when $\pi$ is the three-point prior with the bisection method, thanks to the monotonicity of $\S_{|\pi+Z|}(x)=\PP(|\pi+Z|<x)$, proven in \Cref{lem:q1q2h}.

Once we have calculated the shared magnitude $h$ and the values $q_1$ and $q_2$ with \Cref{lem:3 point}, one can define the effective penalty $\Aeff$ using \eqref{eq: 2level Aeff}. Armed with the $\Aeff$ function, we are in a place to compute $F[\pi_{\min}]$ efficiently. Doing so boils down to the following. We first notice that given $(t_1, t_2, \rho)$ for the three-point prior, for any $(\alpha_1,\alpha_2, s)$, we can represent $F[\pi_{\min}]$ explicitly by substituting $\pi_{\min}$ from \eqref{eq:three point pi} into \eqref{eq:state evolution constraint} and deriving
\begin{align}
\begin{split}
F[\pi_{\min}(t_1,t_2;u);\alpha_1,\alpha_2,s]
&=(1-\epsilon) \mathcal{E}(0)+\epsilon  \rho  \mathcal{E}(t_1)+\epsilon  (1-\rho) \mathcal{E}(t_2),
\end{split}
\label{eq:SE t1t2}
\end{align}
where we have defined
\begin{align}
\mathcal{E}(t):=\mathbb{E}\left[\left(\eta_{\mathrm{soft}}\left(t+Z ; \mathrm{A}_{\mathrm{eff}}(t+Z)\right)-t\right)^2\right].    
\label{eq:mathcal E}
\end{align}
In fact, each of $\mathcal{E}(0), \mathcal{E}(t_1)$, and $\mathcal{E}(t_2)$ in \eqref{eq:SE t1t2} can be computed exactly and efficiently using \Cref{lem:expand E} (on the other hand, this is not true for general SLOPE), so as to make the TPP-FDP trade-off of 2-level SLOPE explicit. While we leave the details of the \Cref{lem:expand E} proof  for the appendix, we highlight that a key quantity in the derivation of $\mathcal{E}(t)$ in \eqref{eq:mathcal E} is the shared magnitude $h$ calculated with \Cref{lem:3 point}. We mention that the essential tools to  prove \Cref{lem:expand E} are simple facts about truncated Gaussian expectations given in Fact \ref{fact:truncated normal}.

In what follows, we formally present the lemmas, whose proofs are given in \Cref{sec:remainingproofs} and \Cref{sec:proof56}, respectively.

\begin{lemma}
\label{lem:3 point}
Given $\pi_{min}$ in \eqref{eq:three point pi} (namely, given $(t_1, t_2, \rho)$) and given $(\alpha_1,\alpha_2,s, \epsilon)$, one can compute the unique shared magnitude $h$ from the equation
\begin{align}
\begin{split}
(1-s) (\alpha_1-\alpha_2)
&=(1-\epsilon)\int_{h+\alpha_2}^{h+\alpha_1}\PP(|Z|<x) \, dx +\epsilon \rho\int_{h+\alpha_2}^{h+\alpha_1}\PP(|t_1+Z|<x) \, dx \\
& \qquad +\epsilon(1-\rho)\int_{h+\alpha_2}^{h+\alpha_1}\PP(|t_2+Z|<x) \, dx,
\end{split}
\label{eq:substitute}
\end{align}
where $Z$ is a standard normal and each of the integrals can be expanded as
\begin{equation}
\begin{split}
\label{eq:int_expand}
\int_{v_2}^{v_1}\PP(|t+Z|<x) \, dx &=\sum_{v \in \{v_1, -v_1\}}(v-t)\Phi(v-t)+\phi(v-t) \\
& \qquad - \sum_{v' \in \{v_2, -v_2\}}(v'-t)\Phi(v'-t)+\phi(v'-t).
\end{split}
\end{equation}
In the above, $\Phi$ is the cumulative distribution function and $\phi$ is the probability density function of the standard normal. 
Finally,  we have 
$$q_1=\S_{|\pi_{\min}+Z|}(h+\alpha_1),\quad q_2=\S_{|\pi_{\min}+Z|}(h+\alpha_2)$$ 
where
\begin{equation}
\begin{split}
\S_{|\pi_{\min}+Z|}(x)=\PP(|\pi_{\min}+Z|<x) &=(1-\epsilon)\left[\Phi(x)-\Phi(-x)\right]  \\
&\qquad +\epsilon \rho\left[\Phi(x-t_1)-\Phi(-x-t_1)\right]\\
&\qquad +\epsilon(1-\rho)\left[\Phi(x-t_2)-\Phi(-x-t_2)\right].
\label{eq:Sexpand}
\end{split}
\end{equation}
\end{lemma}




\begin{lemma}
\label{lem:expand E}
For $\Aeff(x) = \Aeff(x;h,\alpha_1,\alpha_2,s)$ in \eqref{eq: 2level Aeff} and any $t\in\R$, we have
\begin{align*}
&\mathcal{E}(t) = \mathbb{E}\left\{\left(\eta_{\mathrm{soft}}\left(t+Z ; \mathrm{A}_{\mathrm{eff}}(t+Z)\right)-t\right)^2\right\} 
\\
&=\left(\alpha_1^2+1\right)\left[1-\Phi\left(\alpha_1+h-t\right)+\Phi\left(-\alpha_1-h-t\right)\right] +t^2\left[\Phi(\alpha_2-t)-\Phi(-\alpha_2-t)\right]
\\
&+(h-t)^2\left[\Phi(\alpha_1+h-t)-\Phi(\alpha_2+h-t)\right]+(h+t)^2\left[\Phi(-\alpha_2-h-t)-\Phi(-\alpha_1-h-t)\right]
\\
&+\left(\alpha_2^2+1\right)\left[\Phi\left(-\alpha_2-t\right)-\Phi\left(-\alpha_2-h-t\right)+\Phi\left(\alpha_2+h-t\right)-\Phi\left(\alpha_2-t\right)\right]\\
&-(\alpha_1-h+t)\phi(\alpha_1+h-t)-(\alpha_1-h-t)\phi(-\alpha_1-h-t)+(\alpha_2-h-t) \phi(-\alpha_2-h-t) \\
&+(\alpha_2-h+t)\phi(\alpha_2+h-t) -(\alpha_2-t)\phi(-\alpha_2-t)-(\alpha_2+t)\phi(\alpha_2-t).
\end{align*}
\end{lemma}

\subsection{Comparing to the analytic solution of general SLOPE}
\label{sec:analytic slope}

\begin{figure}[!ht]
    \centering
\includegraphics[width=0.45\textwidth]{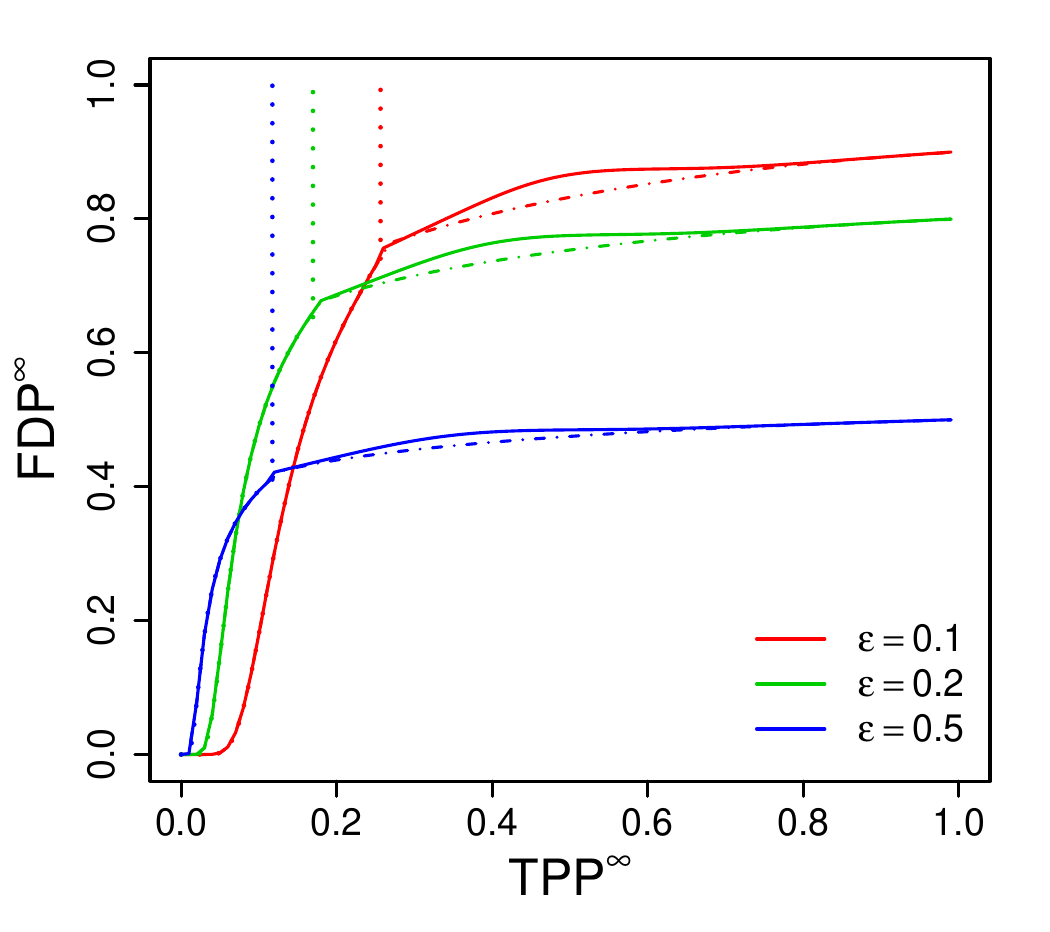}
\includegraphics[width=0.45\textwidth]{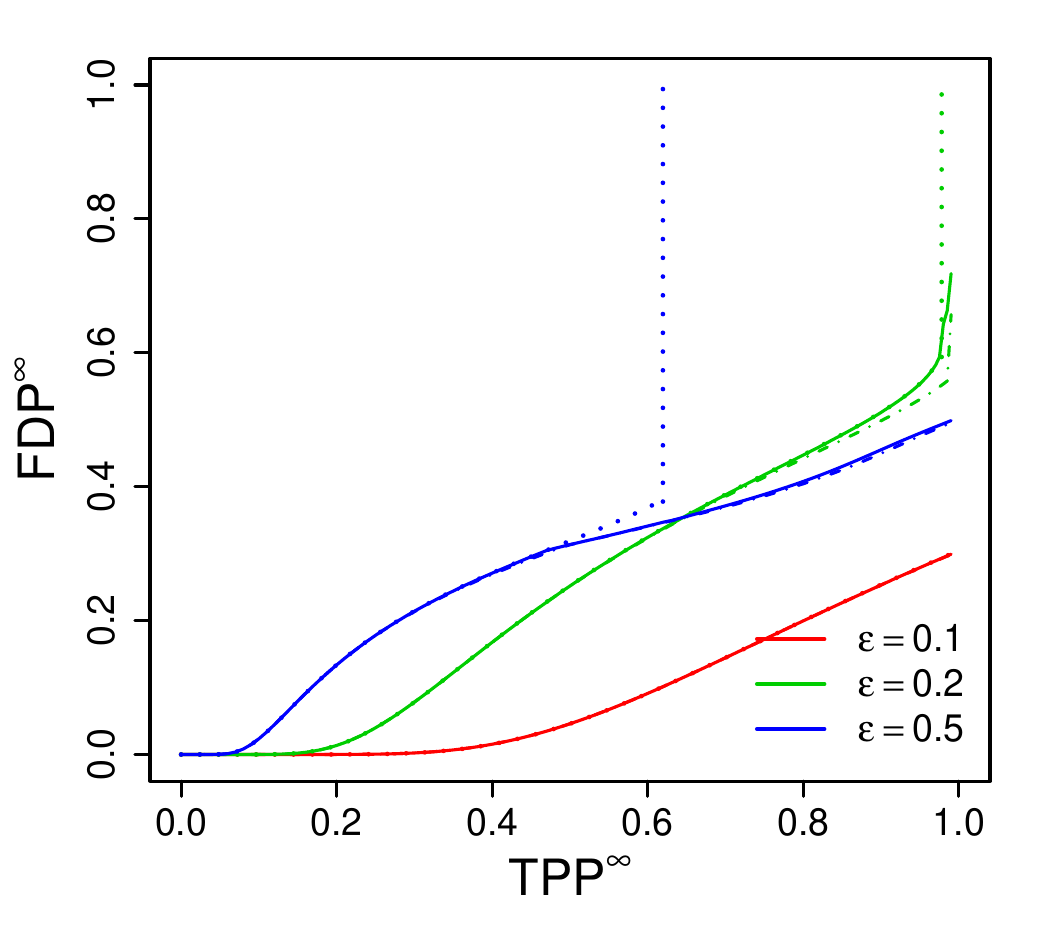}
\includegraphics[width=0.45\textwidth]{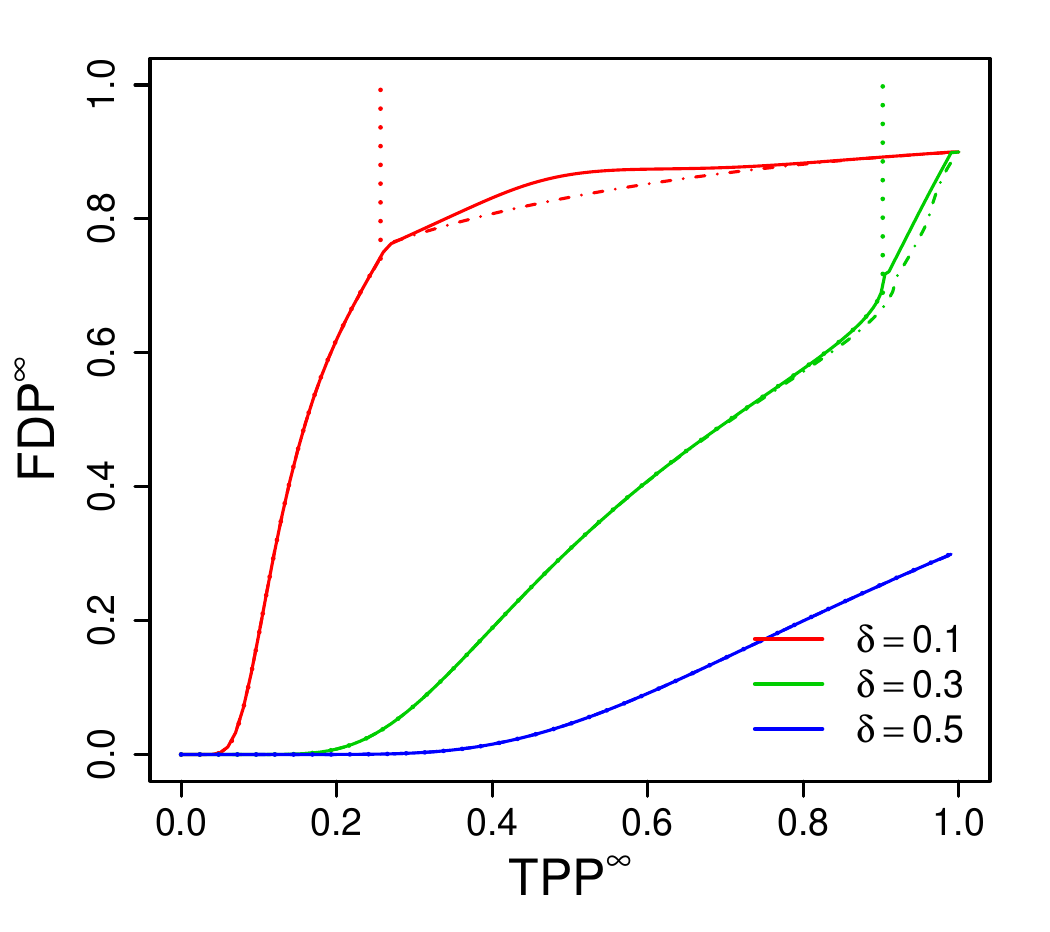}
\includegraphics[width=0.45\textwidth]{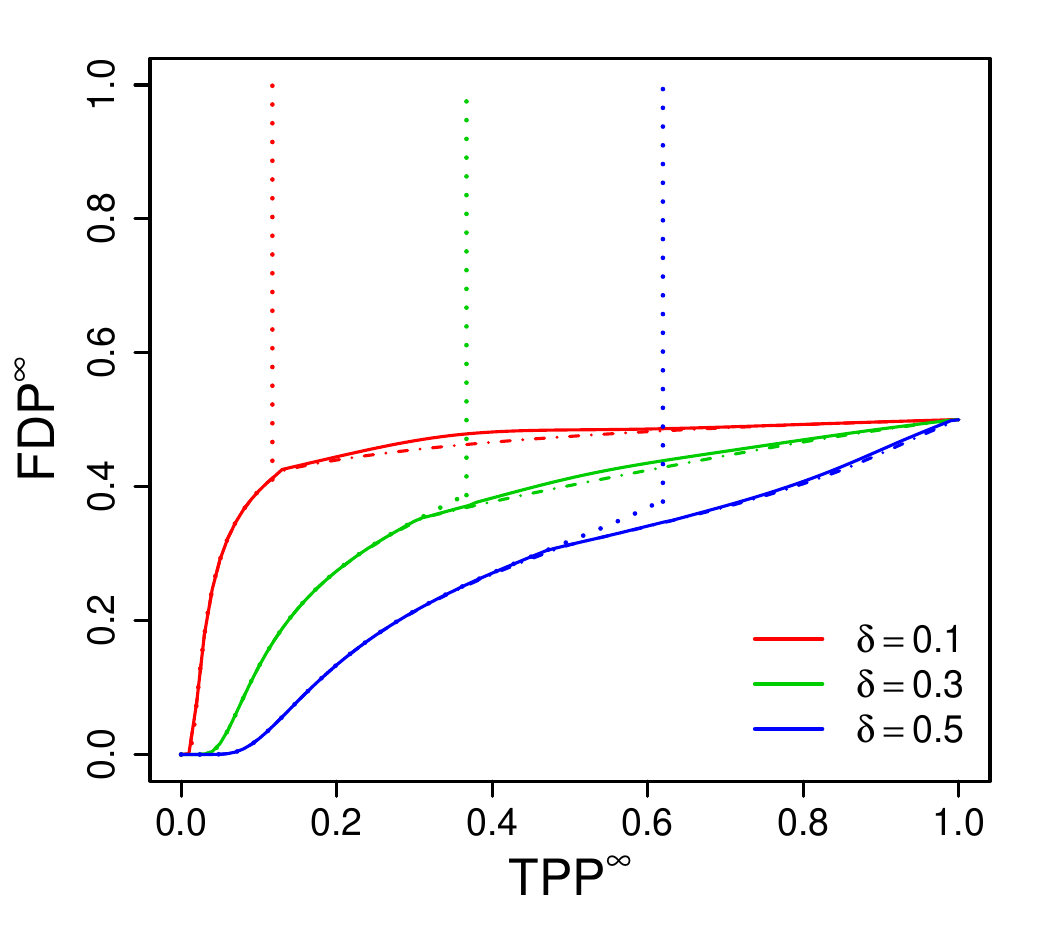}
    \caption{Examples of the SLOPE trade-off bounds $q_\text{LASSO}$, $q_\text{2-level}$ and $q_\text{SLOPE}$ for various $(\epsilon,\delta)$ settings. Top-left:
$\delta= 0.1$; top-right: $\delta= 0.5$; bottom-left: $\epsilon= 0.1$; bottom-right: $\epsilon= 0.5$. In the plots, the solid lines are the 2-level SLOPE trade-off, the dotted lines are the LASSO trade-off, and the dot-dashed lines are the non-tight lower bound of general SLOPE trade-off.}
    \label{fig:trade-off curves}
\end{figure}

We plot the TPP-FDP trade-offs $q_\text{LASSO} 
 (\text{dotted})\geq q_\text{2-level} (\text{solid})\geq q_\text{SLOPE} (\text{dot-dashed})$ in \Cref{fig:trade-off curves}, with respect to various $\epsilon$ and $\delta$ combinations.

For the general SLOPE, because the limiting scalar function is implicit,
the state evolution and the penalty function cannot be precisely evaluated. This means one must use quadratic programming to numerically find the optimal penalty function $\Aeff$, in contrast to the analytic construction of $\Aeff$ described in \Cref{sec:derivation alpha star} for 2-level SLOPE.

Nevertheless, it is possible to analytically construct $\Aeff$ for the general SLOPE under special cases, e.g., for three-point prior $\pi_{\min}(t_1,t_2)$:
from \cite[Appendix E.1]{bu2021characterizing}, if the function in \eqref{eq:AHH} is monotonically increasing, then it is the uniquely optimal SLOPE penalty minimizing the $\FDP$,
\begin{align}
\Aeff(x;\alpha_\text{SLOPE}^*)=
\begin{cases}
-H'(x)/H(x), &\text{ if } |x|>\alpha_\text{SLOPE}^*,
\\
\alpha_\text{SLOPE}^*, &\text{ if } |x|\leq\alpha_\text{SLOPE}^*,
\end{cases}
\label{eq:AHH}
\end{align}
where, with $t_1,t_2$ defined in \eqref{eq:three point pi} and $\rho$ defined in \eqref{eq:compute p}, 
$$H(x):=4(1-\epsilon) \phi(x)+2 \epsilon\rho\left[\phi(x-t_1)+\phi(x+t_1)\right]+2 \epsilon(1-\rho)\left[\phi(x-t_2)+\phi(x+t_2)\right].$$


However, from \eqref{eq: 2level Aeff}, the 2-level SLOPE has $\Aeff(x)=|x|-h$ whenever it is non-constant.
Clearly the linear shape of $|x|-h$ does not match that of $-H'(x)/H(x)$. Thus, 2-level SLOPE is not the optimal SLOPE in terms of the TPP-FDP trade-off. On the other hand, 2-level SLOPE is similar to optimal SLOPE. This can be viewed in \Cref{fig:optimal SLOPE} and from the TPP-FDP trade-offs in \Cref{fig:summary,fig:complete or not trade-off,fig:trade-off curves} and later in \Cref{fig:fixed prior TPPFDP,fig:fixed prior vary}.

\subsection{Remaining lemma proofs} \label{sec:remainingproofs}

\begin{proof}[Proof of \Cref{lem:second is zero-threshold}]
We first quote \cite[Proposition C.5]{bu2021characterizing}, which shows that the zero-threshold $\alpha(\pi,\A)$ must be some quantile of $\A$ for any SLOPE penalty $\A$. Given that 2-level SLOPE penalty $\A=\langle\alpha_1,\alpha_2;s\rangle$ only takes two values, it is clear that $\alpha(\pi,\A)=\alpha_1$ or $\alpha(\pi,\A)=\alpha_2$. Hence the full set is defined as $\{\A=\langle\alpha_1,\alpha_2;s\rangle \, \text{ s.t. } \, \alpha(\pi,\A)=\alpha_2 \text{ or } \alpha_1\}$.

Next, we define $\mathcal{A}_2:=\{\langle \alpha_1,\alpha_2,s\rangle:\alpha_1\geq \alpha_2,\alpha(\pi,\langle \alpha_1,\alpha_2;s\rangle)=\alpha_2\}$, i.e.,\ $\mathcal{A}_2$ is a subset containing all the 2-level SLOPE penalties whose second level serves as the zero-threshold, and denote $\A^*=\langle a^*_1,a^*_2,s^*\rangle$ as the minimizer of \eqref{eq:fdp is function of tpp} in the normalized regime under 2-level SLOPE. We aim to prove $\A^*\in \mathcal{A}_2$.

Suppose $\A^*\not\in\mathcal{A}_2$, then it must hold that the zero-threshold of $\A^*$ is its first level $\alpha^*_1$, i.e., $\alpha(\pi,\A^*)= \alpha_1^*$. In addition, we can define a LASSO penalty $\A'=\langle \alpha^*_1,\alpha^*_1,s^*\rangle\in \mathcal{A}_2$ with the same zero-threshold $\alpha(\pi,\A')=\alpha(\pi,\A^*)=\alpha^*_1$. Recall from \eqref{eq:fdp is function of tpp} that the same zero-threshold means the same $\FDP$; hence, we can conclude that $\A'\in \mathcal{A}_2$ also minimizes the $\FDP$. Therefore, to find the minimum of $\FDP$, it is sufficient to consider the smaller subset $\mathcal{A}_2$, as required.
\end{proof}

\begin{proof}[Proof of \Cref{lem:rho_def}]

First notice that for the optimal three-point prior in \eqref{eq:three point pi}, 
using the definition of $\TPP$ in \eqref{eq:tppfdp comp},
we have
\begin{equation}
\begin{split}
\label{eq:rho_def}
u = \TPP&=\mathbb{P}\left(\eta_{\pi+Z, \A}\left(\pi_{\min}^*+Z\right) \neq 0\right)\\
&=\mathbb{P}\left(|\pi_{\min}^*+Z| \geq \alpha_2\right)
\\
&=\rho\cdot \mathbb{P}\left(|t_1+Z| \geq \alpha_2\right)+(1-\rho)\cdot \mathbb{P}\left(|t_2+Z| \geq \alpha_2\right)
\\
&=
\rho\cdot[\Phi(t_1-\alpha_2)+\Phi(-t_1-\alpha_2)]+(1-\rho)\cdot[\Phi(t_2-\alpha_2)+\Phi(-t_2-\alpha_2)],
\end{split}
\end{equation}
in which we have used that
\begin{align*}
\mathbb{P}\left(|t+Z| \geq \alpha_2\right) &= \mathbb{P}\left(\{t+Z \geq \alpha_2\} \cup \{t+Z \leq -\alpha_2\}\right) \\
&= \mathbb{P}\left(\{Z \geq -t + \alpha_2\} \cup \{Z \leq -t -\alpha_2\}\right) \\
&= \mathbb{P}\left(\{Z \geq -t + \alpha_2\}\right) + \mathbb{P}\left( \{Z \leq -t -\alpha_2\}\right) \\
&= \Phi(t - \alpha_2) + \Phi(-t -\alpha_2),
\end{align*}
Notice that the value of $\rho$ in \eqref{eq:compute p} follows from rearrangement of \eqref{eq:rho_def}.
\end{proof}

\begin{proof}[Proof of \Cref{lem:3 point}] The representation in 
\eqref{eq:substitute} follows immediately by replacing $\pi$ in \eqref{eq:hcalc} of \Cref{lem:q1q2h} with the three-point prior $\pi_{\min}$ in \eqref{eq:three point pi}.

Next, we derive the formula of $\int_{v_2}^{v_1}\PP(|t+Z|<x) \, dx$ given in \eqref{eq:int_expand}. We first notice that
\begin{equation}
\begin{split}
\PP(|t+Z|<x) =\PP(-x<t+Z<x)&=\PP(-x-t<Z<x-t) \\
&= \PP(Z<x-t)-\PP(Z<-x-t) \\
&= \Phi(x-t)-\Phi(-x-t). 
\label{eq:prob_expand}
\end{split}
\end{equation}
Therefore, using that
$$\int \Phi(a+bx)dx=\frac{1}{b}\left((a+bx)\Phi(a+bx)+\phi(a+bx)\right) + \text{constant},$$
we have
\begin{align*}
&\int_{v_2}^{v_1}\PP(|t+Z|<x) \, dx
=\int_{v_2}^{v_1}\left[\Phi(x-t)-\Phi(-x-t) \right]dx\\
&=\left[\left((x-t)\Phi(x-t)+\phi(x-t)\right)+\left((-x-t)\Phi(-x-t)+\phi(-x-t)\right)\right]^{v_1}_{v_2},
\end{align*}
which expands to the desired formula. 

We finally notice that the result in \eqref{eq:Sexpand} follows from \eqref{eq:prob_expand}:
\begin{align*}
\PP(|\pi_{\min}+Z|<x) &= (1-\epsilon)\PP(|Z|<x) +\epsilon \rho\PP(|t_1+Z|<x)+\epsilon(1-\rho)\PP(|t_2+Z|<x) \\
&= (1-\epsilon)\left[\Phi(x)-\Phi(-x)\right] +\epsilon \rho\left[\Phi(x-t_1)-\Phi(-x-t_1)\right] \\
&\qquad +\epsilon(1-\rho)\left[\Phi(x-t_2)-\Phi(-x-t_2)\right].
\end{align*}
\end{proof}

\section{TPP-FDP of 2-level SLOPE under a fixed prior}
\label{sec:6}
In this section, we aim to construct the optimal asymptotic $\TPP$-$\FDP$ trade-off curve when the prior distribution on the signal is fixed and the optimization is only over the penalty sequence $\Lambda$ (or equivalently, $\A$). In particular, as discussed previously in \eqref{eq:fixed prior max zero thres}, to find the infimum $\FDP$ in \eqref{eq:fdp is function of tpp}, we will find the supremum over zero-thresholds; hence, we wish to solve 
\begin{align}
\alpha^*_\text{2-level}(u; \epsilon, \delta, \pi):=\sup_{ \A=\langle\alpha_1,\alpha_2;s\rangle}\{\alpha(\pi,\A): (\pi,\A) \text{ s.t. } \eqref{eq:state evolution}, \TPP=u\},
\label{eq:fixed prior max zero thres_new}
\end{align}
where now the prior $\Pi$ (or equivalently, the non-zero component $\Pi^*$) is fixed and given, but $\pi=\Pi/\tau(\A)$ is dependent on $\A$. As in \eqref{eq:all prior max zero thres_new} in \Cref{sec:tpp fdp all priors},
we now have made explicit in the optimization that we are considering only 2-level SLOPE penalty sequences $\A=\langle\alpha_1,\alpha_2;s\rangle$ and that $(\pi, \A)$ must satisfy the state evolution equations \eqref{eq:state evolution}.  We work under the same conditions as \Cref{sec:tpp fdp all priors}, namely under assumptions (A1)$\sim$(A5) in \Cref{sec:AMP}, including that $\PP(\Pi=\Pi^*)=\epsilon=1-\PP(\Pi=0)$. In \Cref{sec:TPP FDP fixed priors}, we discuss in more detail the optimization in \eqref{eq:fixed prior max zero thres_new} and in \Cref{sec:constant_or_nothing} we consider a specific type of prior for which the trade-off curve can be derived explicitly.

\subsection{Minimizing FDP given TPP for general, fixed priors}
\label{sec:TPP FDP fixed priors}

We begin with a discussion of the optimization in  \eqref{eq:fixed prior max zero thres_new} for general priors, but we highlight that the derivation in this section is pedagogical. In particular,  in \Cref{sec:constant_or_nothing} we consider a special subclass of fixed priors, in which case the derivation of $\alpha^*$ from \eqref{eq:fixed prior max zero thres_new} can be made explicit.
We consider the TPP-FDP trade-off in \eqref{eq:fdp is function of tpp}:
$$q^{(\Pi)}(u)=\min_\Lambda \{\FDP(\Pi,\Lambda):\TPP(\Pi,\Lambda)=u\}.$$
Clearly, we have $q^{(\Pi)}(u)\geq q(u)\equiv \min_{\Pi,\Lambda} \{\FDP(\Pi,\Lambda):\TPP(\Pi,\Lambda)=u\}$ 
for all $u$, where $q(u)$ denotes the trade-off curve in \eqref{eq:qu}, because $q(u)\equiv \min_\Pi q^{(\Pi)}(u)$.\footnote{Comparing \Cref{fig:fixed prior TPPFDP} and \Cref{fig:trade-off curves}, we see that $q^{(\Pi)}$ in general has a different shape than $q$ for the general SLOPE, 2-level SLOPE, and LASSO. For example, at $(\epsilon,\delta)=(0.5,0.5)$, the $\TPP$-$\FDP$ curves in the right column of \Cref{fig:fixed prior TPPFDP} have a sharp increase near $\TPP=0$, whereas the blue curves in \Cref{fig:trade-off curves} are almost flat.}

We now describe the algorithm to derive $q_\text{2-level}^{(\Pi)}$ for the 2-level SLOPE, as summarized in \Cref{alg:fixed prior}. Given $\TPP=u$, from \eqref{eq:tppfdp rewrite2}, for a fixed $\A=\langle\alpha_1,\alpha_2;s\rangle$ this means
\begin{align}
\PP\left(\left|\Pi^*/\tau(\alpha_1,\alpha_2,s)+Z \right|>\alpha_2\right)=u.
\label{eq:tpp fixed prior}
\end{align}
Recall that \Cref{lem:second is zero-threshold} allows us to only consider $\alpha_2$ as the zero-threshold. Therefore, as discussed previously, to obtain the minimum $\FDP$, we need to search over all feasible $(\alpha_1,\alpha_2,s)$, where feasible means that $(\pi,\A)$ satisfies \eqref{eq:state evolution} subject to \eqref{eq:tpp fixed prior}, and determine the largest zero-threshold $\alpha^*_\text{2-level}$. This is summarized by the optimization program in  \eqref{eq:fixed prior max zero thres_new}. 

\begin{algorithm}[!htb]
\caption{Minimizing $\FDP$ at $\TPP=u$ over 2-level SLOPE penalty with fixed prior $\Pi$}
\begin{algorithmic}[1]
\State Find LASSO zero-threshold $\alpha_\text{LASSO}(u)$ such that $\PP\left(|\Pi^*/\tau(\alpha_\text{LASSO})+Z|>\alpha_\text{LASSO}\right)=u$ by \eqref{eq:tpp fixed prior}
\For{$\alpha_2\in [\alpha_\text{LASSO}(u),\infty)$}
\For{$\alpha_1\in [\alpha_2,\infty)$ and $s\in[0,1]$}
\State Compute $\tau(\alpha_1,\alpha_2,s)$ iteratively by \eqref{eq:state evolution}
\If{$\TPP=u$ by \eqref{eq:tpp fixed prior}}
\State $\alpha_2$ is feasible, break the loop and move to larger $\alpha_2$
\EndIf
\EndFor
\EndFor
\State Output: the maximum feasible zero-threshold $\alpha_2$ and the minimum $\FDP$ by \eqref{eq:fdp is function of tpp}
\end{algorithmic}
\label{alg:fixed prior}
\end{algorithm}

We emphasize that for any LASSO penalty $\alpha_\text{LASSO}$ and any $\TPP$, one can always find a 2-level SLOPE penalty with the same $\TPP$ but a larger zero-threshold $\alpha_2$; hence, a smaller $\tau$. Equivalently, 2-level SLOPE enjoys a smaller $\FDP$ -- because of a larger zero-threshold, and a smaller estimation error for free -- because of a smaller $\tau$ and the fact that $\text{MSE}=\delta(\tau^2-\sigma^2)$ (c.f.,\ Corollary 3.4 in \cite{bu2019algorithmic}).

In fact, a stronger statement by Theorem 3 in \cite{bu2021characterizing} shows that the general SLOPE strictly outperforms any LASSO in terms of higher $\TPP$, lower $\FDP$, and smaller MSE simultaneously (c.f.,\ their Appendix F where the statement is actually proven by constructing 2-level SLOPE to outperform LASSO). We visualize such outperformance of 2-level SLOPE in \Cref{fig:fixed prior TPPFDP}, where we also observe the same behavior as in \Cref{fig:summary}: the 2-level SLOPE overcomes the $\TPP$ upper limit that restricts the LASSO, and enjoys strictly smaller $\FDP$ than the LASSO for a wide range of $\TPP$ values.

\begin{figure}[!htb]
    \centering
    \includegraphics[width=0.32\textwidth]{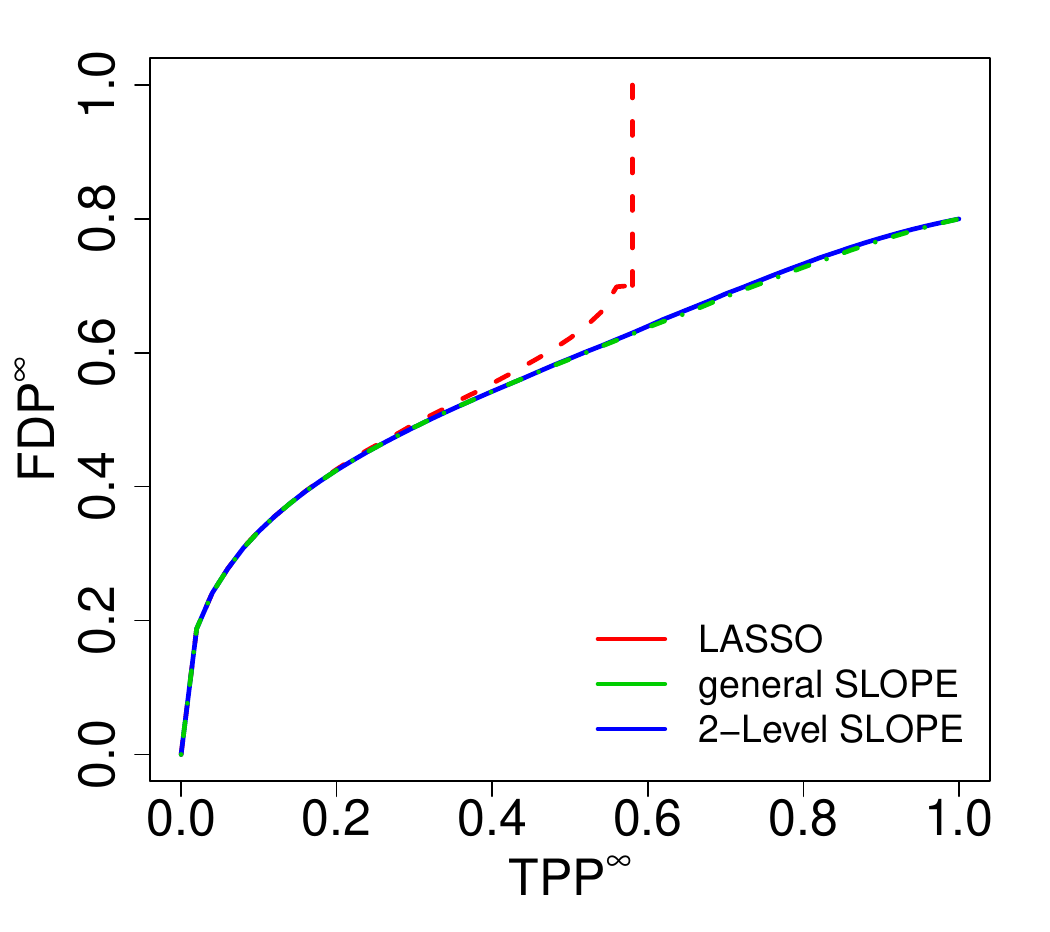}
    \includegraphics[width=0.32\textwidth]{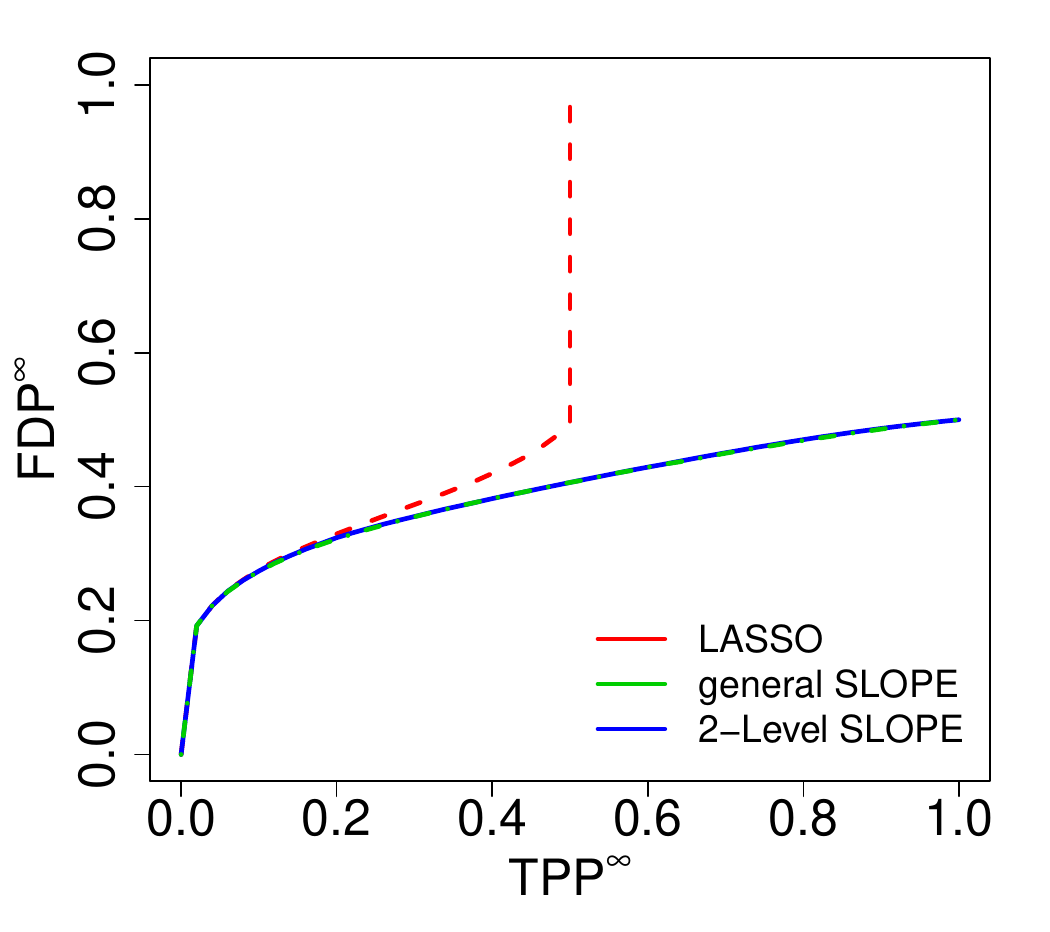}
        \includegraphics[width=0.32\textwidth]{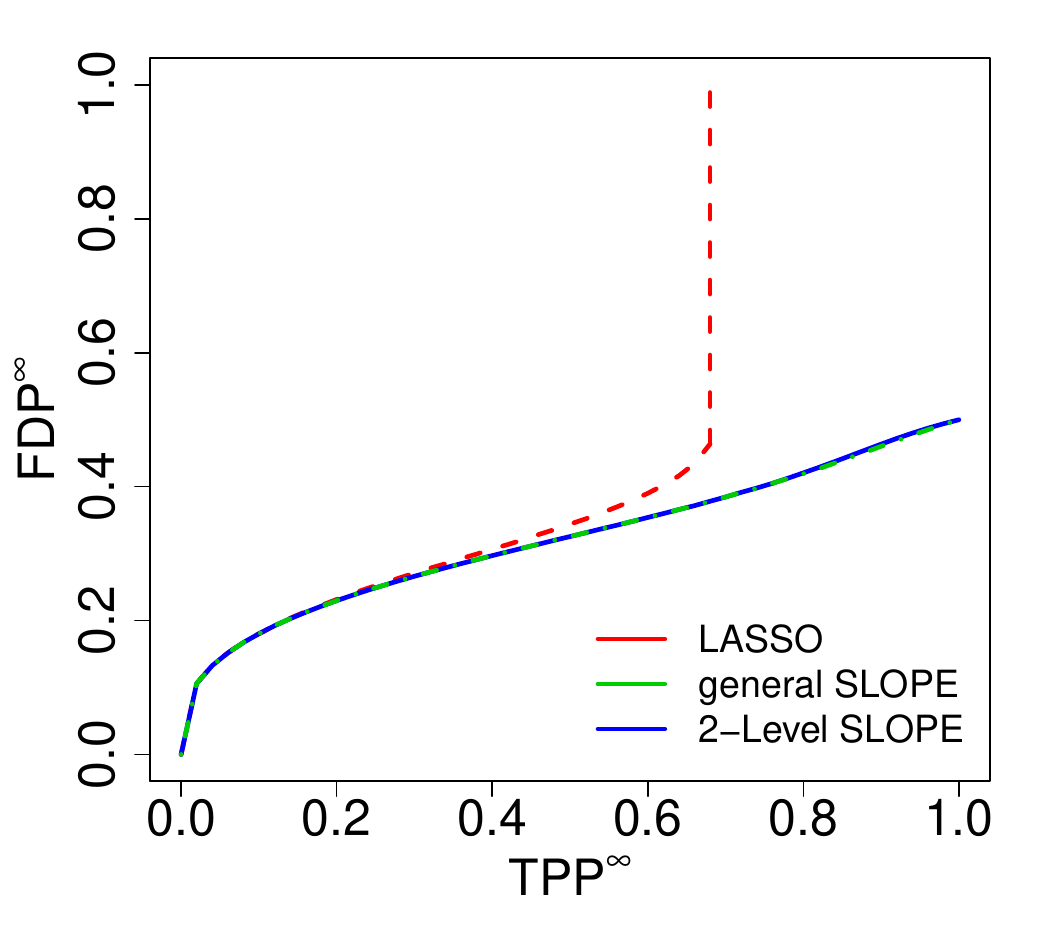}
    \includegraphics[width=0.32\textwidth]{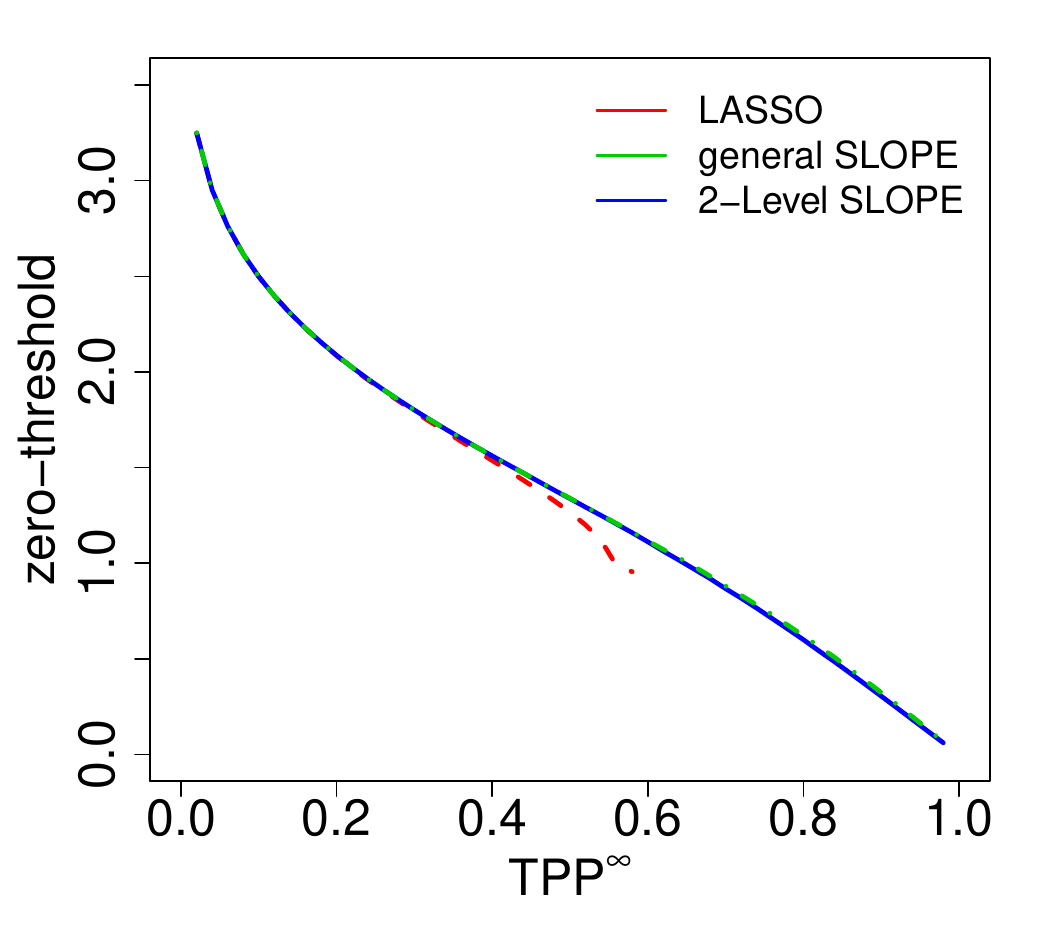}
    \includegraphics[width=0.32\textwidth]{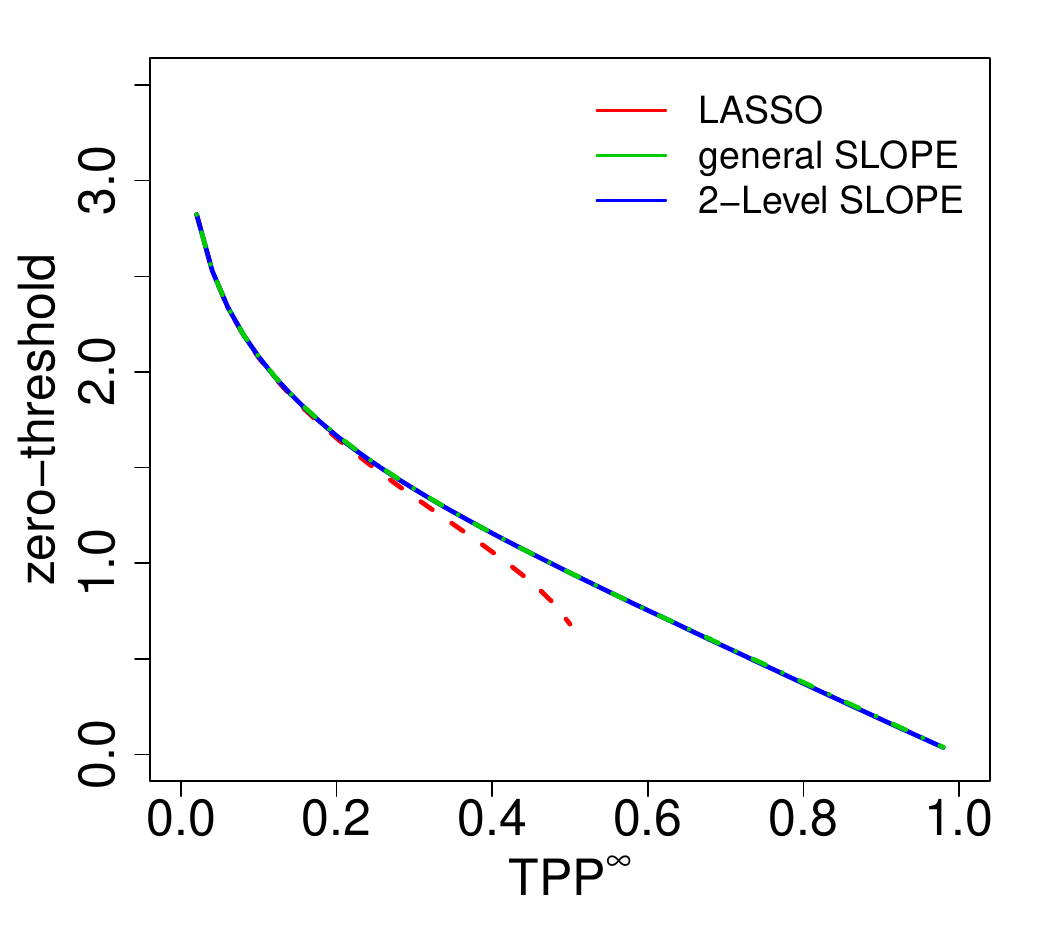}
    \includegraphics[width=0.32\textwidth]{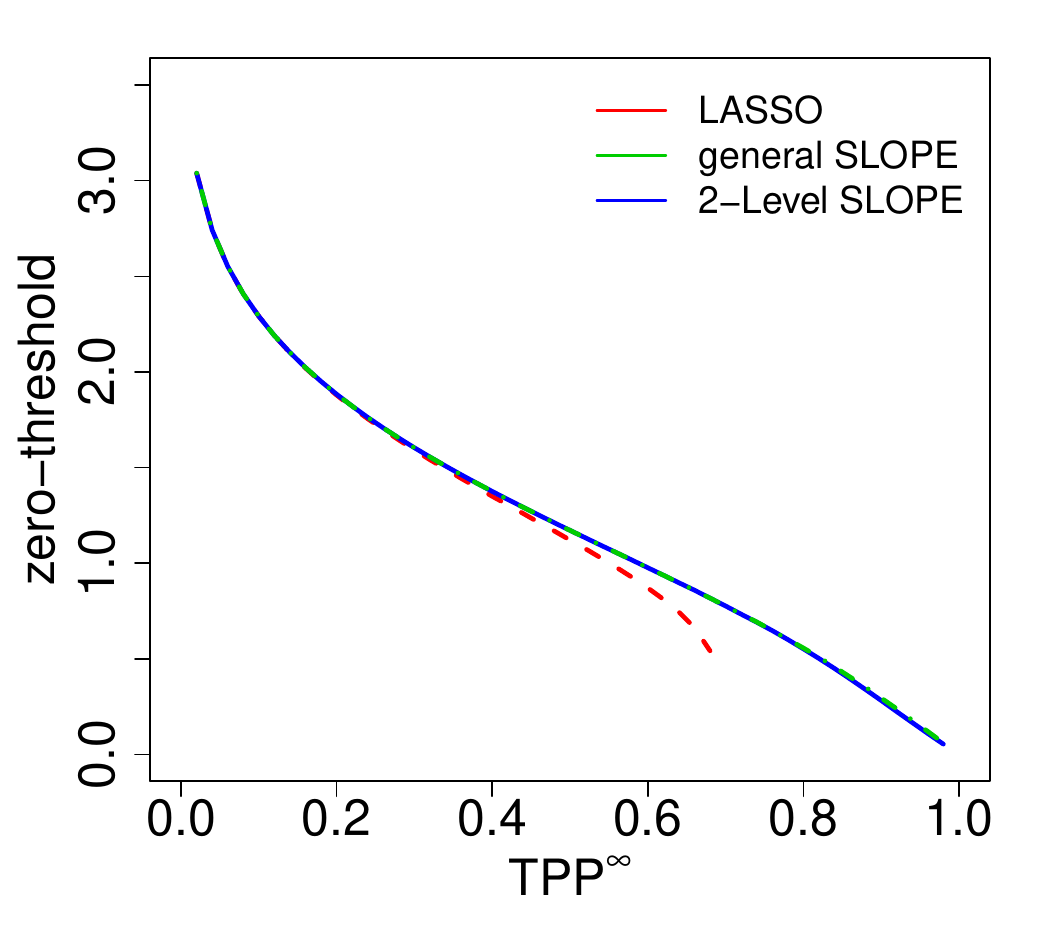}
    \includegraphics[width=0.32\textwidth]{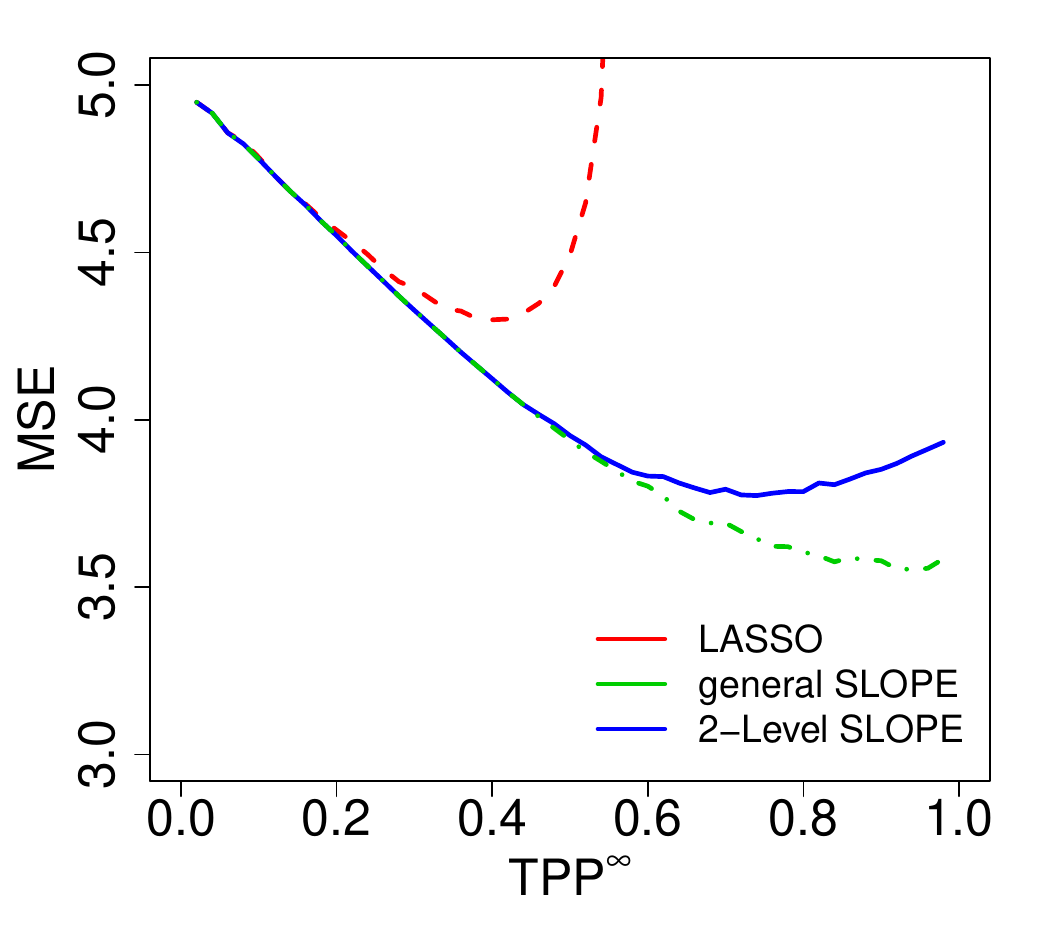}
    \includegraphics[width=0.32\textwidth]{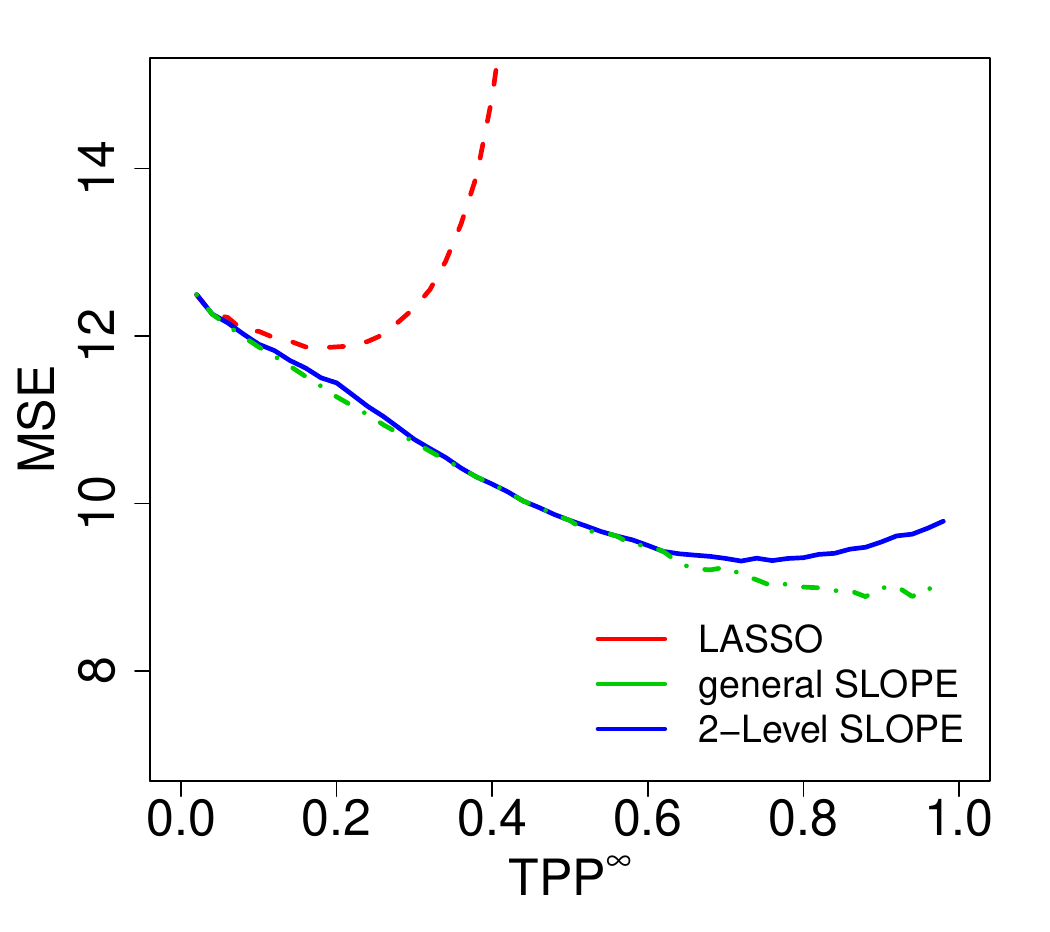}    \includegraphics[width=0.32\textwidth]{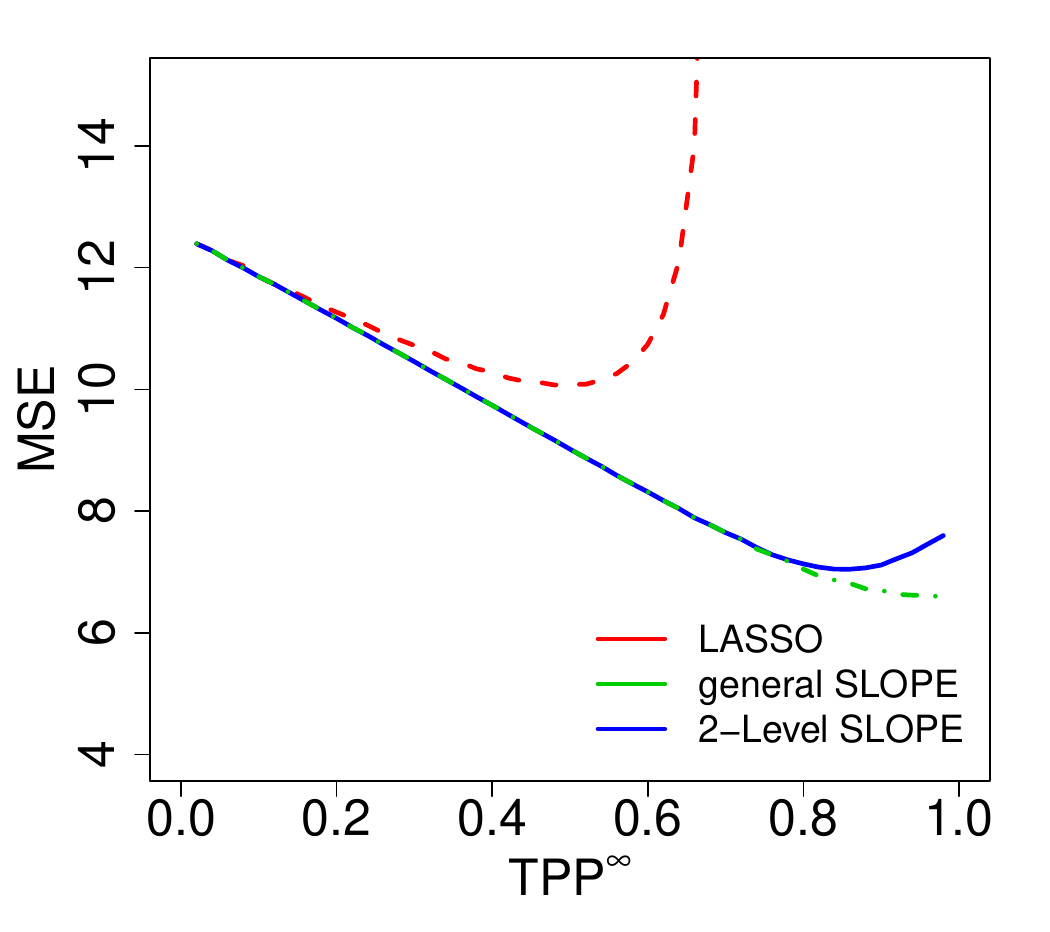}
   \caption{Examples of the $\TPP$-$\FDP$ trade-off curve, zero-threshold, and MSE of different SLOPE methods under a constant-or-nothing prior $\Pi^*=5$ for various $(\epsilon,\delta)$ settings. Left column: $\epsilon=0.2,\delta=0.3$. Middle column: $\epsilon=0.5,\delta=0.3$. Right column: $\epsilon=0.5,\delta=0.5$. 
   For the LASSO, we sweep through $\alpha_\text{LASSO}\in (0,\infty)$, each corresponding to one pair of $(\TPP,\FDP)$. The 2-level SLOPE $\FDP$ is given by \Cref{alg:fixed prior} and the general SLOPE $\FDP$ is computed by the analytic penalty in \eqref{eq:AHH}.
   }
    \label{fig:fixed prior TPPFDP}
\end{figure}

\subsection{Minimizing FDP given TPP when prior is constant-or-nothing} \label{sec:constant_or_nothing}
In this section, we consider a special family of fixed priors, namely, the constant-or-nothing priors, meaning $\Pi^*=T$ for some $T\in\R$, or in other words $\PP(\Pi=T)=\epsilon=1-\PP(\Pi=0)$. Equivalently, the normalized prior satisfies $\PP(\pi=t)=\epsilon=1-\PP(\pi=0)$, where $t \equiv t(\A):=T/\tau(\A)$. Thanks to the simple form of these priors, for each constant-or-nothing prior (i.e.,\ for each given $t$) we can derive the closed form for the SLOPE trade-off curves while following \Cref{alg:fixed prior}.

For the general $\p$-level SLOPE, we recall that the optimal SLOPE penalty is given in \eqref{eq:AHH}; moreover, in the case of constant-or-nothing prior, the format of $H(x)$ in \eqref{eq:AHH} simplifies to $t_1=t_2=t$: 
$$H(x;t)=4(1-\epsilon) \phi(x)+2 \epsilon\left[\phi(x-t)+\phi(x+t)\right] \Longrightarrow \Aeff \Longrightarrow\tau(\Aeff),$$
where $t$ is determined by $t(\Aeff)\equiv T/\tau(\Aeff)$. In this way, although the form of $q_\text{SLOPE}$ is not explicit, we can derive the explicit form of $q_\text{SLOPE}^{(\Pi)}$ if $\Pi$ is a constant-or-nothing prior and if \eqref{eq:AHH} is monotonically increasing under some $(\epsilon,\delta)$.

For 2-level SLOPE, including the LASSO, in the case of the constant-or-nothing prior, we translate the constraint on the $\TPP$ in \eqref{eq:tpp fixed prior} into
\begin{align}
\PP\left(|t(\A)+Z|>\alpha_2(\A)\right)=u  \Longrightarrow  1-\Phi(\alpha_2(\A)-t(\A))+\Phi(-\alpha_2(\A)-t(\A))=u,
\label{eq:s2}
\end{align}
and translate the constraint given by the state evolution \eqref{eq:state evolution} into
\begin{align}
\delta\geq F\left[\pi_{\min}(t,t;u)\right]=(1-\epsilon)\mathcal{E}(0)+\epsilon\mathcal{E}(t)
\label{eq:s3}    
\end{align}
following the discussion between \eqref{eq:state evolution constraint}-\eqref{eq:alpha star explicit}, where $\mathcal{E}(t)$ is defined in \eqref{eq:mathcal E}. Consequently, we modify \Cref{alg:fixed prior} to \Cref{alg:fixed prior simpli} for the constant-or-nothing priors, whereas the latter is more precise because the state evolution constraint is verified analytically in \eqref{eq:s3} but only numerically in \eqref{eq:state evolution}.

\begin{algorithm}[!htb]
\caption{Minimizing $\FDP$ at $\TPP=u$ over 2-level SLOPE penalty with constant-or-nothing prior $\Pi^*=T$}
\begin{algorithmic}[1]
\State Find LASSO zero-threshold $\alpha_\text{LASSO}(u)$ such that $\PP\left(|\Pi^*/\tau(\alpha_\text{LASSO})+Z|>\alpha_\text{LASSO}\right)=u$ by \eqref{eq:tpp fixed prior}
\State Find the maximum zero-threshold for the general SLOPE $\alpha^*_\text{SLOPE}(u)$ 
\For{$\alpha_2\in [\alpha_\text{LASSO}(u),\alpha^*_\text{SLOPE}(u)]$}
\State Determine $t$ by \eqref{eq:s2}
\For{$\alpha_1\in [\alpha_2,\infty)$}
\For{$s\in[0,1]$}
\State Determine the effective 2-level SLOPE penalty $\Aeff$ and compute $\mathcal{E}(t)$
\If{$\delta\geq (1-\epsilon)\mathcal{E}(0)+\epsilon\mathcal{E}(t)$ in \eqref{eq:s3}}
\State $\alpha_2$ is feasible, break the loop and move to larger $\alpha_2$
\EndIf
\EndFor
\EndFor
\EndFor
\State Output: the maximum feasible zero-threshold $\alpha_2$ and the minimum $\FDP$ by \eqref{eq:fdp is function of tpp}
\end{algorithmic}
\label{alg:fixed prior simpli}
\end{algorithm}

Leveraging \Cref{alg:fixed prior simpli}, we compare the $\FDP$ and MSE of the general SLOPE, 2-level SLOPE, and LASSO in \Cref{fig:fixed prior vary}, varying $\epsilon$ or $\delta$. Generally speaking, 2-level SLOPE achieves similar $\FDP$ and MSE to the general SLOPE, especially outperforming the LASSO when $\delta$ is small or $\epsilon$ is large.

\begin{figure}[!htb]
    \centering
                    \includegraphics[width=0.4\textwidth]{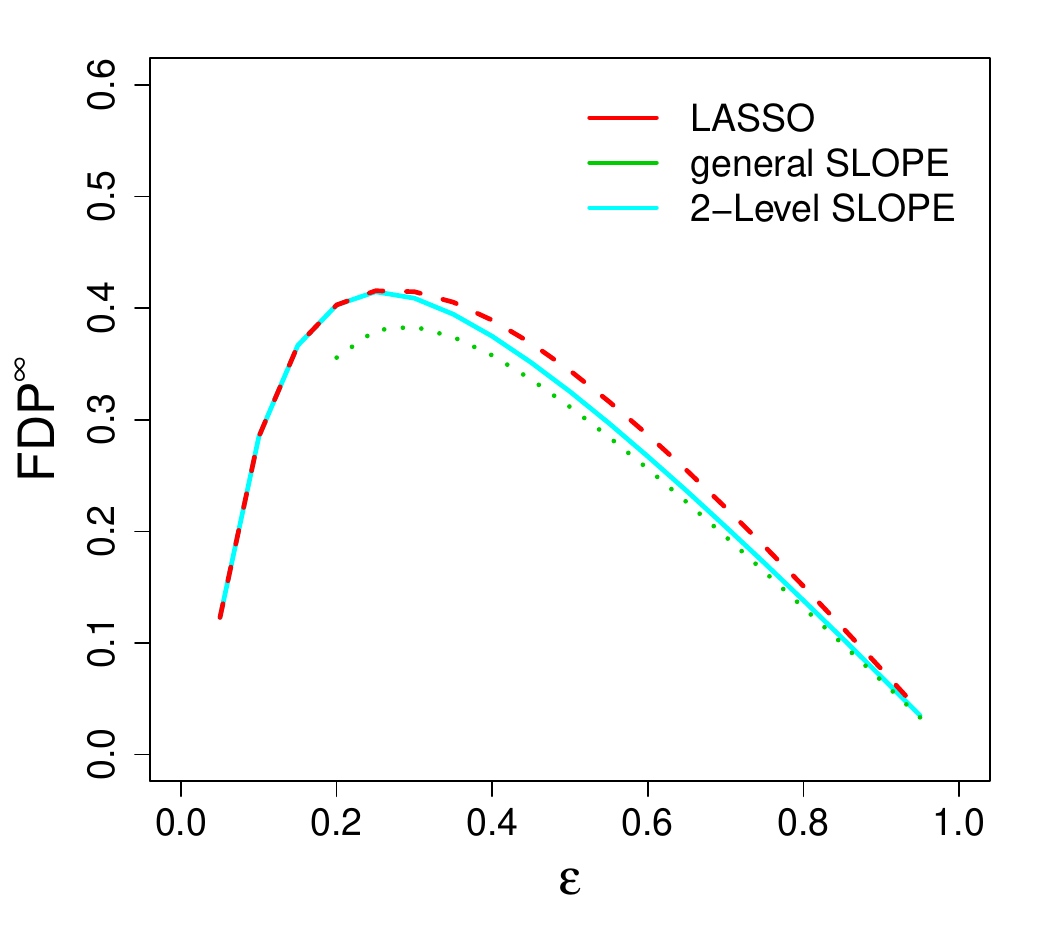}
        \includegraphics[width=0.4\textwidth]{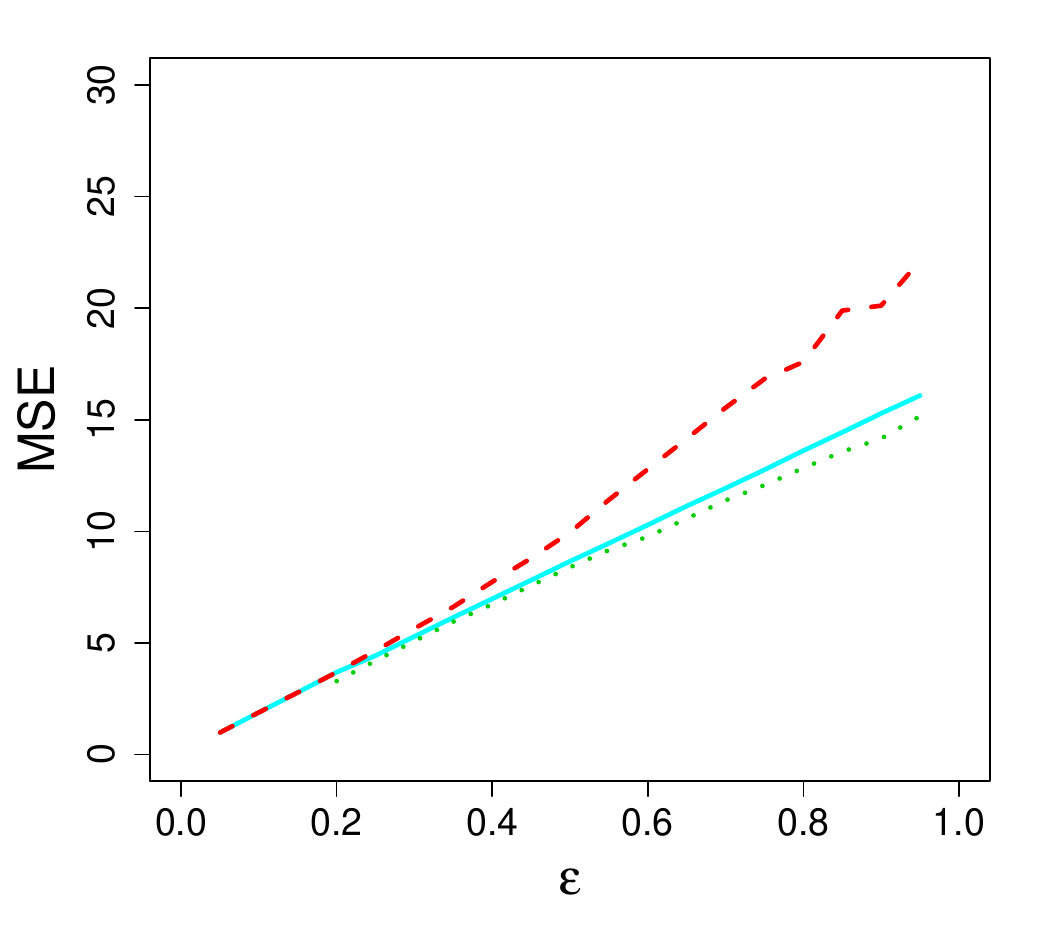}
                            \includegraphics[width=0.4\textwidth]{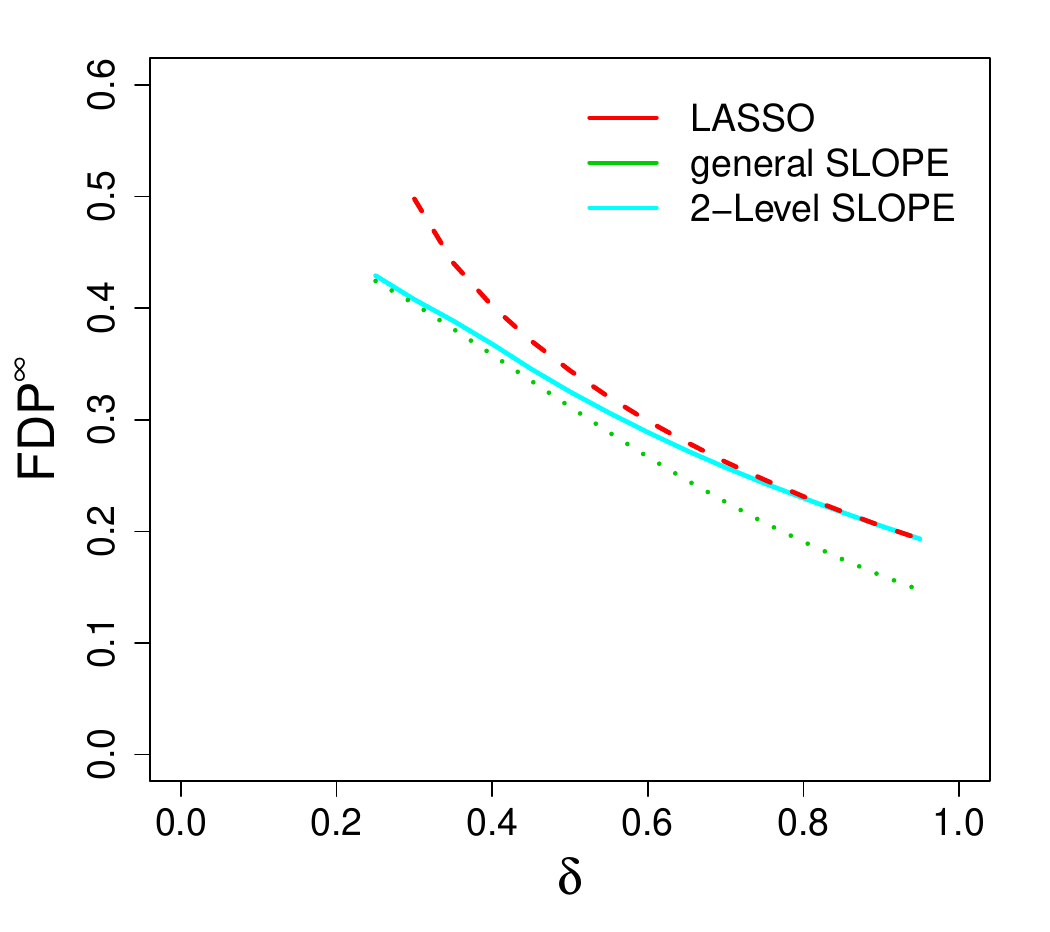}
        \includegraphics[width=0.4\textwidth]{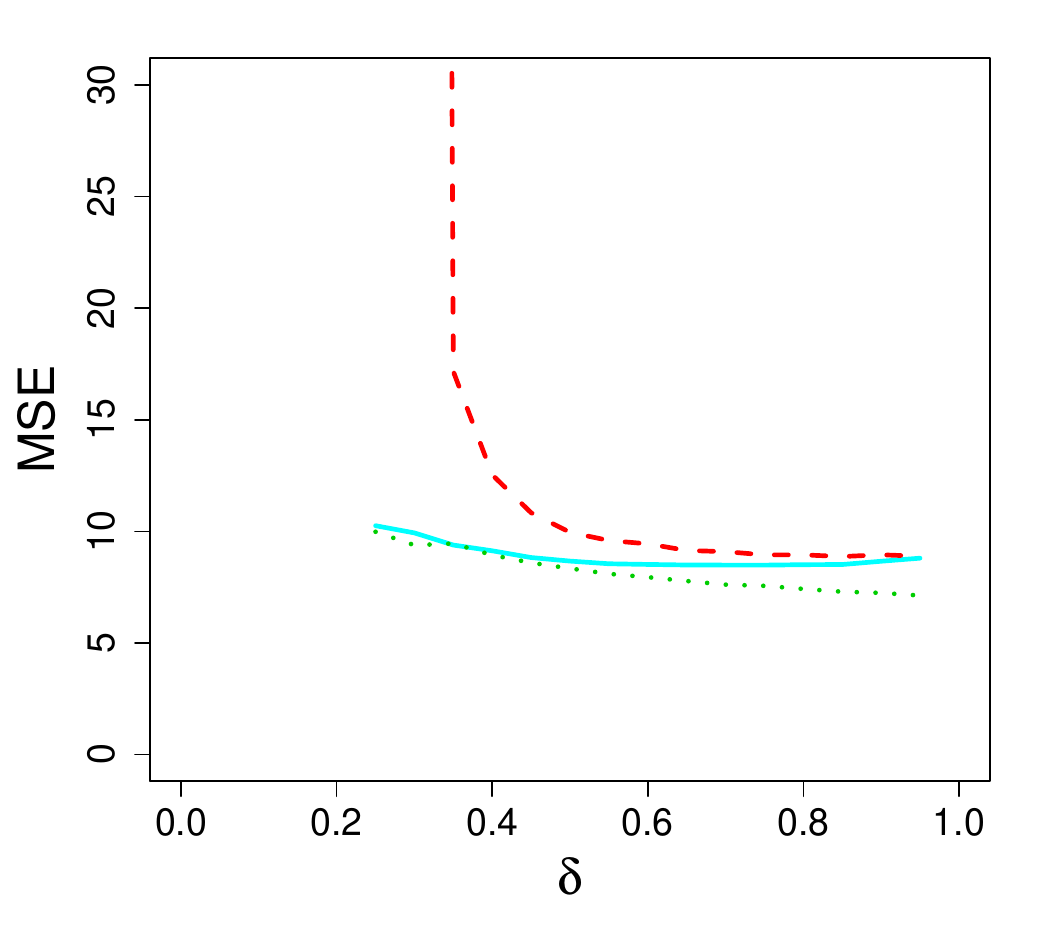}
    \caption{Examples of $\FDP$ and MSE of different SLOPE methods, under a constant-or-nothing prior $\Pi^*=5$ and fixed $\TPP=0.5$. We vary one of $(\epsilon,\delta)$ and fix the other at $\epsilon=0.5$ or $\delta=0.5$. 
    Here the general SLOPE trade-off $q_\text{SLOPE}^{(\Pi)}$ is computed by the analytic penalty in \eqref{eq:AHH}, which may not exist for sufficiently small $\delta$. That is, while $q_\text{SLOPE}^{(\Pi)}(\TPP=0.5)$ does exist for all $(\epsilon,\delta)$ settings, its form may not always be captured by \eqref{eq:AHH}.
    }
    \label{fig:fixed prior vary}
\end{figure}

\newpage
\section{Estimation error of 2-level SLOPE under a fixed prior}
\label{sec:error fixed priors}
In this section, we focus on minimizing the estimation error $\text{MSE}(\Lambda)=\plim_\p\frac{1}{\p}\|\widehat\bet(\Lambda)-\bet\|^2$ over SLOPE penalty distributions $\Lambda$ for a general, fixed prior $\Pi$.



Notice that this is different from \Cref{sec:TPP FDP fixed priors} and \Cref{sec:constant_or_nothing}, where our main goal is to minimize $\FDP$ for every given $\TPP=u$, and the fact that 2-level SLOPE can achieve a smaller MSE than LASSO is a by-product. In \Cref{fig:arian}, our main goal is to minimize the MSE without the constraints on $\TPP$ and $\FDP$, for which we have employed \Cref{alg:MSE}. 

\begin{algorithm}[!htb]
\caption{Minimizing MSE  over 2-level SLOPE penalty with fixed prior $\Pi$}
\begin{algorithmic}[1]
\State Initialize a small $\gamma>0$ and $M_{\min}=\infty$
\For{$a\in [1,\infty)$}
\For{$s\in[0,1]$}
\State Construct $\blam=\langle\gamma a,\gamma;s\rangle
$ by \eqref{eq:K-level SLOPE seq} and solve the 2-level SLOPE problem \eqref{eq:SLOPE}
\State Compute the MSE$(\widehat\bet)$ of the SLOPE solution
\If{MSE$(\widehat\bet)>M$} 
    \State Break the loops and stop
\Else
    \State Set $M_{\min}=\text{MSE}(\widehat\bet)$ and increase $\gamma$
\EndIf
\EndFor
\EndFor
\State Output: the minimum MSE $M_{\min}$
\end{algorithmic}
\label{alg:MSE}
\end{algorithm}

Following the setting in \cite{wang2022does}, we consider
\begin{itemize}
    \item Data: $\X\in\R^{n\times \p}$ is i.i.d. Gaussian $\X_{iid}=\mathcal{N}(0,1/n\mathbf{I}_n)$, or heavy-tailed $\X_{ht}$ which is t-distributed with 3 degrees of freedom, normalized by $\sqrt{3n}$. Furthermore, we consider the corelation of data by multiplying $\X_{iid}$ and $\X_{ht}$ with $\Sigma(\rho)^{1/2}$, where $\Sigma$ is the auto-correlation matrix such that $\Sigma_{i,j}=\rho^{|i-j|}$ for $\rho=0.8$ (with correlation) or $\rho=0$ (without correlation). In short, we have four choices of $\X$: $\X_{iid}\Sigma(0)^{1/2}$, $\X_{iid}\Sigma(0.8)^{1/2}$, $\X_{ht}\Sigma(0)^{1/2}$, $\X_{ht}\Sigma(0.8)^{1/2}$.
    \item Noise: $\w\sim\mathcal{N}(0,\sigma^2\mathbf{I}_n)$.
    \item Prior: $\bet\sim\Pi$ where $\PP(\Pi\neq 0)=1-\epsilon$ and $\epsilon=0.5$. The non-zero component $\Pi^*$ is i.i.d. Uniform[0, 5], or the constant 5 so that $\Pi$ is Bernoulli. We term the uniform $\Pi^*$ as the non-tied prior, and the constant one as the tied prior.
    \item Dimension: Here $\p=500, \delta=0.9, n=\delta \p$.
    \item SLOPE penalty: LASSO uses a scalar penalty $\lambda$. Our 2-level SLOPE uses a triplet $(\lambda_1,\lambda_2,s)$. The uniform SLOPE (SLOPE:unif) uses $\lambda_i=\gamma\cdot (1-0.99(i-1)/\p)$, and the Benjamini–Hochberg SLOPE (SLOPE:BH) uses $\lambda_i=\gamma\cdot \Phi^{-1}(1-iq/2\p)/\Phi^{-1}(1-q/2\p)$, where $q=0.5$ and $\gamma$ is a regularization parameter introduced in Equation 2.2 of \cite{wang2022does}, which should be tuned.
    \item Computation: For each setting, we run 10 independent simulations and report the mean and the standard deviation. We note that searching for the optimal 2-level SLOPE penalty takes more time than the other SLOPE penalty, as 2-level SLOPE has 3 degrees of freedom (in $\lambda_1,\lambda_2,s$) while the others have only 1 degree of freedom. Technically speaking, SLOPE:BH has 2 degrees of freedom (in $q$ and $\gamma$) but $q$ is fixed in \cite{wang2022does}. We highlight that all SLOPE problems are solved efficiently by the fast proximal gradient descent \cite{su2016slope}.
\end{itemize}

\begin{figure}[!htb]
 \begin{subfigure}{0.48\textwidth}
\includegraphics[width=0.48\textwidth]{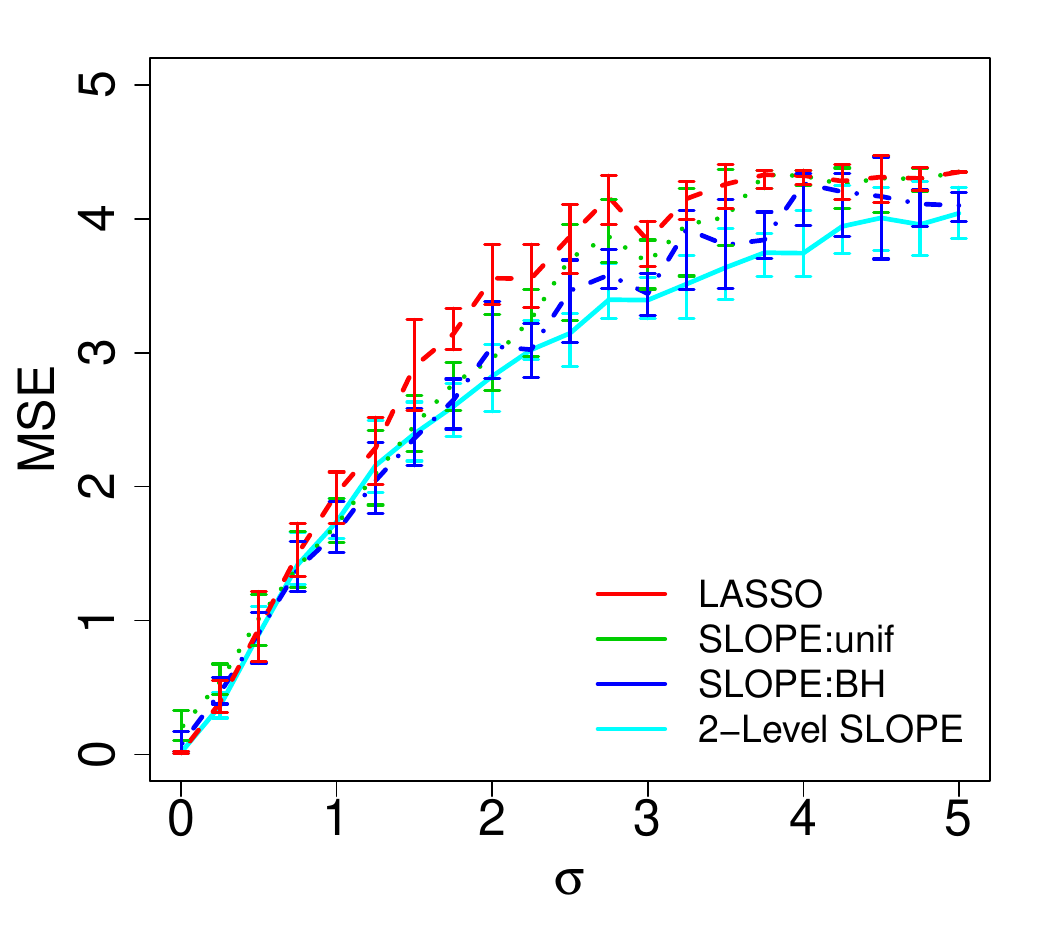}
\includegraphics[width=0.48\textwidth]{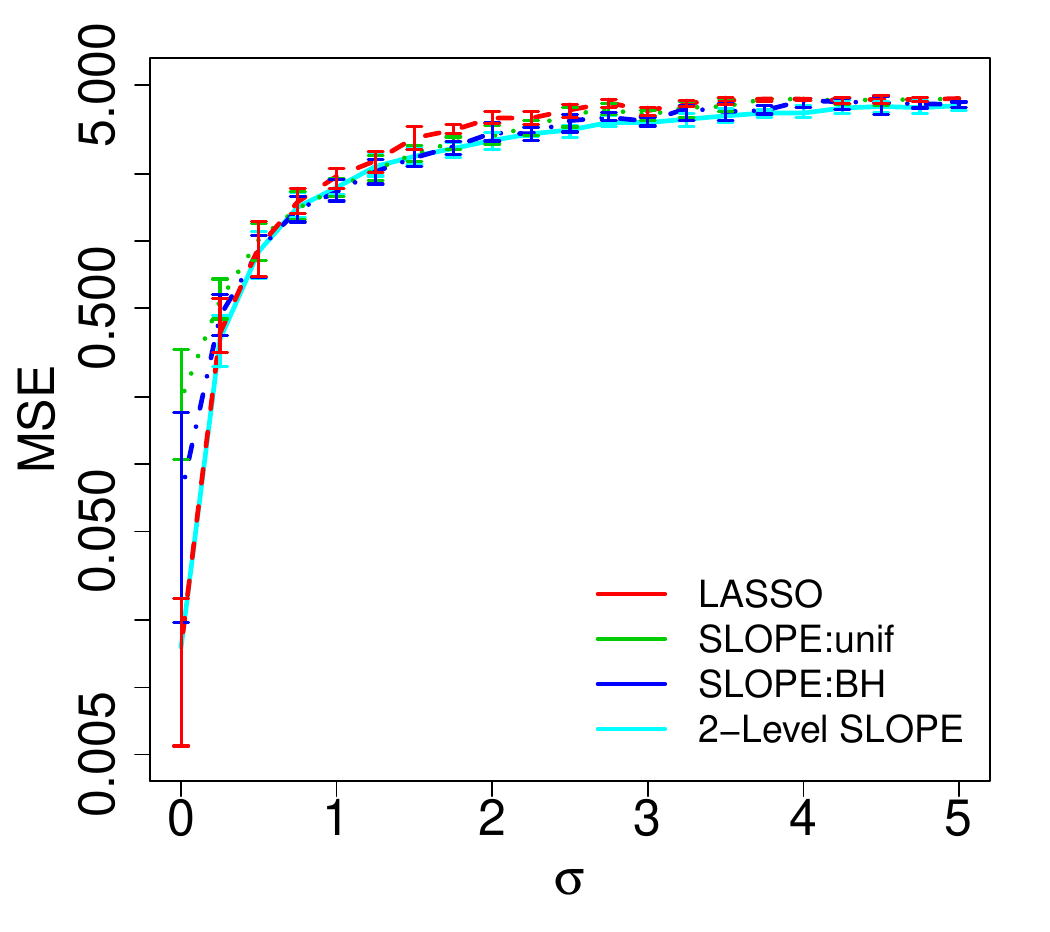}
\caption{Gaussian + iid $\X$, non-tied $\bet$}  
\label{fig:a}
 \end{subfigure}
 \begin{subfigure}{0.48\textwidth}
 \includegraphics[width=0.48\textwidth]{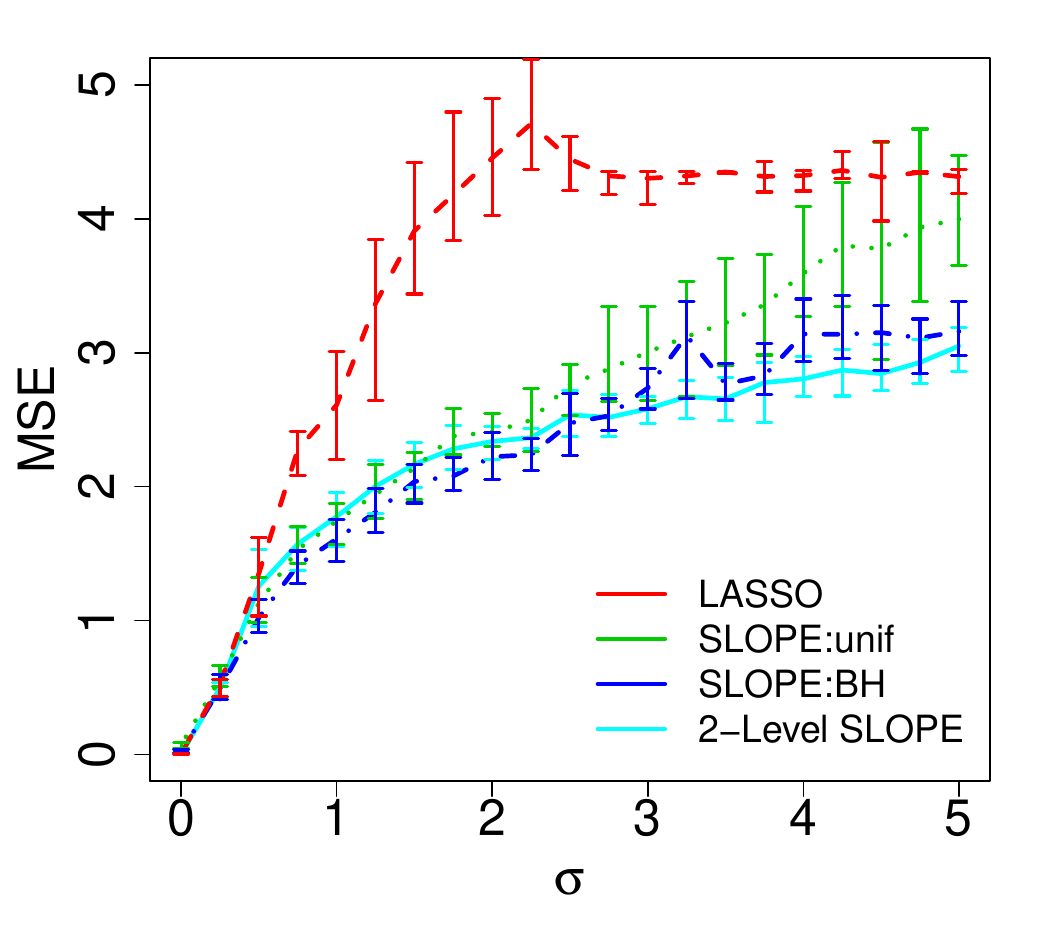}
 \includegraphics[width=0.48\textwidth]{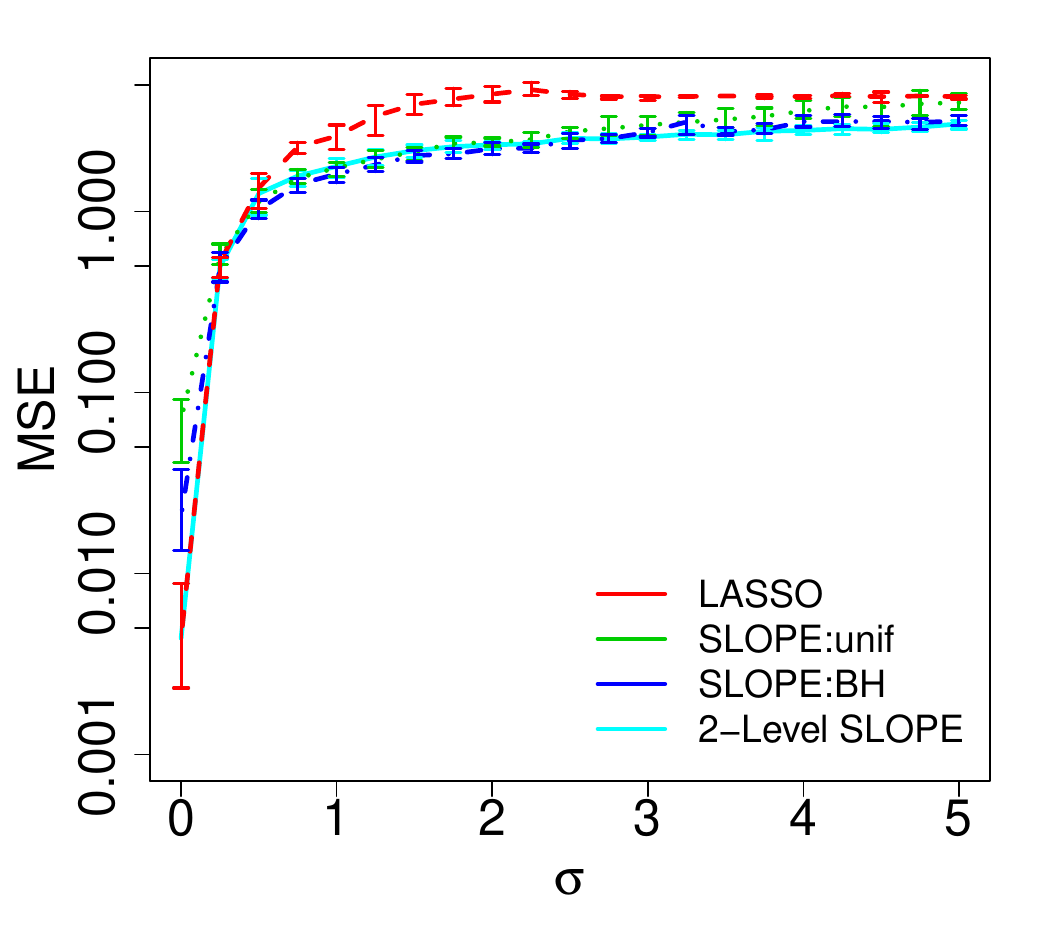}
\caption{Gaussian + correlated $\X$, non-tied $\bet$}  
\label{fig:b}
 \end{subfigure}
\\
 \begin{subfigure}{0.48\textwidth}
 \includegraphics[width=0.48\textwidth]{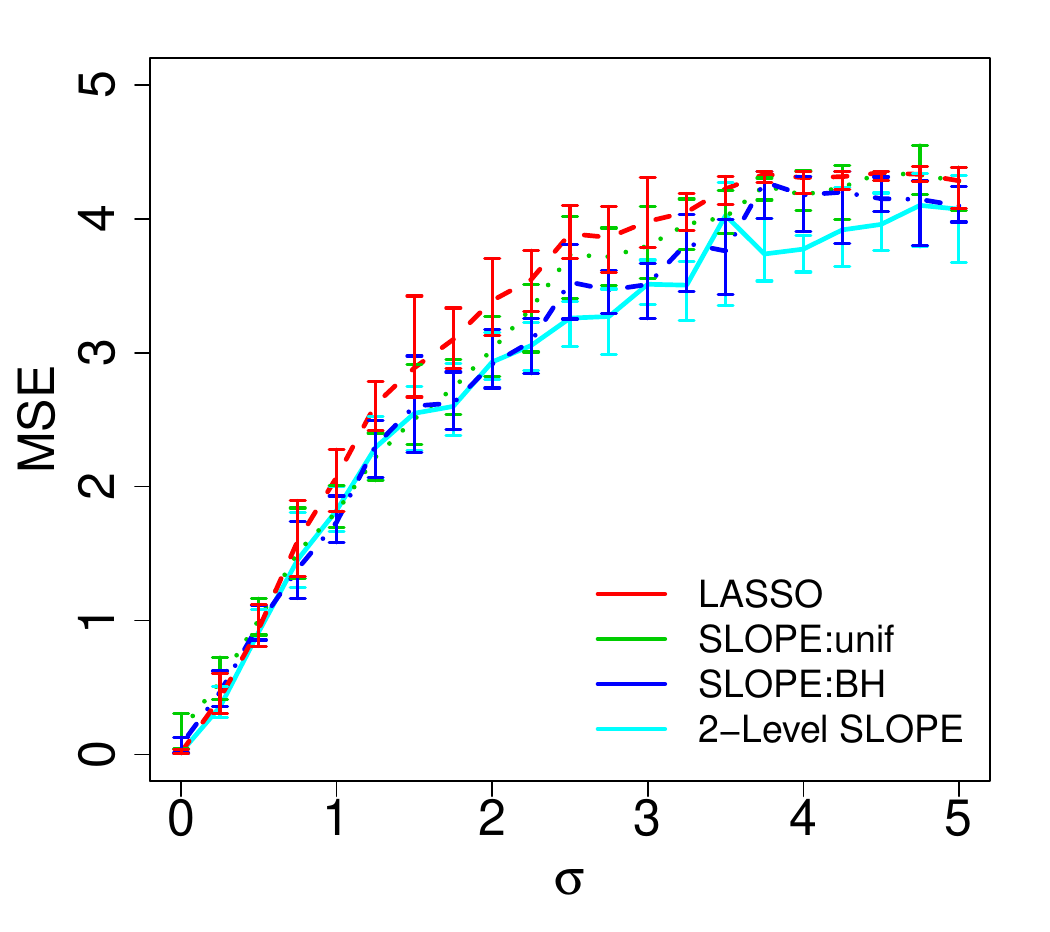}
 \includegraphics[width=0.48\textwidth]{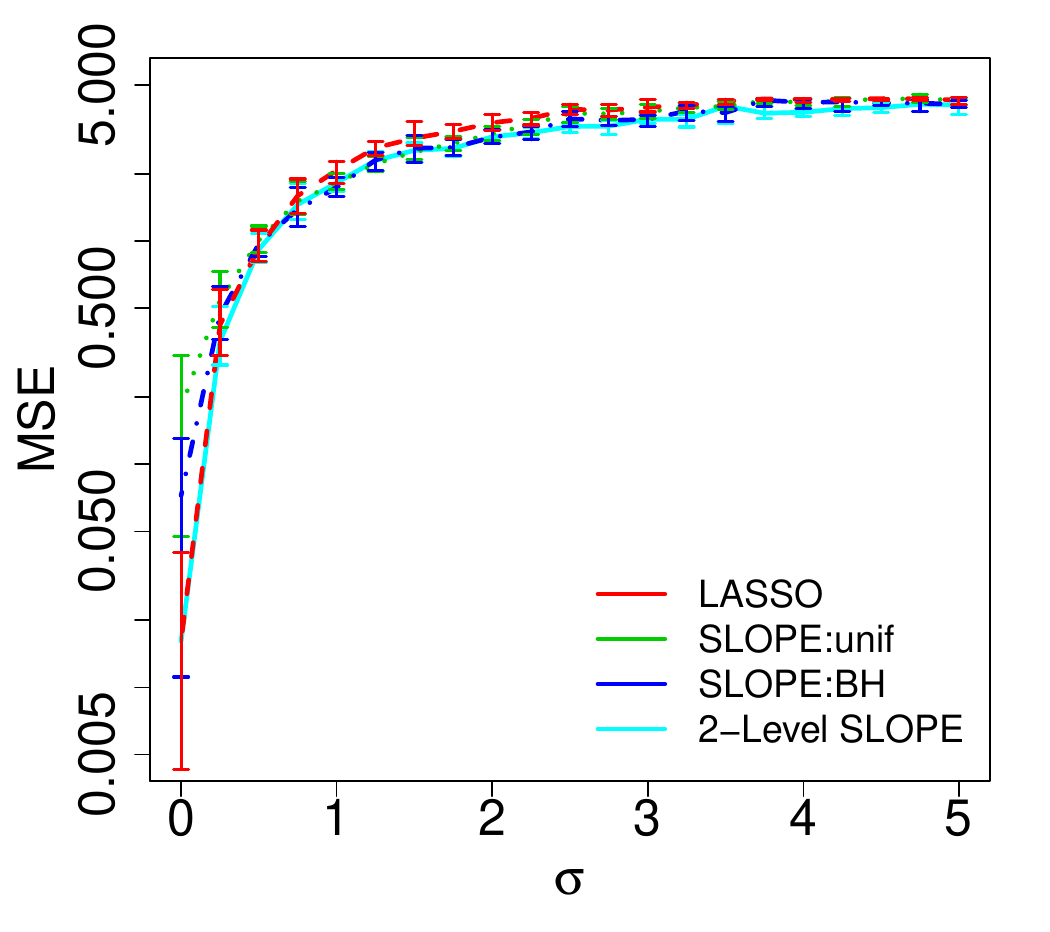}
\caption{heavy-tail + iid $\X$, non-tied $\bet$}  
\label{fig:c}
 \end{subfigure}
 \begin{subfigure}{0.48\textwidth}
     \includegraphics[width=0.48\textwidth]{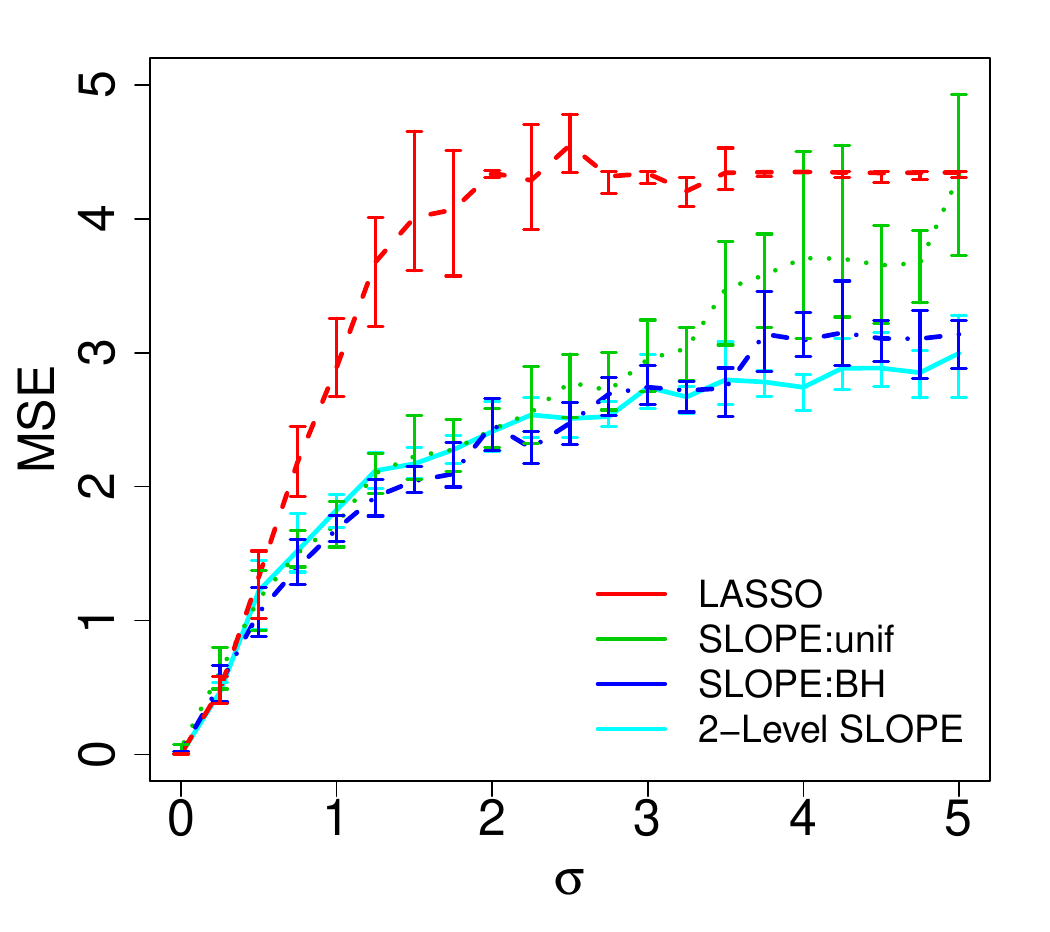}
     \includegraphics[width=0.48\textwidth]{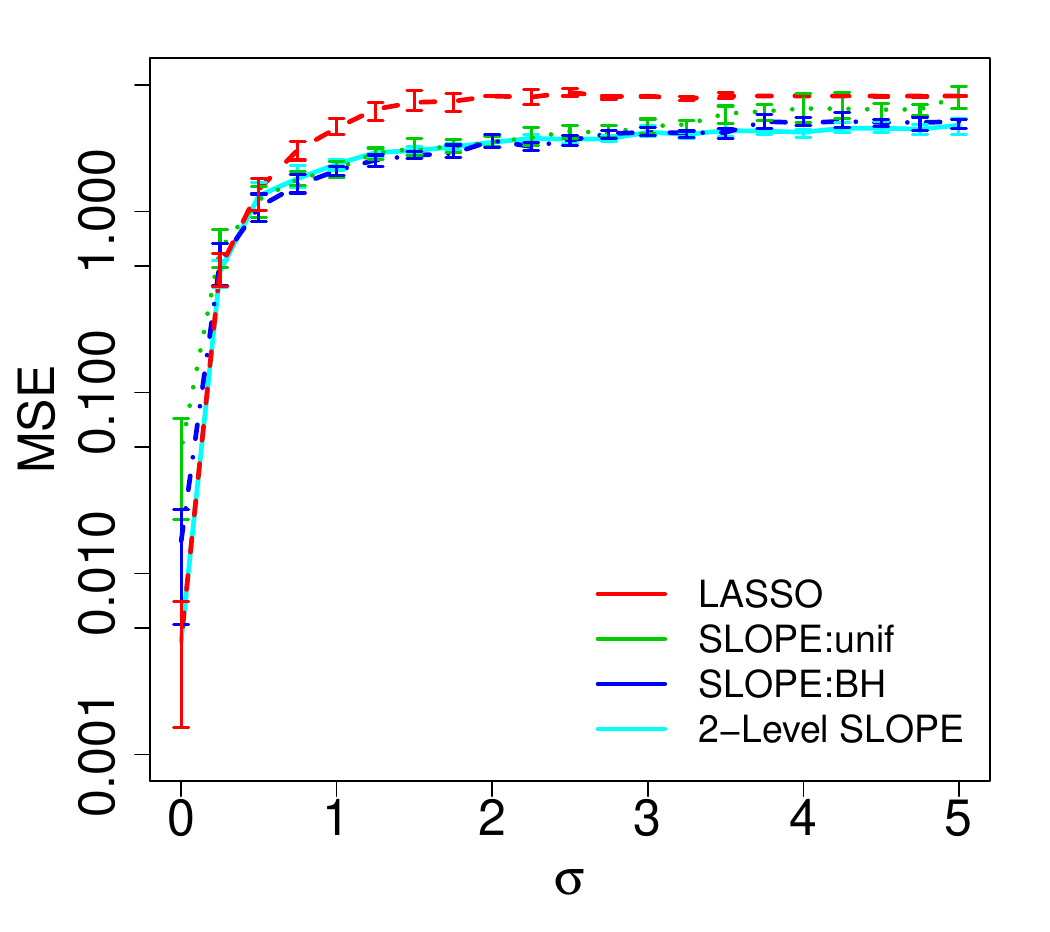}\caption{heavy-tail + correlated $\X$, non-tied $\bet$}  
     \label{fig:d}
 \end{subfigure}
\\
 \begin{subfigure}{0.48\textwidth}
\includegraphics[width=0.48\textwidth]{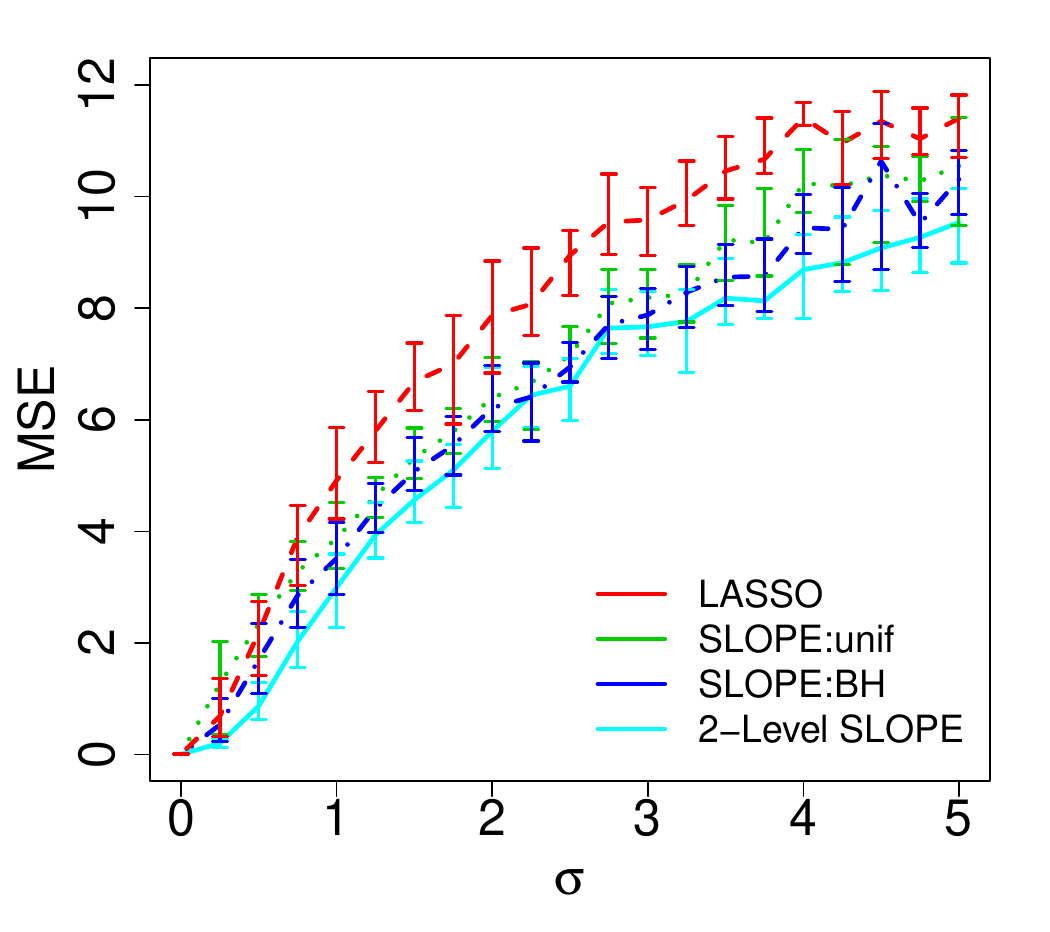}
\includegraphics[width=0.48\textwidth]{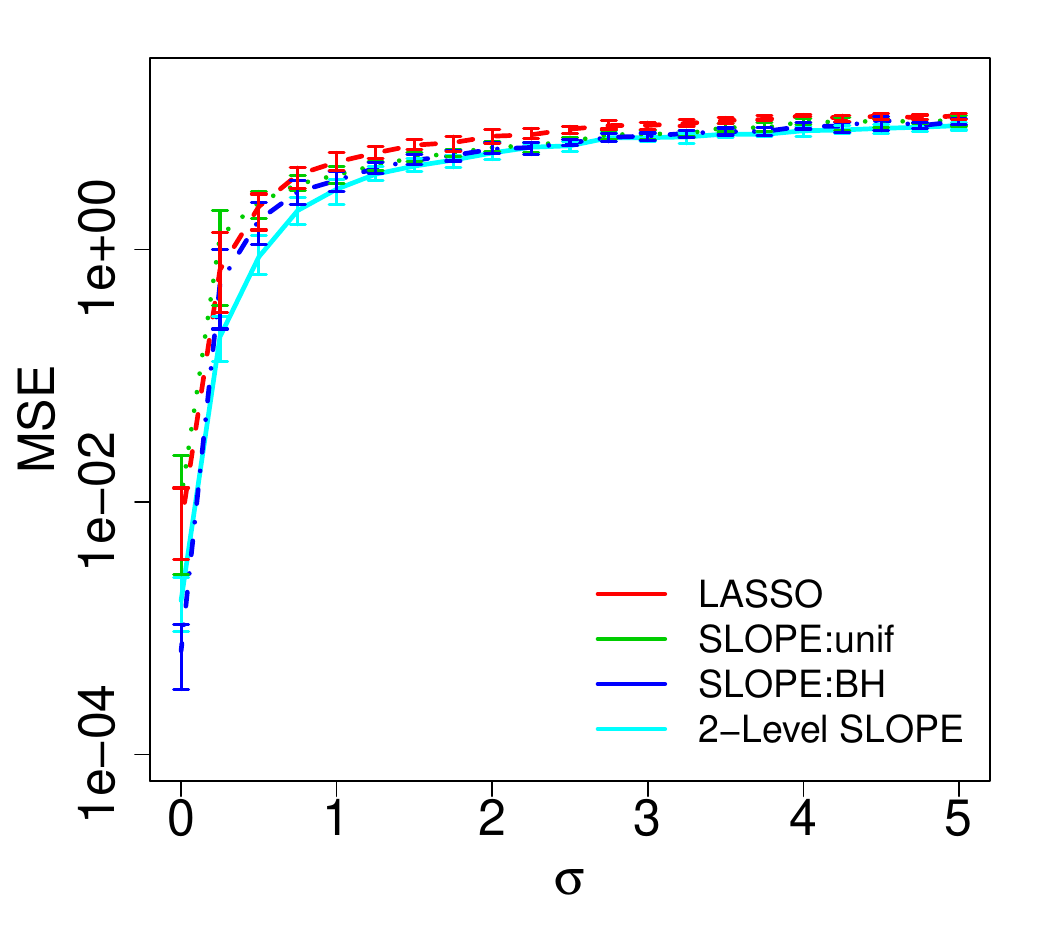}
\caption{Gaussian + iid $\X$, tied $\bet$}     
\label{fig:i}
 \end{subfigure}
 \begin{subfigure}{0.48\textwidth}
     \includegraphics[width=0.48\textwidth]{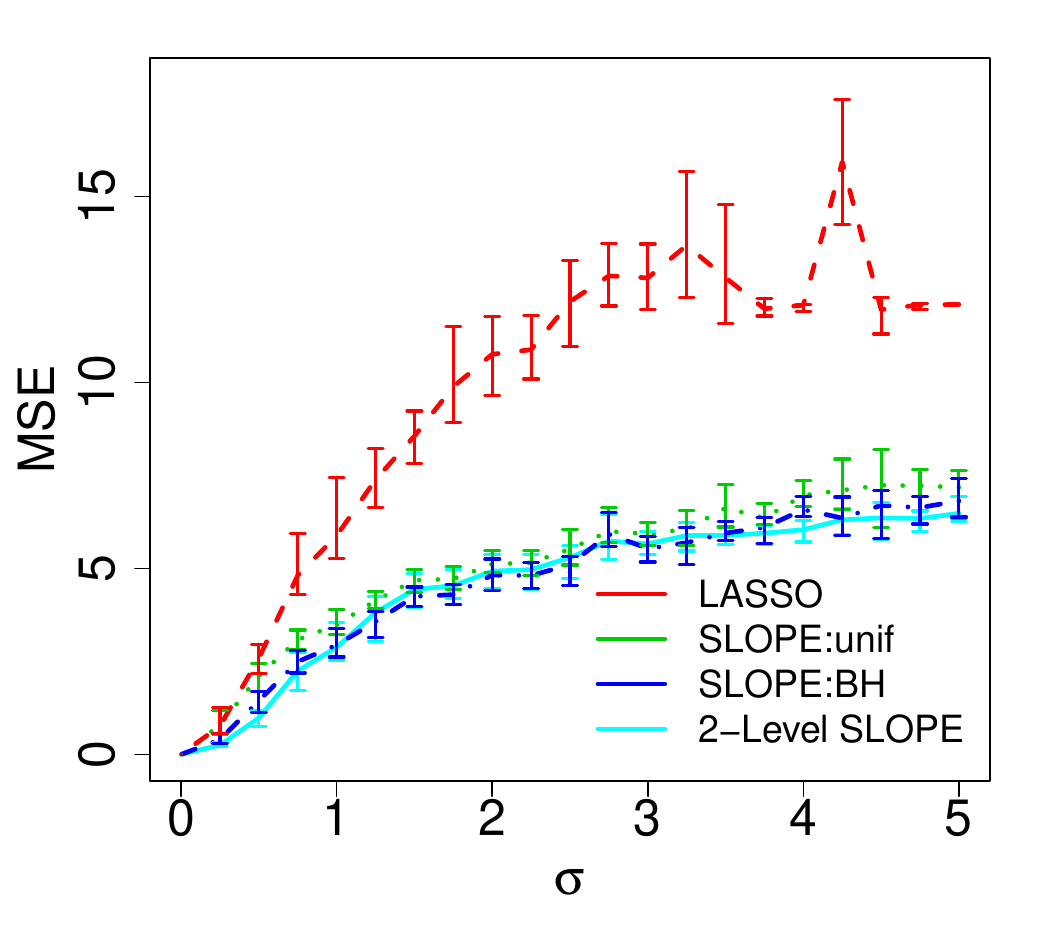}
     \includegraphics[width=0.48\textwidth]{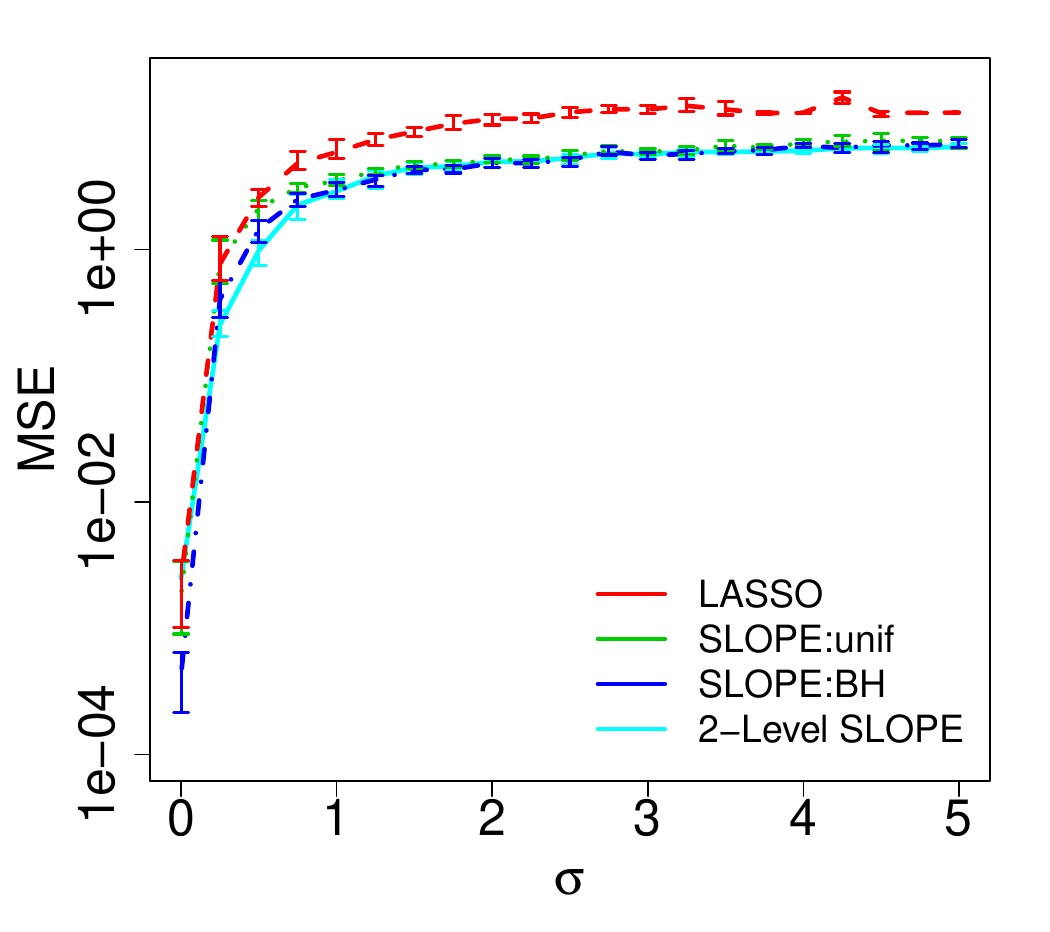}
\caption{Gaussian + correlated $\X$, tied $\bet$}  
\label{fig:j}
 \end{subfigure}
 \begin{subfigure}{0.48\textwidth}
 \includegraphics[width=0.48\textwidth]{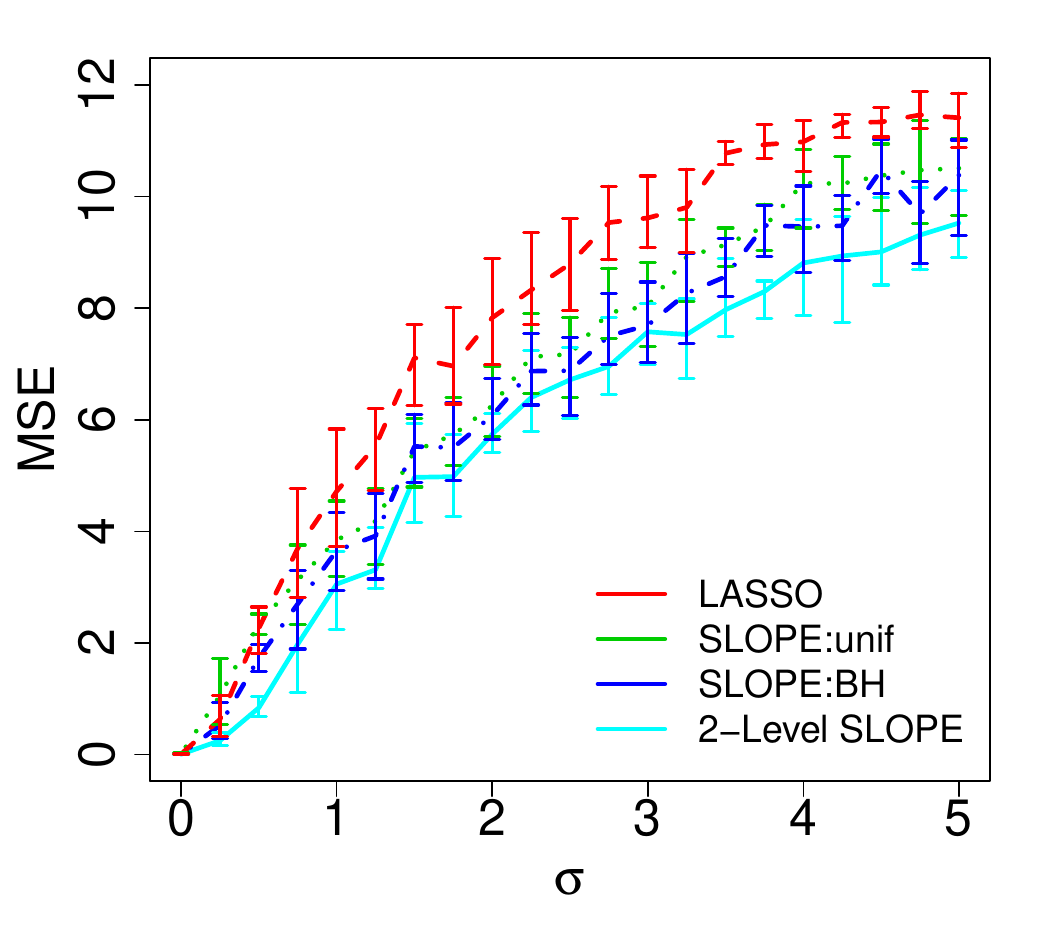}
\includegraphics[width=0.48\textwidth]{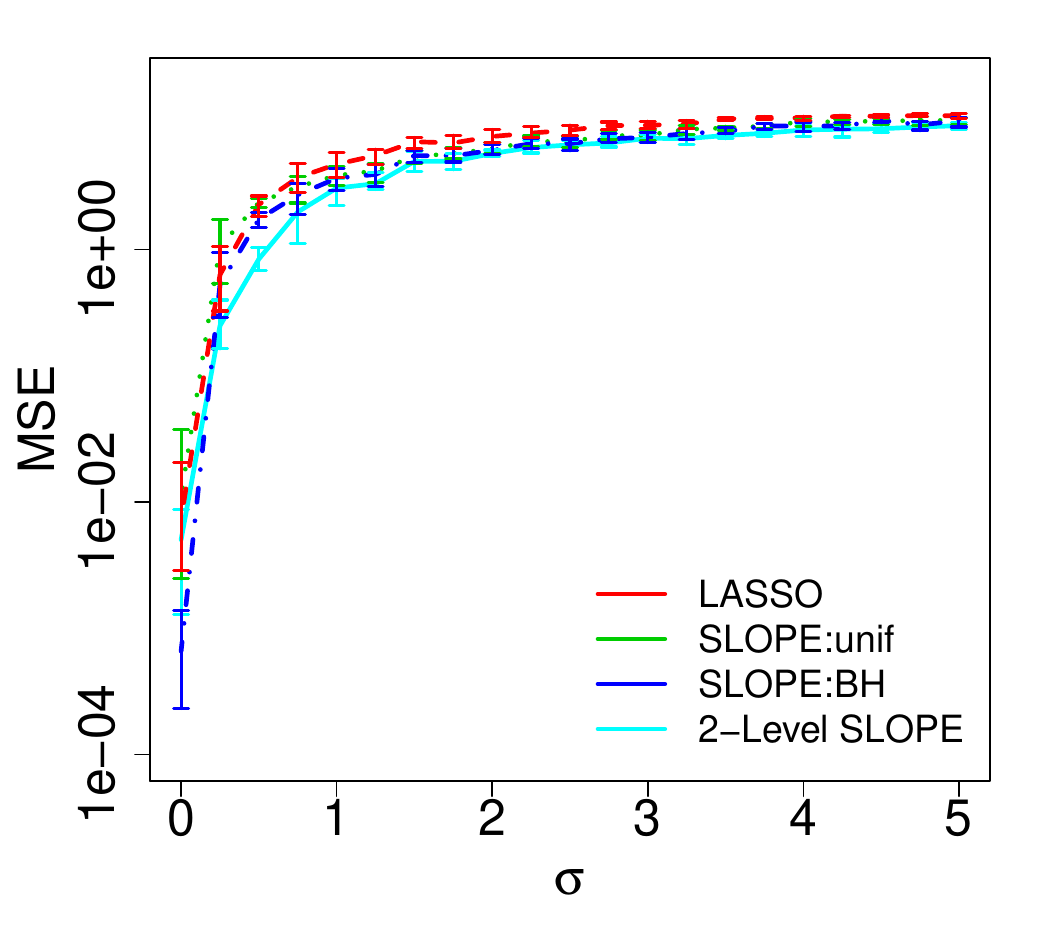}
\caption{heavy-tail + iid $\X$, tied $\bet$}  
\label{fig:k}
 \end{subfigure}
 \begin{subfigure}{0.48\textwidth}
     \includegraphics[width=0.48\textwidth]{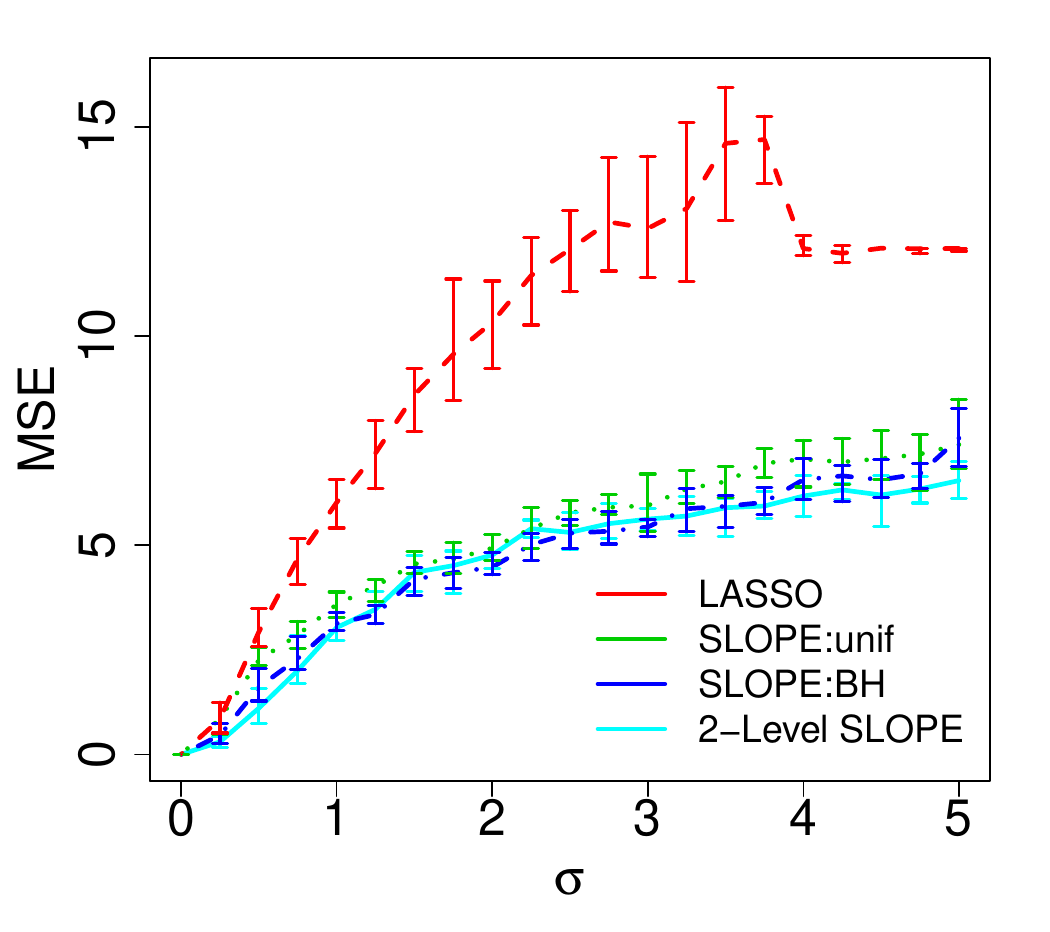}
     \includegraphics[width=0.48\textwidth]{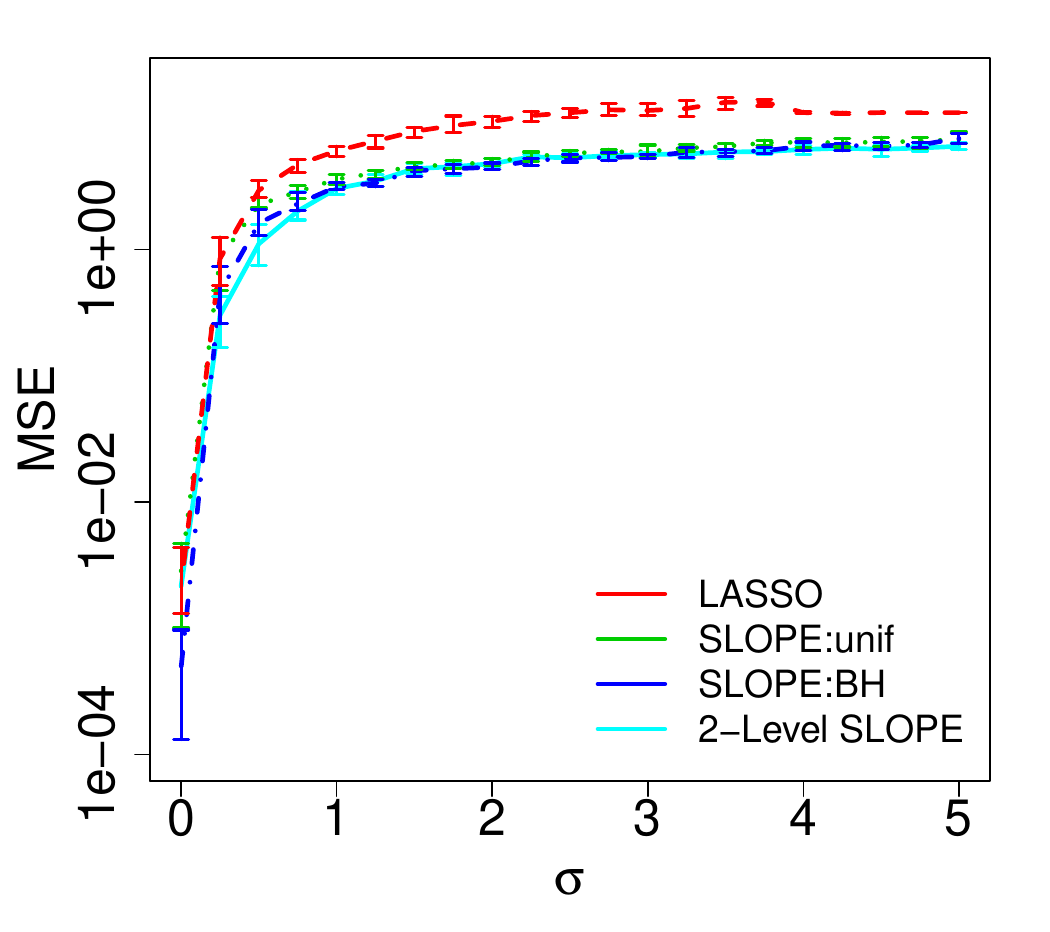}
\caption{heavy-tail + correlated $\X$, tied $\bet$}  
     \label{fig:l}
 \end{subfigure}
 \caption{MSE of LASSO, 2-level SLOPE and other SLOPE penalties, under Gaussian and heavy-tailed data, i.i.d. and correlated distributions, tied and non-tied priors. For each sub-plot, the left figure shows the uniform scale while the right figure shows the logarithm scale.}
 \label{fig:arian}
\end{figure}

In fact, Theorem 2.2 and Proposition 2.3 in \cite{wang2022does} have theoretically shown that LASSO is the best SLOPE when (1) the data $\X$ is Gaussian i.i.d., (2) the prior $\bet$ is non-tied, and (3) the noise $\sigma$ is small. This setting is visualized in the sub-plot (a) of \Cref{fig:arian}, where 2-level SLOPE reduces to the LASSO under small noise. Empirically speaking, the optimality of LASSO extends to correlated or heavy-tailed $\X$, as we visualize in sub-plots (b--d), where the red curves have the smallest MSE under small noise. However, (2) and (3) cannot be relaxed: for (2), we see in sub-plots (e--h) that 2-level SLOPE does not reduce to LASSO even under small noise; for (3), we see in sub-plots (a--h) that red curves are worse than others under large noise.
In particular,
the LASSO suffers from high MSE when $\X$ is correlated as in right column of \Cref{fig:arian}. 

In summary, 2-level SLOPE achieves remarkably small MSE under all settings: the 2-level SLOPE reduces to the LASSO when LASSO is indeed the best SLOPE, and remains robust to large noise and correlated $\X$, regardless of whether the priors are tied or not.

\section{Discussion}
\label{sec:discussion}

In this work, we study the 2-level SLOPE as the simplest extension of LASSO that adaptively applies a larger penalty to the larger elements of the estimator. From the perspective of practitioners, 2-level SLOPE is much easier to design in comparison to the general SLOPE, only requiring 3 hyper-parameters even for large $p$-dimension problems. This allows practitioners to use grid search to tune the penalty, which is infeasible for the general SLOPE.

Surprisingly, 2-level SLOPE bears a strong resemblance to the general SLOPE, in terms of TPP-FDP trade-off and estimation error. Specifically, 2-level SLOPE has explicit TPP-FDP trade-off, whereas the general SLOPE does not, and the 2-level SLOPE tradeoff curve matches closely to the general SLOPE tradeoff curve lower bound from \cite{bu2021characterizing}.

In addition, the 2-level SLOPE regularization is amazingly generalizable to other problems and models. From the model perspective, the 2-level SLOPE regularization can be applied whenever the LASSO regularization is applicable: besides the linear regression model \eqref{eq:K-level}, we can apply 2-level SLOPE regularization on the linear classification model (i.e.,\ the logistic regression), matrix completion, and deep learning models via neural networks. 
From the optimization problem perspective, one can easily formulate the 2-level group SLOPE as a subclass of the group SLOPE and further study the performance in combination with the above-mentioned models.

We finally mention that there remain many open questions about the relationships between asymptotic model selection and estimation error properties of SLOPE, 2-level SLOPE, and the LASSO in a setting where the measurement matrix $\mathbf{X}$ is not assumed to be Gaussian. By using the theory for non-separable AMP algorithms~\cite{berthier2020state}, one could potentially extend this work the settings where $\mathbf{X}$ remains Gaussian, but with correlated columns. Many recent works have extended AMP algorithms beyond the regime of i.i.d.\ Gaussian  $\mathbf{X}$, see for example~\cite{chen2021universality, dudeja2023universality, wang2024universality, rangan2019vector} and these works could be used to avoid the Gaussianity assumption on $\mathbf{X}$ altogether.




	
	\begin{acks}[Acknowledgments]
		Rush acknowledges support from the National Science Foundation under DMS-2413828. Klusowski acknowledges support from the National Science Foundation under CAREER DMS-2239448.
	\end{acks}

	\clearpage
	\bibliographystyle{imsart-nameyear}
	\bibliography{ref}

\begin{thebibliography}{30}

\bibitem[\protect\citeauthoryear{Bayati and Montanari}{2011a}]{bayati2011lasso}
\begin{barticle}[author]
\bauthor{\bsnm{Bayati},~\bfnm{Mohsen}\binits{M.}} \AND
  \bauthor{\bsnm{Montanari},~\bfnm{Andrea}\binits{A.}}
(\byear{2011}a).
\btitle{The LASSO risk for Gaussian matrices}.
\bjournal{IEEE Transactions on Information Theory}
\bvolume{58}
\bpages{1997--2017}.
\end{barticle}
\endbibitem

\bibitem[\protect\citeauthoryear{Bayati and
  Montanari}{2011b}]{bayati2011dynamics}
\begin{barticle}[author]
\bauthor{\bsnm{Bayati},~\bfnm{Mohsen}\binits{M.}} \AND
  \bauthor{\bsnm{Montanari},~\bfnm{Andrea}\binits{A.}}
(\byear{2011}b).
\btitle{The dynamics of message passing on dense graphs, with applications to
  compressed sensing}.
\bjournal{IEEE Transactions on Information Theory}
\bvolume{57}
\bpages{764--785}.
\end{barticle}
\endbibitem

\bibitem[\protect\citeauthoryear{Berthier, Montanari and
  Nguyen}{2020}]{berthier2020state}
\begin{barticle}[author]
\bauthor{\bsnm{Berthier},~\bfnm{Raphael}\binits{R.}},
  \bauthor{\bsnm{Montanari},~\bfnm{Andrea}\binits{A.}} \AND
  \bauthor{\bsnm{Nguyen},~\bfnm{Phan-Minh}\binits{P.-M.}}
(\byear{2020}).
\btitle{State evolution for approximate message passing with non-separable
  functions}.
\bjournal{Information and Inference: A Journal of the IMA}
\bvolume{9}
\bpages{33--79}.
\end{barticle}
\endbibitem

\bibitem[\protect\citeauthoryear{Bogdan et~al.}{2015}]{bogdan2015slope}
\begin{barticle}[author]
\bauthor{\bsnm{Bogdan},~\bfnm{Ma{\l}gorzata}\binits{M.}}, \bauthor{\bsnm{Van
  Den~Berg},~\bfnm{Ewout}\binits{E.}},
  \bauthor{\bsnm{Sabatti},~\bfnm{Chiara}\binits{C.}},
  \bauthor{\bsnm{Su},~\bfnm{Weijie}\binits{W.}} \AND
  \bauthor{\bsnm{Cand{\`e}s},~\bfnm{Emmanuel~J}\binits{E.~J.}}
(\byear{2015}).
\btitle{SLOPE—adaptive variable selection via convex optimization}.
\bjournal{The annals of applied statistics}
\bvolume{9}
\bpages{1103}.
\end{barticle}
\endbibitem

\bibitem[\protect\citeauthoryear{Bondell and
  Reich}{2008}]{bondell2008simultaneous}
\begin{barticle}[author]
\bauthor{\bsnm{Bondell},~\bfnm{Howard~D}\binits{H.~D.}} \AND
  \bauthor{\bsnm{Reich},~\bfnm{Brian~J}\binits{B.~J.}}
(\byear{2008}).
\btitle{Simultaneous regression shrinkage, variable selection, and supervised
  clustering of predictors with OSCAR}.
\bjournal{Biometrics}
\bvolume{64}
\bpages{115--123}.
\end{barticle}
\endbibitem

\bibitem[\protect\citeauthoryear{Brunk et~al.}{1972}]{brunk1972statistical}
\begin{barticle}[author]
\bauthor{\bsnm{Brunk},~\bfnm{HD}\binits{H.}},
  \bauthor{\bsnm{Barlow},~\bfnm{Richard~E}\binits{R.~E.}},
  \bauthor{\bsnm{Bartholomew},~\bfnm{Daniel~J}\binits{D.~J.}} \AND
  \bauthor{\bsnm{Bremner},~\bfnm{James~M}\binits{J.~M.}}
(\byear{1972}).
\btitle{Statistical inference under order restrictions.(the theory and
  application of isotonic regression)}.
\bjournal{International Statistical Review}
\bvolume{41}
\bpages{395}.
\end{barticle}
\endbibitem

\bibitem[\protect\citeauthoryear{Bu et~al.}{2020}]{bu2019algorithmic}
\begin{barticle}[author]
\bauthor{\bsnm{Bu},~\bfnm{Zhiqi}\binits{Z.}},
  \bauthor{\bsnm{Klusowski},~\bfnm{Jason~M}\binits{J.~M.}},
  \bauthor{\bsnm{Rush},~\bfnm{Cynthia}\binits{C.}} \AND
  \bauthor{\bsnm{Su},~\bfnm{Weijie~J}\binits{W.~J.}}
(\byear{2020}).
\btitle{Algorithmic Analysis and Statistical Estimation of SLOPE via
  Approximate Message Passing}.
\bjournal{IEEE Transactions on Information Theory}
\bvolume{67}
\bpages{506--537}.
\end{barticle}
\endbibitem

\bibitem[\protect\citeauthoryear{Bu et~al.}{2023}]{bu2021characterizing}
\begin{barticle}[author]
\bauthor{\bsnm{Bu},~\bfnm{Zhiqi}\binits{Z.}},
  \bauthor{\bsnm{Klusowski},~\bfnm{Jason}\binits{J.}},
  \bauthor{\bsnm{Rush},~\bfnm{Cynthia}\binits{C.}} \AND
  \bauthor{\bsnm{Su},~\bfnm{Weijie~J}\binits{W.~J.}}
(\byear{2023}).
\btitle{Characterizing the SLOPE Trade-off: A Variational Perspective and the
  Donoho-Tanner Limit}.
\bjournal{The Annals of Statistics}
\bvolume{51}
\bpages{33--61}.
\end{barticle}
\endbibitem

\bibitem[\protect\citeauthoryear{Candes, Romberg and
  Tao}{2006}]{candes2006stable}
\begin{barticle}[author]
\bauthor{\bsnm{Candes},~\bfnm{Emmanuel~J}\binits{E.~J.}},
  \bauthor{\bsnm{Romberg},~\bfnm{Justin~K}\binits{J.~K.}} \AND
  \bauthor{\bsnm{Tao},~\bfnm{Terence}\binits{T.}}
(\byear{2006}).
\btitle{Stable signal recovery from incomplete and inaccurate measurements}.
\bjournal{Communications on Pure and Applied Mathematics: A Journal Issued by
  the Courant Institute of Mathematical Sciences}
\bvolume{59}
\bpages{1207--1223}.
\end{barticle}
\endbibitem

\bibitem[\protect\citeauthoryear{Candes and Tao}{2005}]{candes2005decoding}
\begin{barticle}[author]
\bauthor{\bsnm{Candes},~\bfnm{Emmanuel~J}\binits{E.~J.}} \AND
  \bauthor{\bsnm{Tao},~\bfnm{Terence}\binits{T.}}
(\byear{2005}).
\btitle{Decoding by linear programming}.
\bjournal{IEEE transactions on information theory}
\bvolume{51}
\bpages{4203--4215}.
\end{barticle}
\endbibitem

\bibitem[\protect\citeauthoryear{Chen, Bu and Xu}{2021}]{chen2021asymptotic}
\begin{binproceedings}[author]
\bauthor{\bsnm{Chen},~\bfnm{Kan}\binits{K.}},
  \bauthor{\bsnm{Bu},~\bfnm{Zhiqi}\binits{Z.}} \AND
  \bauthor{\bsnm{Xu},~\bfnm{Shiyun}\binits{S.}}
(\byear{2021}).
\btitle{Asymptotic Statistical Analysis of Sparse Group LASSO via Approximate
  Message Passing}.
In \bbooktitle{Joint European Conference on Machine Learning and Knowledge
  Discovery in Databases}
\bpages{510--526}.
\bpublisher{Springer}.
\end{binproceedings}
\endbibitem

\bibitem[\protect\citeauthoryear{Chen and Lam}{2021}]{chen2021universality}
\begin{barticle}[author]
\bauthor{\bsnm{Chen},~\bfnm{Wei-Kuo}\binits{W.-K.}} \AND
  \bauthor{\bsnm{Lam},~\bfnm{Wai-Kit}\binits{W.-K.}}
(\byear{2021}).
\btitle{Universality of approximate message passing algorithms}.
\bjournal{Electric Journal of Probability}.
\end{barticle}
\endbibitem

\bibitem[\protect\citeauthoryear{Donoho, Maleki and
  Montanari}{2009}]{donoho2009message}
\begin{barticle}[author]
\bauthor{\bsnm{Donoho},~\bfnm{David~L}\binits{D.~L.}},
  \bauthor{\bsnm{Maleki},~\bfnm{Arian}\binits{A.}} \AND
  \bauthor{\bsnm{Montanari},~\bfnm{Andrea}\binits{A.}}
(\byear{2009}).
\btitle{Message-passing algorithms for compressed sensing}.
\bjournal{Proceedings of the National Academy of Sciences}
\bvolume{106}
\bpages{18914--18919}.
\end{barticle}
\endbibitem

\bibitem[\protect\citeauthoryear{Dudeja, M.~Lu and
  Sen}{2023}]{dudeja2023universality}
\begin{barticle}[author]
\bauthor{\bsnm{Dudeja},~\bfnm{Rishabh}\binits{R.}},
  \bauthor{\bsnm{M.~Lu},~\bfnm{Yue}\binits{Y.}} \AND
  \bauthor{\bsnm{Sen},~\bfnm{Subhabrata}\binits{S.}}
(\byear{2023}).
\btitle{Universality of approximate message passing with semirandom matrices}.
\bjournal{The Annals of Probability}
\bvolume{51}
\bpages{1616--1683}.
\end{barticle}
\endbibitem

\bibitem[\protect\citeauthoryear{Feng et~al.}{2022}]{feng2022unifying}
\begin{barticle}[author]
\bauthor{\bsnm{Feng},~\bfnm{Oliver~Y}\binits{O.~Y.}},
  \bauthor{\bsnm{Venkataramanan},~\bfnm{Ramji}\binits{R.}},
  \bauthor{\bsnm{Rush},~\bfnm{Cynthia}\binits{C.}},
  \bauthor{\bsnm{Samworth},~\bfnm{Richard~J}\binits{R.~J.}} \betal{et~al.}
(\byear{2022}).
\btitle{A unifying tutorial on approximate message passing}.
\bjournal{Foundations and Trends{\textregistered} in Machine Learning}
\bvolume{15}
\bpages{335--536}.
\end{barticle}
\endbibitem

\bibitem[\protect\citeauthoryear{Figueiredo and
  Nowak}{2016}]{figueiredo2016ordered}
\begin{binproceedings}[author]
\bauthor{\bsnm{Figueiredo},~\bfnm{Mario}\binits{M.}} \AND
  \bauthor{\bsnm{Nowak},~\bfnm{Robert}\binits{R.}}
(\byear{2016}).
\btitle{Ordered weighted l1 regularized regression with strongly correlated
  covariates: Theoretical aspects}.
In \bbooktitle{Artificial Intelligence and Statistics}
\bpages{930--938}.
\bpublisher{PMLR}.
\end{binproceedings}
\endbibitem

\bibitem[\protect\citeauthoryear{Hu and Lu}{2019}]{hu2019asymptotics}
\begin{binproceedings}[author]
\bauthor{\bsnm{Hu},~\bfnm{Hong}\binits{H.}} \AND
  \bauthor{\bsnm{Lu},~\bfnm{Yue~M}\binits{Y.~M.}}
(\byear{2019}).
\btitle{Asymptotics and optimal designs of SLOPE for sparse linear regression}.
In \bbooktitle{2019 IEEE International Symposium on Information Theory (ISIT)}
\bpages{375--379}.
\bpublisher{IEEE}.
\end{binproceedings}
\endbibitem

\bibitem[\protect\citeauthoryear{Kruskal}{1964}]{kruskal1964nonmetric}
\begin{barticle}[author]
\bauthor{\bsnm{Kruskal},~\bfnm{Joseph~B}\binits{J.~B.}}
(\byear{1964}).
\btitle{Nonmetric multidimensional scaling: a numerical method}.
\bjournal{Psychometrika}
\bvolume{29}
\bpages{115--129}.
\end{barticle}
\endbibitem

\bibitem[\protect\citeauthoryear{Lustig et~al.}{2008}]{lustig2008compressed}
\begin{barticle}[author]
\bauthor{\bsnm{Lustig},~\bfnm{Michael}\binits{M.}},
  \bauthor{\bsnm{Donoho},~\bfnm{David~L}\binits{D.~L.}},
  \bauthor{\bsnm{Santos},~\bfnm{Juan~M}\binits{J.~M.}} \AND
  \bauthor{\bsnm{Pauly},~\bfnm{John~M}\binits{J.~M.}}
(\byear{2008}).
\btitle{Compressed sensing MRI}.
\bjournal{IEEE signal processing magazine}
\bvolume{25}
\bpages{72--82}.
\end{barticle}
\endbibitem

\bibitem[\protect\citeauthoryear{Ogutu, Schulz-Streeck and
  Piepho}{2012}]{ogutu2012genomic}
\begin{binproceedings}[author]
\bauthor{\bsnm{Ogutu},~\bfnm{Joseph~O}\binits{J.~O.}},
  \bauthor{\bsnm{Schulz-Streeck},~\bfnm{Torben}\binits{T.}} \AND
  \bauthor{\bsnm{Piepho},~\bfnm{Hans-Peter}\binits{H.-P.}}
(\byear{2012}).
\btitle{Genomic selection using regularized linear regression models: ridge
  regression, lasso, elastic net and their extensions}.
In \bbooktitle{BMC proceedings}
\bvolume{6}
\bpages{1--6}.
\bpublisher{Springer}.
\end{binproceedings}
\endbibitem

\bibitem[\protect\citeauthoryear{Rangan, Schniter and
  Fletcher}{2019}]{rangan2019vector}
\begin{barticle}[author]
\bauthor{\bsnm{Rangan},~\bfnm{Sundeep}\binits{S.}},
  \bauthor{\bsnm{Schniter},~\bfnm{Philip}\binits{P.}} \AND
  \bauthor{\bsnm{Fletcher},~\bfnm{Alyson~K}\binits{A.~K.}}
(\byear{2019}).
\btitle{Vector approximate message passing}.
\bjournal{IEEE Transactions on Information Theory}
\bvolume{65}
\bpages{6664--6684}.
\end{barticle}
\endbibitem

\bibitem[\protect\citeauthoryear{Su et~al.}{2016}]{su2016slope}
\begin{barticle}[author]
\bauthor{\bsnm{Su},~\bfnm{Weijie}\binits{W.}},
  \bauthor{\bsnm{Candes},~\bfnm{Emmanuel}\binits{E.}} \betal{et~al.}
(\byear{2016}).
\btitle{SLOPE is adaptive to unknown sparsity and asymptotically minimax}.
\bjournal{The Annals of Statistics}
\bvolume{44}
\bpages{1038--1068}.
\end{barticle}
\endbibitem

\bibitem[\protect\citeauthoryear{Su et~al.}{2017}]{su2017false}
\begin{barticle}[author]
\bauthor{\bsnm{Su},~\bfnm{Weijie}\binits{W.}},
  \bauthor{\bsnm{Bogdan},~\bfnm{Ma{\l}gorzata}\binits{M.}},
  \bauthor{\bsnm{Candes},~\bfnm{Emmanuel}\binits{E.}} \betal{et~al.}
(\byear{2017}).
\btitle{False discoveries occur early on the lasso path}.
\bjournal{The Annals of statistics}
\bvolume{45}
\bpages{2133--2150}.
\end{barticle}
\endbibitem

\bibitem[\protect\citeauthoryear{Tibshirani}{1996}]{tibshirani1996regression}
\begin{barticle}[author]
\bauthor{\bsnm{Tibshirani},~\bfnm{Robert}\binits{R.}}
(\byear{1996}).
\btitle{Regression shrinkage and selection via the lasso}.
\bjournal{Journal of the Royal Statistical Society: Series B (Methodological)}
\bvolume{58}
\bpages{267--288}.
\end{barticle}
\endbibitem

\bibitem[\protect\citeauthoryear{Wang, Weng and Maleki}{2022}]{wang2022does}
\begin{barticle}[author]
\bauthor{\bsnm{Wang},~\bfnm{Shuaiwen}\binits{S.}},
  \bauthor{\bsnm{Weng},~\bfnm{Haolei}\binits{H.}} \AND
  \bauthor{\bsnm{Maleki},~\bfnm{Arian}\binits{A.}}
(\byear{2022}).
\btitle{Does SLOPE outperform bridge regression?}
\bjournal{Information and Inference: A Journal of the IMA}
\bvolume{11}
\bpages{1--54}.
\end{barticle}
\endbibitem

\bibitem[\protect\citeauthoryear{Wang, Zhong and
  Fan}{2024}]{wang2024universality}
\begin{barticle}[author]
\bauthor{\bsnm{Wang},~\bfnm{Tianhao}\binits{T.}},
  \bauthor{\bsnm{Zhong},~\bfnm{Xinyi}\binits{X.}} \AND
  \bauthor{\bsnm{Fan},~\bfnm{Zhou}\binits{Z.}}
(\byear{2024}).
\btitle{Universality of approximate message passing algorithms and tensor
  networks}.
\bjournal{The Annals of Applied Probability}
\bvolume{34}
\bpages{3943--3994}.
\end{barticle}
\endbibitem

\bibitem[\protect\citeauthoryear{Wang et~al.}{2020}]{wang2020complete}
\begin{barticle}[author]
\bauthor{\bsnm{Wang},~\bfnm{Hua}\binits{H.}},
  \bauthor{\bsnm{Yang},~\bfnm{Yachong}\binits{Y.}},
  \bauthor{\bsnm{Bu},~\bfnm{Zhiqi}\binits{Z.}} \AND
  \bauthor{\bsnm{Su},~\bfnm{Weijie}\binits{W.}}
(\byear{2020}).
\btitle{The complete Lasso tradeoff diagram}.
\bjournal{Advances in Neural Information Processing Systems}
\bvolume{33}
\bpages{20051--20060}.
\end{barticle}
\endbibitem

\bibitem[\protect\citeauthoryear{Zeng and
  Figueiredo}{2014}]{zeng2014decreasing}
\begin{barticle}[author]
\bauthor{\bsnm{Zeng},~\bfnm{Xiangrong}\binits{X.}} \AND
  \bauthor{\bsnm{Figueiredo},~\bfnm{M{\'a}rio~AT}\binits{M.~A.}}
(\byear{2014}).
\btitle{Decreasing Weighted Sorted $\ell_1$ Regularization}.
\bjournal{IEEE Signal Processing Letters}
\bvolume{21}
\bpages{1240--1244}.
\end{barticle}
\endbibitem

\bibitem[\protect\citeauthoryear{Zhang and Bu}{2021}]{zhang2021efficient}
\begin{binproceedings}[author]
\bauthor{\bsnm{Zhang},~\bfnm{Yiliang}\binits{Y.}} \AND
  \bauthor{\bsnm{Bu},~\bfnm{Zhiqi}\binits{Z.}}
(\byear{2021}).
\btitle{Efficient designs of slope penalty sequences in finite dimension}.
In \bbooktitle{International Conference on Artificial Intelligence and
  Statistics}
\bpages{3277--3285}.
\bpublisher{PMLR}.
\end{binproceedings}
\endbibitem

\bibitem[\protect\citeauthoryear{Zou and Hastie}{2005}]{zou2005regularization}
\begin{barticle}[author]
\bauthor{\bsnm{Zou},~\bfnm{Hui}\binits{H.}} \AND
  \bauthor{\bsnm{Hastie},~\bfnm{Trevor}\binits{T.}}
(\byear{2005}).
\btitle{Regularization and variable selection via the elastic net}.
\bjournal{Journal of the royal statistical society: series B (statistical
  methodology)}
\bvolume{67}
\bpages{301--320}.
\end{barticle}
\endbibitem

\end{thebibliography}

	\clearpage
	\appendix
	\section{List of functions}
	
	\begin{table}[H]
		\centering
		\begin{tabular}{|c|c|c|c|}
			\hline
			& Input &Output& Reference \\\hline
			
			$\widehat\bet$  & $\X,\y,\blam$ & general/2-level SLOPE minimizer& \Cref{eq:SLOPE,eq:K-level} \\
			
			$\bet^k$  & $\X,\y,\blam$ & general/2-level SLOPE iterates by FISTA& \Cref{eq:fista} \\       
			
			$\bet$  & $\Pi$ & true signal of linear regression& $\y=\X\bet+\w$ \\       
			
			$\blam$  & $\Lambda$ & general/2-level SLOPE penalty  & \Cref{eq:asymptotic lambda}\\
			
			$\eta_\text{SLOPE}$  & $\boldsymbol{v},\thet$ & SLOPE proximal operator  & \Cref{eq:SLOPE prox}\\
			
			$\eta_{V,\Theta}$  & $V,\Theta$ & SLOPE limiting scalar function  & \Cref{eq:limiting scalar}\\
			
			$\alpha$  & $(\Pi,\Lambda)$ or $(\pi,\A)$ & zero-threshold & \Cref{def: zero threshold}\\
			
			$\A$&$\Pi,\Lambda,\tau,\delta$& normalized penalty&\Cref{eq:calibration}\\
			
			$\tau$&$\Pi,\A,\delta,\sigma$&state evolution&\Cref{eq:state evolution}\\
			
			$\pi,\pi^*$&$\Pi,\tau$&normalized prior&\Cref{eq:define pi}\\
			
			$\TPP,\FDP$&$(\Pi,\Lambda)$ or $(\pi,\A)$&asymptotic TPP, FDP&\Cref{eq:tppfdp comp}\\
			
			$\min\FDP$&$\alpha^*,\TPP,\epsilon$&minimum $\FDP$ over $(\Pi,\Lambda)$&\Cref{eq:fdp is function of tpp}\\
			
			$\alpha^*$&$\delta,\epsilon,\TPP$&maximum zero-threshold&\Cref{eq:all prior max zero thres}\\
			
			$\Aeff$&$\pi,\alpha_1,\alpha_2,s,h$& 2-level effective penalty&\Cref{eq: 2level Aeff}\\
			
			$h,q_1,q_2$&$\pi,\alpha_1,\alpha_2,s$& quantities for $\Aeff$&\Cref{lem:q1q2h}\\
			
			
			$\rho$&$t_1,t_2,\TPP,\alpha_1,\alpha_2,$&probability of $\pi_{\min}$ being $t_1$ if non-zero&\Cref{eq:compute p}\\
			
			$\pi_{\min}$&$t_1,t_2,\rho,\epsilon$&optimal $\pi$, i.e., $\text{argmin}_\pi\alpha(\pi,\A)$ &\Cref{eq:three point pi}\\

			$F[\pi]$&$\pi,\alpha_1,\alpha_2,s,\Aeff$&quantity for state evolution&\Cref{eq:state evolution constraint}\\
			\hline 
		\end{tabular}
		\caption{List of notations used throughout this paper.}
		\label{tab:notations}
	\end{table}

	\section{Proofs}
	\label{app:proofs}
	
	\subsection{Proof of \Cref{lem:1 flat}} \label{app:lem1proof}
	
	\begin{proof}
		While the lemma states a generalized result for the $K$-level SLOPE, we present the proof for the 2-level SLOPE because the same analysis readily generalizes.
		
		Notice that the SLOPE minimizer satisfies
		$$\widehat\bet=\eta_\text{SLOPE}\left(\widehat\bet
		+ t_\infty \X^\top(\y-\X\widehat\bet);t_\infty \blam \right).$$
		This convergence is guaranteed for such a convex problem under proper conditions on the sequence of step sizes \{$t_k$\} (see Section 2 in \cite{bogdan2015slope}). 
		Given that $\widehat\bet=\eta_\text{SLOPE}(\bm v;\thet)$ for some $(\bm v,\thet)$, it suffices to prove that, for any $(\bm v,\thet)$ subject to 2-level penalty, the elements of $|\eta_\text{SLOPE}(\bm v;\thet)|$ share at most one non-zero value.
		
		The proof relies on how the SLOPE proximal operator is computed, for which we leverage the pool adjacent violators algorithm (PAVA) and \cite[Algorithm 3]{bogdan2015slope}. These algorithms provably solve \eqref{eq:SLOPE prox} in five steps; outlined in \Cref{alg:PAVA} below and visualized in \Cref{fig:normalized PAVA}. 
		
		\begin{algorithm}[!htb]
			\caption{PAVA algorithm to compute the proximal operator $\eta_\text{SLOPE}(\bm v;\thet)$}
			\begin{algorithmic}[1]
				\State \textbf{Sorting: }Sort $|\bm v|$ in increasing order, returning $\operatorname{sort}(|\boldsymbol{v}|)$;
				\State\textbf{Differencing: }Calculate a difference sequence, $\boldsymbol{V}$, defined as
				$
				\boldsymbol{V} :=\operatorname{sort}(|\boldsymbol{v}|)-\thet;
				$
				\State\textbf{Averaging: }Repeatedly average out strictly decreasing sub-sequences in $\boldsymbol{V}$ until none remain. We refer to the increasing sequence after all the averaging as $\operatorname{AVE}(\boldsymbol{V})$, and reassign
				$
				\boldsymbol{V}=\operatorname{AVE}(\boldsymbol{V});
				$
				\State \textbf{Truncating: }Set negative values in the difference sequence to 0 and reassign
				$
				\boldsymbol{V}=\max (\boldsymbol{V}, 0) ;
				$
				\State \textbf{Restoring: }Restore the order and the sign of $\boldsymbol{v}$ from Step 1 to $\boldsymbol{V}$. 
			\end{algorithmic}
			\label{alg:PAVA}
		\end{algorithm}
		
		
		We will show that for any $\bm v\in \R^p$ and any 2-level penalty $\thet\in\R^p$ as in \Cref{eq:K-level SLOPE seq}, $\{\eta_\text{SLOPE}(\bm v;\thet)\}_\text{shared}$ contains either \textbf{(I.1)} only one element and it is zero, \textbf{(I.2)} no elements, \textbf{(II.1)} only one element and it is non-zero, or \textbf{(II.2)} two elements, one that is non-zero and one that is zero. Without loss of generality, we assume $\bm v$ is sorted, strictly increasing, and non-negative, i.e.,\ $\bm v=\operatorname{sort}(|\boldsymbol{v}|)$, so as to omit Step 1 and 5 for the ease of presentation.
		
		In Step 2 of PAVA (the second row of \Cref{fig:normalized PAVA}), the 2-level penalty $\thet$ is subtracted from the monotonic sequence $\bm v$, resulting in two pieces of increasing sub-sequences:
		\begin{align*}
			\thet=\langle\theta_1,\theta_2,s\rangle\Longrightarrow
			\begin{cases}
				[\bm v-\thet]_i=v_{i}-\theta_2\leq[\bm v-\thet]_{i+1}=v_{i+1}-\theta_2, &\quad \forall i<(1-s)p,
				\\
				[\bm v-\thet]_i=v_{i}-\theta_1\leq[\bm v-\thet]_{i+1}=v_{i+1}-\theta_1, &\quad \forall i>(1-s)p.
			\end{cases}
		\end{align*}
		But at $i=(1-s)p$, there are two possibilities:
		
		\textbf{(I)} $[\bm v-\thet]_{(1-s)p}=v_{(1-s)p}-\theta_2\leq [\bm v-\thet]_{(1-s)p+1}=v_{(1-s)p+1}-\theta_1$, i.e.,\ the overall $\R^p$ sequence $\bm v-\thet$ is still monotonically increasing; 
		
		\textbf{(II)} $[\bm v-\thet]_{(1-s)p}=v_{(1-s)p}-\theta_2 > [\bm v-\thet]_{(1-s)p+1}=v_{(1-s)p+1}-\theta_1$, i.e.,\ a vertical drop exists and the overall $\R^p$ sequence $\bm v-\thet$ is not monotonically decreasing.
		
		
		For possibility \textbf{(I)}, Step 3 is not performed as there is no strictly decreasing sub-sequences. For possibility \textbf{(II)}, Step 3 (the third row of \Cref{fig:normalized PAVA}) averages the strictly decreasing sub-sequence between indices $(1-s)p$ and $(1-s)p+1$. After the averaging, the overall sequence may still be non-increasing. Therefore the averaging is repeated on the new (and longer) sub-sequence until the flattened region is larger than its left endpoint and smaller than its right endpoint, i.e.,\ until the overall sequence is monotonically increasing. Because there is only one vertical drop at $(1-s)p$ and it is easy to prove by induction that the flattened region always covers $(1-s)p$, we obtain that the overall sequence has exactly one flattened region after Step 3.
		
		Denoting the flattened region as $h$, Step 4 (the fourth row of \Cref{fig:normalized PAVA}) may either truncate this flattened region to zero, if $h\leq 0$ and $\min(\bm v-\thet)<0$, leading to case \textbf{(I.1)}; or introduce a second flattened region at zero, if $h>0$ and $\min(\bm v-\thet)<0$, leading to case \textbf{(II.2)}; or does nothing if the sequence is all non-negative, i.e.,\ $h>0$ and $\min(\bm v-\thet)\geq 0$, leading to case \textbf{(II.1)}. Finally, if there is no flattened region, then we have case \textbf{(I.2)} if $\min(\bm v-\thet)>0$, or case \textbf{(I.1)} if $\min(\bm v-\thet)\leq 0$.
		
		Therefore, we have proved the four cases \textbf{(I.1)}, \textbf{(I.2)}, \textbf{(II.1)} and \textbf{(II.2)}, which indicate that $\{\eta_\text{SLOPE}(\bm v;\thet)\}_\text{shared}$ has at most 1 non-zero element, i.e.,\
		$$\left|\{h\in\R^+: h\in\{\eta_\text{SLOPE}(\bm v;\thet)\}_\text{shared}\}\right|\leq 1$$
		In other words, the SLOPE solution shares at most 1 non-zero magnitude.


	\end{proof}
	
	We give some concrete examples of the four cases analyzed in the proof:
	\begin{itemize}
		\item  \textbf{(II.1)} sharing one non-zero element: $\eta_\text{SLOPE}((4,-4,2.5);(4,2,2))=(1,-1,0.5)$
		\item \textbf{(II.2)} sharing both a non-zero and a zero magnitude: $\eta_\text{SLOPE}((4,-4,2,1);(4,2,2,2))=(1,-1,0,0)$
		\item \textbf{(I.1)} sharing only zero elements: $\eta_\text{SLOPE}((4,-1,0.5);(3,2,2))=(1,0,0)$
		\item \textbf{(I.2)} not sharing any elements: $\eta_\text{SLOPE}((4,-3,2);(3,2.5,2.5))=(1,-0.5,0)$
	\end{itemize}

	\subsection{Proof of \Cref{thm:1 point mass}}
	\begin{proof}
		The result that the empirical distribution $\widehat\bet$ for 2-level SLOPE weakly converges to $\widehat\Pi=\eta_{\Pi+\tau Z,\A\tau}(\Pi+\tau Z)$ follows directly from \cite[Theorem 3]{bu2019algorithmic} and \cite[Theorem 1]{hu2019asymptotics}, which is also presented in \cite[Section 3]{bu2021characterizing} for the general SLOPE.
		
		We now study the non-zero point mass, starting with the derivation of a general SLOPE limiting scalar function by \cite[Proposition 1]{hu2019asymptotics}, which solves 
		\[
		\eta_{V,\Theta}=\text{argmin}_{g \in \mathcal{I}} \frac{1}{2} \mathbb{E}_{V} \big[ (V - g(V))^2 \big] + \int_0^1 F_{\Theta}^{-1}(u) F_{|g(V)|}^{-1}(u) \,du,
		\]
		where 
		\[
		\mathcal{I} := \{ g(v) \mid g(v) \text{ is odd, non-decreasing, and 1-Lipschitz} \}.
		\]
		
		Here $V,\Theta$ are the same as in \eqref{eq:limiting scalar}, $F$ is the cumulative distribution function and $F_X^{-1}(p):=\inf\{x:\PP(X\leq x)\geq p\}$ is the quantile function. For 2-level SLOPE, $\Theta$ takes two values $\theta_1>\theta_2$, split at a ratio $s$. Hence
		\begin{align*}
			&\frac{1}{2} \mathbb{E}_{V} \big[ (V - g(V))^2 \big] + \int_0^1 F_{\Theta}^{-1}(u) F_{|g(V)|}^{-1}(u) \,du
			\\
			=&\frac{1}{2} \mathbb{E}_{V} \big[ (V - g(V))^2 \big] + \int_0^{1-s} F_{\Theta}^{-1}(u) F_{|g(V)|}^{-1}(u) \,du+ \int_{1-s}^1 F_{\Theta}^{-1}(u) F_{|g(V)|}^{-1}(u) \,du
			\\
			=&\frac{1}{2} \mathbb{E}_{V} \big[ (V - g(V))^2 \big] + (1-s) \theta_2\E_{V}\left(|g(V)| \Big||g|<|G|_{1-s}\right)+ s\theta_1\E_{V}\left(|g(V)| \Big||g|>|G|_{1-s}\right) 
		\end{align*}
		where we have used that integral of quantile function is conditional expectation, and denoted $|G|_{1-s}$ as the $(1-s)$ quantile of $|g(V)|$.
		
		With the law of total expectation, we have
		\begin{align*}
			&\frac{1-s}{2} \mathbb{E}_{V}\left[(V - g(V))^2  + 2 \theta_2|g| \Big||g|<|G|_{1-s}\right]
			+ \frac{s}{2}\mathbb{E}_{V}\left[(V - g(V))^2  + 2 \theta_1|g| \Big||g|>|G|_{1-s}\right]
		\end{align*}
		which reduces to, up to constant terms, 
		\begin{align*}
			\frac{1-s}{2} \mathbb{E}_{V} \left[ (|V| -\theta_2- g(|V|))^2\Big||g|<|G|_{1-s} \right] + \frac{s}{2} \mathbb{E}_{V} \left[ (|V| -\theta_1- g(|V|))^2\Big||g|>|G|_{1-s} \right]
		\end{align*}
		To see this, we see the difference in the first term is 
		\begin{align*}
			&\mathbb{E}_{V}\left[(|V| - g(|V|))^2  + 2 \theta_2|g| -(|V| -\theta_2- g(|V|))^2\Big||g|<|G|_{1-s} \right]
			\\
			=&\theta_2\mathbb{E}_{V}\left[(2|V| - 2|g|-\theta_2)  + 2 |g|\Big||g|<|G|_{1-s} \right]
			\\
			=&\theta_2\mathbb{E}_{V}\left[2|V| -\theta_2\Big||g|<|G|_{1-s} \right]
			=\theta_2\mathbb{E}_{V}\left[2|V| -\theta_2\Big||V|<|V|_{1-s} \right]
		\end{align*}
		which is independent of $g$ that we are optimizing, hence it is a constant term. Here we denote $|V|_{1-s}$ as the $(1-s)$ quantile of $|V|$.
		
		Denoting $f(|V|,\theta_i)=|V|-\theta_i$, we have
		\begin{align}
			\frac{1-s}{2} \mathbb{E}_{V} \left[ (f(|V|,\theta_2)- g(|V|))^2\Big||V|<|V|_{1-s} \right] + \frac{s}{2} \mathbb{E}_{V} \left[ (f(|V|,\theta_1)- g(|V|))^2\Big||V|>|V|_{1-s} \right]\nonumber
			\\
			= \frac{1}{2}\int_{|v|<|V|_{1-s}} (f(|v|,\theta_2)- g(|v|))^2 \rho(v)dv + \frac{1}{2}\int_{|v|>|V|_{1-s}}(f(|v|,\theta_1)- g(|v|))^2\rho(v)dv
			\label{eq:two integrals}
		\end{align}
		where $\rho(v)$ is the probability density function of $V$. 
		
		To minimize \eqref{eq:two integrals} up to arbitrary precision, we discretize the integrals by the Euler's finite difference. For each $p$, we define $v_i=V_{i/p}$ for $1\leq i\leq p$, where $V_{i/p}$ is the $(i/p)$-th quantile of $V$, then as $p\to\infty$:
		$$\int_{|v|<|V|_{1-s}} (f(|v|,\theta_2)- g(|v|))^2 \rho(v)dv\leftarrow \frac{1}{p}\sum_{i:|v_i|<|V|_{1-s}}(f(|v_i|,\theta_2)-g(|v_i|))^2$$
		and similarly
		$$\int_{|v|>|V|_{1-s}} (f(|v|,\theta_1)- g(|v|))^2 \rho(v)dv\leftarrow \frac{1}{p}\sum_{i:|v_i|>|V|_{1-s}}(f(|v_i|,\theta_1)-g(|v_i|))^2$$
		
		\if
		To minimize \eqref{eq:two integrals} up to arbitrary precision, we discretize the integrals by the Euler's finite difference: as $\Delta v\to 0$ and $m\to\infty$,
		$$\int_{|v|<|V|_{1-s}} (f(|v|,\theta_2)- g(|v|))^2 \rho(v)dv\leftarrow \sum_{v_i}(f(|v_i|,\theta_2)-g(|v_i|))^2\rho(v_i)\cdot\Delta v$$
		where $v_i\in\{-|V|_{1-s},-|V|_{1-s}+\Delta v,...,|V|_{1-s}-2\Delta v,|V|_{1-s}-\Delta v,|V|_{1-s}\}$, and similarly
		$$\int_{|v|>|V|_{1-s}} (f(|v|,\theta_1)- g(|v|))^2 \rho(v)dv\leftarrow \sum_{v_i}(f(|v_i|,\theta_1)-g(|v_i|))^2\rho(v_i)\cdot\Delta v$$
		where $v_i\in\{-|V|_{1-s}-m\Delta v,-|V|_{1-s}-(m-1)\Delta v,...,-|V|_{1-s}\}\cup\{|V|_{1-s},|V|_{1-s}+\Delta v,...,|V|_{1-s}+m\Delta v\}$. Leveraging that $g$ is odd, i.e., $g(-v)=-g(v)$, we will minimize \jk{Why is $\rho(v_i)$ multiplying  the sum? Shouldn't $ \Delta v_i\rho(v_i) $ be multiplying each term?}\textcolor{blue}{ZB: $\Delta v$ is a constant so the following is an equivalent minimization problem.}
		\begin{align*}
			&\sum_{0<v_i<|V|_{1-s}}\rho(v_i)\cdot(v_i-\theta_2-g(v_i))^2+\sum_{-|V|_{1-s}<v_i<0}\rho(v_i)\cdot(-v_i-\theta_2+g(v_i))^2
			\\
			&+\sum_{v_i>|V|_{1-s}}\rho(v_i)\cdot(v_i-\theta_1-g(v_i))^2+\sum_{v_i<-|V|_{1-s}}\rho(v_i)\cdot(-v_i-\theta_1+g(v_i))^2
		\end{align*}
		\fi
		
		To this end, we will formulate an isotonic regression problem. Firstly, we denote
		\begin{align*}
			y_i=
			\begin{cases}
				v_i+\theta_1 & \text{if }v_i<-|V|_{1-s}
				\\
				v_i+\theta_2 & \text{if }-|V|_{1-s}<v_i<0
				\\
				v_i-\theta_2 & \text{if }0<v_i<|V|_{1-s}
				\\
				v_i-\theta_1 & \text{if }v_i>|V|_{1-s}
			\end{cases}
		\end{align*}
		where we have leveraged that $g$ is odd, i.e., $g(-v)=-g(v)$, and hence for $v<0$, we get $(f(|v|,\theta)-g(|v|))^2=(|v|-\theta-g(|v|))^2=(-v-\theta+g(v))^2=(v+\theta-g(v))^2$.
		
		Secondly, up to a constant factor $1/p$, we minimize over $g_i:=g(v_i)$ by
		$$\sum_{i} (y_i-g_i)^2 \text{ s.t. $g$ is non-decreasing}.$$
		
		
		This isotonic regression can be readily solved by the PAVA algorithm that iteratively averages the increasing subsequences of $\{y_i\}$. We note there are three vertical drops (where the increasing sub-sequences exist) at $-|V|_{1-s}, 0, |V|_{1-s}$ respectively. By the proof of \Cref{lem:1 flat}, the averaging will result in three flattend regions in the sequence $\{g_i\}$. Given that $g$ is odd, the height of the flattened region around $0$ must be 0, since $g(0)=0$; the height of the flattened region around $|V|_{1-s}$ must equal the negative of that around $-|V|_{1-s}$, which we denote as an unknown $h'$. Therefore, we have
		\begin{align*}
			g(|v|)=
			\begin{cases}
				0 & \text{ if } |v|<|V|_{1-s}-a, |v|<\theta_2
				\\
				|v|-\theta_2 & \text{ if } |v|<|V|_{1-s}-a, |v|>\theta_2
				\\
				h' & \text{ if } |v|>|V|_{1-s}-a, |v|<|V|_{1-s}+b
				\\
				0 & \text{ if } |v|>|V|_{1-s}+b, |v|<\theta_1
				\\
				|v|-\theta_1 & \text{ if } |v|>|V|_{1-s}+b, |v|>\theta_1
			\end{cases}
		\end{align*}
		and generally we have, 
		\begin{align*}
			g(v)=
			\begin{cases}
				\eta_\text{soft}(v;\theta_2) & \text{ if } |v|<|V|_{1-s}-a
				\\
				\text{sign}(v)h' & \text{ otherwise }
				\\
				\eta_\text{soft}(v;\theta_1) & \text{ if } |v|>|V|_{1-s}+b
			\end{cases}
		\end{align*}
		
		
		Finally, we check whether $g$ is 1-Lipschitz, odd, and non-decreasing. Indeed, it is 1-Lipschitz because soft-thesholding function and constant function are 1-Lipschitz; it is odd because soft-thresholding function and sign function are odd; it is obviously non-decreasing by construction.
		
		Note this result recovers \eqref{eq:2-level limiting} with a substitution of $V\to \pi+Z$, $\Theta\to\A$, and $h'\to\max\{h,0\}$ where $(a,b,h)$ can be determined in \Cref{lem:q1q2h}.

		\if
		If we do not satisfy the monotonicity constraint on $g$ imposed by $\mathcal{I}$, then 
		\begin{align*}
			g_\text{infeasible}(v)=
			\begin{cases}
				\text{sign}(v)\max\{f(|v|,\theta_2),0\}=\eta_\text{soft}(v;\theta_2) & \text{ if } |v|<|V|_{1-s}
				\\
				\text{sign}(v)\max\{f(|v|,\theta_1),0\}=\eta_\text{soft}(v;\theta_1) & \text{ if } |v|>|V|_{1-s}
			\end{cases}
		\end{align*}
		is a solution that optimizes each of the two integrals above separately (c.f., Example 1 in \cite{hu2019asymptotics} that solves the LASSO). However, it is generally infeasible (i.e., $g_\text{infeasible}\not\in\mathcal{I}$) because it is not monotonic, due to the vertical drop at $|v|=|V|_{1-s}$. See the second row of \Cref{fig:normalized PAVA} for a visualization. Therefore, we seek a non-decreasing projection of $g_\text{infeasible}$ which is the continuous version of the iterative averaging in the PAVA algorithm or the isotonic regression. 
		
		To flatten the vertical drop of $g_\text{infeasible}$ at $|v|=|V|_{1-s}$, we write the candidate solution
		\begin{align*}
			g_\text{feasible}(v)=
			\begin{cases}
				\eta_\text{soft}(v;\theta_2) & \text{ if } |v|<|V|_{1-s}-a
				\\
				\text{sign}(v)\max\{h,0\} & \text{ otherwise }
				\\
				\eta_\text{soft}(v;\theta_1) & \text{ if } |v|>|V|_{1-s}+b
			\end{cases}
		\end{align*}
		for some unknown constants $a,b,h$. By continuity, we have $|V|_{1-s}-a-\theta_2=h$ if $|V|_{1-s}-a>\theta_2$, and similarly $|V|_{1-s}+b-\theta_1=h$ if $|V|_{1-s}+b>\theta_1$. In fact, $g_\text{feasible}$ satisfies the complementary slackness condition: when $g_\text{feasible}$ is strictly increasing at $v$, i.e., the constraint of $\mathcal{I}$ is binding, $g_\text{feasible}(v)=g_\text{infeasible}(v)$; otherwise, $g_\text{feasible}(v)$ remains constant.
		
		We note that it is equivalent to rewrite
		\begin{align*}
			g_\text{feasible}(v)=
			\begin{cases}
				\eta_\text{soft}(v;\theta_2) & \text{ if } |v|<\theta_2+h
				\\
				\text{sign}(v)\max\{h,0\} & \text{ otherwise }
				\\
				\eta_\text{soft}(v;\theta_1) & \text{ if } |v|>\theta_1+h
			\end{cases}
		\end{align*}
		where $h$ can be determined in the same way as \Cref{lem:q1q2h} with a substitution of $V\to \Pi+\tau Z$ and $\Theta\to\A$. Hence, $g_\text{feasible}$ has at most one non-zero shared magnitude: if $h>0$, it has exactly one; else, it has none.
		\fi

		\if
		Without loss of generality, we consider $v\geq 0$ so that $\text{sign}(v)=1$ and $|v|=v$. 
		Note then the condition on $g(|v|)$ can be translated to $v$ (denoting $|V|_{1-s}$ as the quantile): \jk{Where did $\text{sign}(v)$ go below?}\textcolor{blue}{ZB: assumed $v\geq 0$.}
		\begin{align*}
			g_\text{infeasible}(v)=
			\begin{cases}
				\max\{v-\theta_2,0\} & \text{ if } v<|V|_{1-s}
				\\
				\max\{v-\theta_1,0\} & \text{ if } v>|V|_{1-s}
			\end{cases}
		\end{align*}
		{\color{blue}
			For the simplicity of presentation, we analyze the case when $|V|_{1-s}>\theta_1>\theta_2$ (note that the cases $\theta_1>|V|_{1-s}>\theta_2$ and $\theta_1>\theta_2>|V|_{1-s}$ can be analyzed in a similar manner), where
			\begin{align*}
				g_\text{infeasible}(v)=
				\begin{cases}
					0 & \text{ if } v<\theta_2
					\\
					v-\theta_2 & \text{ if } \theta_2<v<|V|_{1-s}
					\\
					v-\theta_1 & \text{ if } v>|V|_{1-s}
				\end{cases}
			\end{align*}

			To flatten the vertical drop of $g_\text{infeasible}$ at $v=|V|_{1-s}$, we write the candidate solution
			\begin{align*}
				g_\text{feasible}(v)=
				\begin{cases}
					0 & \text{ if } v<\theta_2
					\\
					v-\theta_2 & \text{ if } \theta_2<v<|V|_{1-s}-a
					\\
					c & \text{ otherwise }
					\\
					v-\theta_1 & \text{ if } v>|V|_{1-s}+b
				\end{cases}
			\end{align*}
			for some constants $a,b,c$. In fact, $g_\text{feasible}$ satisfies the complementary slackness condition: when $g_\text{feasible}$ is strictly increasing at $v$, i.e., the constraint of $\mathcal{I}$ is binding, $g_\text{feasible}(v)=g_\text{infeasible}(v)$; otherwise, $g_\text{feasible}(v)$ remains constant. 
			
			Substituting into \eqref{eq:two integrals}, we have
			\begin{align*}
				&\frac{1-s}{2\PP(v<|V|_{1-s})}\int_0^{|V|_{1-s}}\left(v-\theta_2-g_\text{feasible}
				(v)\right)^2 \rho(v)dv
				\\
				+&\frac{s}{2\PP(v>|V|_{1-s})}\int_{|V|_{1-s}}^\infty\left(v-\theta_1-g_\text{feasible}
				(v)\right)^2 \rho(v)dv
				\\
				=&\frac{1-s}{2\PP(v<|V|_{1-s})}\int_0^{\theta_2}\left(v-\theta_2\right)^2 \rho(v)dv+\frac{1-s}{2\PP(v<|V|_{1-s})}\int_{|V|_{1-s}-a}^{|V|_{1-s}}\left(v-\theta_2-c\right)^2 \rho(v)dv
				\\
				+&\frac{s}{2\PP(v>|V|_{1-s})}\int_{|V|_{1-s}}^{|V|_{1-s}+b}\left(v-\theta_1-c\right)^2 \rho(v)dv
			\end{align*}
			
			To minimize this over $a,b,c$, we denote that,
			\begin{align*}
				Q(a,b,c):=&\frac{1-s}{2\PP(v<|V|_{1-s})}\int_{|V|_{1-s}-a}^{|V|_{1-s}}\left(v-\theta_2-c\right)^2 \rho(v)dv
				\\
				&+\frac{s}{2\PP(v>|V|_{1-s})}\int_{|V|_{1-s}}^{|V|_{1-s}+b}\left(v-\theta_1-c\right)^2 \rho(v)dv
			\end{align*}
			and take the partial derivatives $\frac{\partial Q}{\partial a}=0$ and $\frac{\partial Q}{\partial b}=0$. We obtain the optimal $a,b$ conditioning on $c$ as
			$$a=|V|_{1-s}-c-\theta_2, b=c+\theta_1-|V|_{1-s}$$
			which enforces the continuity of $g_\text{feasible}$ at $|V|_{1-s}-a$ and at $|V|_{1-s}+b$. Finally, the optimal $c$ can be derived by solving $\frac{\partial Q}{\partial c}\big|_{a=|V|_{1-s}-c-\theta_2, b=c+\theta_1-|V|_{1-s}}=0$.
		}
		
		It is not hard to see \jk{I don't immediately see this. Could you add more details?} that the non-decreasing projection is
		\begin{align*}
			g_\text{feasible}(v)=
			\begin{cases}
				\max\{v-\theta_2,0\} & \text{ if } v<|V|_{1-s}-\frac{\theta_1-\theta_2}{2}
				\\
				\max\{|V|_{1-s}-\theta_1/2-\theta_2/2,0\} & \text{ otherwise}
				\\
				\max\{v-\theta_1,0\} & \text{ if } v>|V|_{1-s}+\frac{\theta_1-\theta_2}{2}
			\end{cases}
		\end{align*}
		which has at most one non-zero shared magnitude that is $|V|_{1-s}-\theta_1/2-\theta_2/2$.
		\fi

	\end{proof}

	\subsection{Proof of \Cref{thm:complete trade-off}}
	\begin{proof}
		Our proof is based on the homotopy result in \cite{wang2020complete} that is also used to derive the achievable $(\TPP,\FDP)$ of the general SLOPE (c.f., Proposition D.3 in \cite{bu2021characterizing}). We briefly describe the idea of homotopy. Consider a continuous transformation $ f $ moving from our trade-off curve $q_\text{2-level}(\TPP)$ to the horizontal line $\FDP=1-\epsilon$. During the movement, $f$ sweeps out a region of $(\TPP,\FDP)$ during the transformation, so that each $(\TPP,\FDP)$ is realizable by some $f$. \begin{lemma}[Lemma 3.7 in \cite{wang2020complete}]\label{lem:wanghua}
			If a continuous curve is parameterized by $ f:[0,1] \times[0,1] \rightarrow \mathbb{R}^{2} $ and if the four curves
			\begin{itemize}
				\item$\mathcal{C}_{1}=\{f(u, 0): 0 \leq u \leq 1\}$,
				\item$\mathcal{C}_{2}=\{f(u, 1): 0 \leq u\leq 1\}$,
				\item$\mathcal{C}_{3}=\{f(0, s): 0 \leq s \leq 1\}$,
				\item$\mathcal{C}_{4}=\{f(1, s): 0 \leq s \leq 1\}$,  
			\end{itemize}
			join together as a simple closed curve, $\mathcal{C}:=\mathcal{C}_{1} \cup \mathcal{C}_{2} \cup \mathcal{C}_{3} \cup \mathcal{C}_{4}, $ then $ \mathcal{C} $ encloses an interior area $ \mathcal{D}, $ and $ \forall(x, y) \in \mathcal{D}, \exists(u, s) \in [0,1] \times[0,1] $ such that $ f(u, s)=(x, y) . $ In other words, every point inside the region $\mathcal{D}$ enclosed by curve $ \mathcal{C}$ is realizable by some $ f(u, s)$.
		\end{lemma}
		
		We will prepare some notations before leveraging \Cref{lem:wanghua}. For $\sigma\neq \infty$, we denote $ (\TPP,\FDP) $ as functions of $ (\Pi, \Lambda, \sigma) $ and parameterize the 2-level penalty distribution $\Lambda_*(u)$ and the three-point prior distribution $\Pi_*(u)$ such that 
		\begin{align*}
			\TPP(\Pi_*(u),\Lambda_*(u),\sigma\neq\infty)&=u,
			\\
			\FDP(\Pi_*(u),\Lambda_*(u),\sigma\neq\infty)&=q_\text{2-level}(u).
		\end{align*}
		Such parameterization is valid due to the tightness of $q_\text{2-level}$ by \Cref{thm:trade-off}.
		
		Now we can define
		$$f(u,s)=\left(u,\FDP(\Pi_*(u),\Lambda_*(u),\text{tan}\frac{\pi s}{2})\right),$$
		and leverage \Cref{lem:wanghua} to analyze each of the four curves.
		\begin{itemize}
			\item $\mathcal{C}_1$ is $\FDP=q_\text{2-level}(\TPP)$, i.e., the 2-level SLOPE trade-off;
			\item $\mathcal{C}_2$ is $\FDP=1-\epsilon$, since the noise $\sigma=\infty$ when $s=1$, and no signal of $\Pi$ is extractable;
			\item $\mathcal{C}_3$ is $\TPP=0$;
			\item $\mathcal{C}_4$ is a single-point $(\TPP,\FDP)=(1,1-\epsilon)$, or a curve with no length.
		\end{itemize}
		
		Given that $\{\mathcal{C}_1,\mathcal{C}_2,\mathcal{C}_3,\mathcal{C}_4\}$ forms a closed curve, we have shown that $\mathcal{D}_{\epsilon,\delta}$ is achievable by omitting the single-point $\mathcal{C}_4$.
	\end{proof}
	
	\subsection{Proof of \Cref{lem:expand E}}
	\label{sec:proof56}
	\begin{proof}
		We first use the definition of the 2-level SLOPE limiting scalar function given in \eqref{eq:2-level limiting}, which is restated below for convenience, to expand $\E\{(\eta(t+Z;\Aeff(t+Z))-t)^2\}$. Recall,
		\begin{align*}
			\begin{split}
				\textbf{2-level SLOPE}
				\\
				\textbf{limiting scalar function}
			\end{split}
			:\, \, \eta_{\pi+Z,\A}(x)=
			\begin{cases}
				\eta_\text{soft}(x;\alpha_1) &\text{ if } \alpha_1+h < |x|,
				\\
				\text{sign}(x) \max\{h,0\} &\text{ if } \alpha_2+h<|x|<\alpha_1+h, \\
				\eta_\text{soft}(x;\alpha_2) &\text{ if } |x|<\alpha_2+h,
			\end{cases}
		\end{align*}
		Using the above, we have
		\begin{align}
			\E(\eta(t+Z;\Aeff)-t)^2&=\int\left(\eta(t+z;\Aeff(t+z))-t\right)^2\phi(z)dz \nonumber
			\\
			&=\int_{|t+z|>\alpha_1+h}\left(\eta_\text{soft}(t+z;\alpha_1)-t\right)^2\phi(z)dz
			\label{eq:first term of lengthy}
			\\
			&\quad +\int_{\alpha_2+h<|t+z|<\alpha_1+h}(h\cdot\text{sign}(t+z)-t)^2\phi(z)dz
			\label{eq:second term of lengthy}
			\\
			&\quad +\int_{\alpha_2<|t+z|<\alpha_2+h}(\eta_\text{soft}(t+z;\alpha_2)-t)^2\phi(z)dz
			\label{eq:third term of lengthy}
			\\
			&\quad +\int_{|t+z|<\alpha_2}(\eta_\text{soft}(t+z;\alpha_2)-t)^2\phi(z)dz
			\label{eq:fourth term of lengthy}
		\end{align}
		
		Next, we will expand each of \eqref{eq:first term of lengthy}, \eqref{eq:second term of lengthy}, \eqref{eq:third term of lengthy} and \eqref{eq:fourth term of lengthy} using the fact that the soft-thresholding function is defined as $\eta_\text{soft}(\boldsymbol{v};\theta)=\text{sign}(\boldsymbol{v})(|\boldsymbol{v}|-\theta)_+$ and some facts about truncated Gaussian expectations given in \Cref{fact:truncated normal}.

		To expand the first term notice that if $t+z>\alpha_1+h$ then $\eta_\text{soft}(t+z;\alpha_1) = t+z-\alpha_1$ and if $t+z<-\alpha_1-h$ then $\eta_\text{soft}(t+z;\alpha_1) = t+z+\alpha_1$.
		Hence, we have 
		\begin{align*}
			\text{\eqref{eq:first term of lengthy}} &= \int_{t+z>\alpha_1+h}(t+z-\alpha_1-t)^2\phi(z)dz+\int_{t+z<-\alpha_1-h}(t+z+\alpha_1-t)^2\phi(z)dz\\
			&=\int_{z-\alpha_1>h-t}(z-\alpha_1)^2\phi(z)dz+\int_{z+\alpha_1<-h-t}(z+\alpha_1)^2\phi(z)dz\\
			&=\E(N_1^2|N_1>h-t)\PP(N_1>h-t)+\E(N_2^2|N_2<-h-t)\PP(N_2<-h-t)\\
			&=(1-\Phi(\alpha_1+h-t))\E(N_1^2|N_1>h-t)+\Phi(-\alpha_1-h-t) \E(N_2^2|N_2<-h-t),
		\end{align*}
		where $N_1\sim \mathcal{N}(-\alpha_1,1)$ and $N_2\sim \mathcal{N}(\alpha_1,1)$. The expectation of the truncated normal distribution is given in \Cref{fact:truncated normal}, which leads to 
		\begin{align*}
			(1-\Phi(\alpha_1+h-t))\mathbb{E} (N_1^2\mid N_1>h-t)
			&=(\alpha_1^2+1)(1-\Phi(\alpha_1+h-t)) \\
			&\quad +(-\alpha_1+h-t)\phi(\alpha_1+h-t),
		\end{align*}
		and
		\begin{align*}
			\Phi(-\alpha_1-h-t) \mathbb{E} (N_2^2\mid N_2<-h-t)
			&=(\alpha_1^2+1)\Phi(-\alpha_1-h-t)\\
			&\quad -(\alpha_1-h-t)\phi(-\alpha_1-h-t)\,.
		\end{align*}

		The second term \eqref{eq:second term of lengthy} can be easily expanded as
		\begin{align*}
			\text{\eqref{eq:second term of lengthy}} &= \int_{\alpha_2+h<t+z<\alpha_1+h}(h-t)^2\phi(z)dz+\int_{-\alpha_1-h<t+z<-\alpha_2-h}(-h-t)^2\phi(z)dz\\
			&=(h-t)^2\left[\Phi(\alpha_1+h-t)-\Phi(\alpha_2+h-t)\right] \\
			&\quad +(h+t)^2\left[\Phi(-\alpha_2-h-t)-\Phi(-\alpha_1-h-t)\right].
		\end{align*}

		Next, to expand the third term \eqref{eq:third term of lengthy}, we notice that  if $-\alpha_2-h<z+t<-\alpha_2$, then $\eta_\text{soft}(t+z;\alpha_2) = t+z+\alpha_2$ and if $\alpha_2<z+t<\alpha_2+h$, then $\eta_\text{soft}(t+z;\alpha_2) = t+z-\alpha_2$. Using this we find
		\begin{align*}
			\text{\eqref{eq:third term of lengthy}} &= \int_{-\alpha_2-h<z+t<-\alpha_2}(t+z+\alpha_2-t)^2\phi(z)dz+\int_{\alpha_2<z+t<\alpha_2+h}(t+z-\alpha_2-t)^2\phi(z)dz\\
			&=\int_{-h-t<z+\alpha_2<-t}(z+\alpha_2)^2\phi(z)dz+\int_{-t<z-\alpha_2<h-t}(z-\alpha_2)^2\phi(z)dz\\
			&=\left[\Phi(-\alpha_2-t)-\Phi(-\alpha_2-h-t)\right]\mathbb{E}(N_3^2|-h-t<N_3<-t)\\
			&\quad +\left[\Phi(\alpha_2+h-t)-\Phi(\alpha_2-t)\right]\mathbb{E}(N_4^2|-t<N_4<h-t),
		\end{align*}
		where $N_3\sim \mathcal{N}(\alpha_2,1)$ and $N_4\sim \mathcal{N}(-\alpha_2,1)$. Again, following \Cref{fact:truncated normal}, we obtain 
		\begin{align*}
			&\left[\Phi(-\alpha_2-t)-\Phi(-\alpha_2-h-t)\right]\mathbb{E}(N_3^2|-h-t<N_3<-t)\\
			&=(1+\alpha_2^2)\left[\Phi(-\alpha_2-t)-\Phi(-\alpha_2-h-t)\right]\\
			&\quad +(\alpha_2-h-t)\phi(-\alpha_2-h-t)-(\alpha_2-t)\phi(-\alpha_2-t),
		\end{align*}
		and \begin{align*}
			&\left[\Phi(\alpha_2+h-t)-\Phi(\alpha_2-t)\right]\mathbb{E}(N_4^2|-t<N_4<h-t)
			\\
			&=(1+\alpha_2^2)\left[\Phi(\alpha_2+h-t)-\Phi(\alpha_2-t)\right]\\
			&\quad -(\alpha_2+t) \phi(\alpha_2-t)+(\alpha_2-h+t)\phi(\alpha_2+h-t).
		\end{align*}
		
		To expand the fourth term \eqref{eq:fourth term of lengthy}, notice that $\eta(t+z;\alpha_2)=0$ when $|t+z|<\alpha_2$. Therefore,
		\begin{align*}
			\text{\eqref{eq:fourth term of lengthy}} =\int_{-\alpha_2<t+z<\alpha_2}t^2\phi(z)dz
			=t^2\int_{-\alpha_2-t<z<\alpha_2-t}\phi(z)dz=t^2\left[\Phi(\alpha_2-t)-\Phi(-\alpha_2-t)\right].
		\end{align*}
		
	\end{proof}
	
	\section{Auxiliary results}
	\subsection{Facts about the expectation of a truncated normal distribution}
	\begin{fact}\label{fact:truncated normal}
		Consider $X\sim \mathcal{N}(\mu,1)$.
		\begin{enumerate}
			\item We have $\mathbb{E} (X\mid X>a)=\mu + \frac{\phi (\alpha )}{1-\Phi(\alpha)}$ and $\text{Var}(X\mid X>a)=\E(X^2|X>a)-\E(X|X>a)^2=1+\alpha \frac{\phi (\alpha )}{1-\Phi(\alpha)}-(\frac{\phi (\alpha )}{1-\Phi(\alpha)})^{2}$ where $\alpha =a-\mu $ 
			\begin{align*}
				\mathbb{E} (X^2\mid X>a)=\mu^2+1+ \frac{\alpha\phi(\alpha)}{1-\Phi(\alpha)}+2\mu \frac{\phi(\alpha)}{1-\Phi(\alpha)}\,.
			\end{align*}
			\item Similarly, by $\mathbb{E} (X\mid X<b)=\mu - {\frac {\phi (\beta)}{\Phi (\beta)}}
			$ and $\text{Var} (X\mid X<b)=1-\alpha {\frac {\phi (\beta)}{\Phi (\beta)}}-\left({\frac {\phi (\beta)}{\Phi (\beta)}}\right)^{2}$ where $\beta=b-\mu$, we have 
			$$\mathbb{E} (X^2\mid X<b)=\mu^2+1- \frac{\beta\phi(\beta)}{\Phi(\beta)}-2\mu \frac{\phi(\beta)}{\Phi(\beta)}\,.$$
			\item By the result that $\mathbb{E} (X\mid a<X<b)=\mu -{\frac {\phi (\beta )-\phi (\alpha )}{\Phi (\beta )-\Phi (\alpha )}}$
			and
			$\text{Var} (X\mid a<X<b)=1-{\frac {\beta \phi (\beta )-\alpha \phi (\alpha )}{\Phi (\beta )-\Phi (\alpha )}}-\left({\frac {\phi (\beta )-\phi (\alpha )}{\Phi (\beta )-\Phi (\alpha )}}\right)^{2}$, we can obtain that 
			\begin{align*}
				\mathbb{E}(X^2|a<X<b)=\mu^2+1-\frac{\phi(\beta)}{\Phi(\beta)-\Phi(\alpha)}(\beta+2\mu)+\frac{\phi(\alpha)}{\Phi(\beta)-\Phi(\alpha)}(\alpha+2\mu)\,.
			\end{align*}
		\end{enumerate}
	\end{fact}

	\subsection{More details about $q_\text{LASSO}$, $q_\text{2-level}$ and $q_\text{SLOPE}$}
	\label{app:q details}
	While $q_\text{LASSO}$, $q_\text{2-level}$ and $q_\text{SLOPE}$ can be unified by \eqref{eq:fdp is function of tpp}, they are different because the maximum zero-thresholds are different. 
	
	To be concrete, for the LASSO, \cite{su2017false} shows that
	\begin{enumerate}
		\item The $\TPP$ domain of LASSO is $[0,1-\frac{(1-\delta)(\epsilon-\epsilon^*)}{\epsilon(1-\epsilon^*)})$ where $\epsilon^*(\delta)$ is determined by $\delta$ via an implicit parametric form in their Equation (C.5)
		$$\delta = \frac{2\phi(t)}{2\phi(t) + t(2\phi(t)-1)}, \epsilon^* = \frac{2\phi(t)-2t\phi(-t)}{2\phi(t) + t(2\phi(t)-1)}.$$
		\item The minimum $\FDP$ is given by
		$$q_{\text{Lasso}}(\TPP; \delta, \epsilon) = \frac{2(1-\epsilon)\Phi(-t^*(\TPP))}{2(1-\epsilon)\Phi(-t^*(\TPP)) + \epsilon\cdot\TPP}$$
		where $t^*(u)$ is the largest positive root of the following, in terms of $x$,
		$$\frac{2(1-\epsilon)[(1+x^2)\Phi(-x)-x\phi(x)] + \epsilon(1+x^2)-\delta}{\epsilon[(1+x^2)(1-2\Phi(-x))+2x\phi(x)]} = \frac{1-u}{1-2\Phi(-x)}.$$
	\end{enumerate}
	For the 2-level SLOPE, we devote \Cref{sec:tpp fdp all priors} to derive the maximum zero-threshold $\alpha^*_\text{2-level}$ in \eqref{eq:alpha star explicit}. The searching process of maximum zero-threshold is like this:
	\begin{enumerate}
		\item For each zero-threshold $\alpha_2$, search over $(\alpha_1,s)$;
		\item For each penalty $(\alpha_2,\alpha_1,s)$, search over priors $(t_1,t_2)$;
		\item If any $(\alpha_2,\alpha_1,s,t_1,t_2)$ gives $F[\pi_{\min}]\leq \delta$, then $\alpha_2$ is a feasible zero-threshold.
	\end{enumerate}
	
	For the general SLOPE, \cite{bu2021characterizing} gives a lower bound (denoted as $q_\star$ in the literature) that can only be found numerically. The searching process is similar to 2-level SLOPE, but must be done in the finite dimensions, using tools from the calculus of variations and quadratic optimization programming.
	\begin{enumerate}
		\item For each zero-threshold $\alpha$, search over priors $(t_1,t_2)$;
		\item For each $(\alpha,t_1,t_2)$, construct a quadratic programming where the objective is $F[\pi_{\min}]$ and minimize it.
		\item If any $(\alpha,t_1,t_2)$ gives $F[\pi_{\min}]\leq \delta$, then $\alpha$ is a feasible zero-threshold.
	\end{enumerate}

	\subsection{Connection between finite-dimensional and asymptotic regime for TPP and FDP}
	\label{app:tpp fdp details}
	In \cite[Theorem 2.1]{su2017false} and \cite[Theorem 1]{bu2021characterizing}, the TPP-FDP trade-off is formulated in finite dimensions ($\p\neq\infty$). As a consequence, the statements in these works require an arbitrarily small constant $c>0$ such that
	$$\textup{FDP}\geq q(\textup{TPP})-c,$$
	holds with probability tending to one as $\p\to\infty$. That is, for all $c > 0$,
	$$\plim_\p \PP(\textup{FDP}\geq q(\textup{TPP})-c)=1;$$
	therefore,
	$$\lim_{c\to 0}\plim_\p \PP(\textup{FDP}\geq q(\textup{TPP})-c)=1.$$
	
	We note these statements can be translated to the asymptotic regime ($\p\to\infty$): with an exchange of $\plim$ and $\PP$, which is allowed by the uniform convergence proven in \cite[Lemma A.1]{su2017false} for LASSO and \cite[Lemma D.1]{bu2021characterizing} for general SLOPE:
	$$\PP(\FDP\geq q(\TPP))=1.$$
	
	To be rigorous, we also provide a finite-dimension version of \Cref{thm:trade-off} that resembles the general SLOPE trade-off in \cite[Theorem 1]{bu2021characterizing}:
	Define 
	\begin{align}
		\operatorname{TPP}_\xi(\boldsymbol{\beta}, \blam) & =\frac{\left|\left\{j:|\widehat{\beta}(\blam)_j|>\xi, \, \beta_j \neq 0\right\}\right|}{\left|\left\{j: \beta_j \neq 0\right\}\right|}\, , \quad 
		\operatorname{FDP}_\xi(\boldsymbol{\beta}, \blam) =\frac{\left|\left\{j:|\widehat{\beta}(\blam)_j|> \xi, \, \beta_j=0\right\}\right|}{\left|\left\{j:|\widehat{\beta}(\blam)_j|>\xi\right\}\right|} \,.
	\end{align}
	Here $\operatorname{TPP}_\xi$ and $\operatorname{FDP}_\xi$ reduce to the standard TPP and FDP when $\xi\to 0$. These thresholded TPP and FDP values are only defined for technical reasons; namely, because a SLOPE solution may have many entries that are close to but not exactly equal to zero.
	
	Under the working assumptions (A1)$\sim$(A5) in \Cref{sec:AMP}, the following inequality holds with probability tending to one as $n,p\to\infty$:
	$$
	\operatorname{FDP}_{\xi}(\boldsymbol{\beta}, \boldsymbol{\lambda}) \geq q_\text{2-level}\left(\operatorname{TPP}_{\xi}(\boldsymbol{\beta}, \boldsymbol{\lambda}) ; \delta, \epsilon\right)-c_{\xi},
	$$
	for some positive constant $c_{\xi}\to 0$ as $\xi \rightarrow 0$. Here, $q_\text{2-level}$ is given in \Cref{thm:trade-off}.

\end{document}